\documentclass[10pt]{amsart}
\usepackage{graphicx,amssymb,amsfonts,amsmath,amsthm,newlfont}
\usepackage{epsfig,url}
\usepackage{color}
\usepackage{enumitem}

\usepackage{axodraw4j}

\usepackage[all,2cell]{xy} \UseAllTwocells \SilentMatrices

\newtheorem{thm}{Theorem}[section]
\newtheorem{cor}[thm]{Corollary}
\newtheorem{lem}[thm]{Lemma}
\newtheorem{prop}[thm]{Proposition}
\theoremstyle{definition}
\newtheorem{defn}[thm]{Definition}
\newtheorem{conj}{Conjecture} 

\newtheorem{ex}[thm]{Examples}
\newtheorem{example}[thm]{Example}

\theoremstyle{remark}
\newtheorem{rem}[thm]{Remark}
\numberwithin{equation}{section}

\newcommand{\Z}{\mathbb Z}
\newcommand{\C}{\mathbb C}

\newcommand{\R}{\mathbb R}
\newcommand{\N}{\mathbb N}
\newcommand{\Pro}{\mathbb P}

\newcommand{\gr}{\mathrm{gr}}

\newcommand{\FF}{\mathcal{F}}

\newcommand{\motic}{motic\,\,}
\newcommand{\Motic}{Motic\,\,}
\newcommand{\Real}{\mathrm{Re}}

\newcommand{\MT}{\mathcal{MT}}

\newcommand{\sd}{\mathrm{sd}}

\newcommand{\fP}{P}
\newcommand{\fD}{D}

\newcommand{\zetam}{\zeta^{ \mathfrak{m}}}

\newcommand{\Q}{\mathbb Q}
\newcommand{\Li}{\mathrm{Li}}

\newcommand{\s}{\mathbf{s}}

\newcommand{\To}{\longrightarrow}

\newcommand{\A}{\mathbb{A}}

\newcommand{\mom}{\mathrm{mom}}

\newcommand{\dR}{\mathfrak{dr}}
\newcommand{\gen}{\mathrm{gen}}

\newcommand{\Or}{\mathcal{O}}

\newcommand{\V}{\mathcal{V}}
\newcommand{\IR}{\mathrm{IR}}
\newcommand{\UV}{\mathrm{UV}}

\newcommand{\mot}{\mathrm{mot}}

\newcommand{\mm}{\mathfrak{m} }

\newcommand{\HH}{\mathcal{H} }
\newcommand{\Lef}{\mathbb{L} }

\newcommand{\GG}{\mathbb{G} }

\newcommand{\id}{\mathrm{id} }

\newcommand{\HF}{\mathcal{FP} }

\newcommand{\ext}{\mathrm{ext}\,}
\newcommand{\q}{/\!\!/}

\newcommand{\per}{\mathrm{per}}

\newcommand{\uu}{\mathfrak{u}}

\newcommand{\Spec}{\mathrm{Spec} \,}

\newcommand{\Pe}{\mathcal{P}}

\begin{document}
\author{Francis Brown}
\begin{title}[Feynman amplitudes and cosmic Galois group]{Feynman amplitudes, coaction principle,  and cosmic Galois group}\end{title}
\maketitle
\begin{abstract} The first part of a set of  notes  based on lectures  given at the IHES in May 2015 on Feynman  amplitudes and motivic periods.  
\end{abstract}

\subsection{Some motivation for physicists} Scattering amplitudes are ubiquitous in high energy physics and have been intensively studied from  at least three angles:
\begin{enumerate}
\item in phenomenology, where amplitudes in quantum field theory are obtained as a sum of Feynman integrals associated to graphs  which represent interactions between fundamental particles. This presents a huge computational challenge with important applications to collider experiments.
\item in superstring perturbation theory, where amplitudes are expressed as integrals over moduli spaces of curves with marked points.
\item  in various modern  approaches, most notably in the planar limit of $N=4$ SYM,  which avoid the use of Feynman graphs altogether and seek to construct the amplitude directly, either  via the bootstrap method, or  via geometric approaches such as on-shell diagrams or the amplituhedron. 
\end{enumerate}
The goal of these notes is to study a new kind of structure which is potentially satisfied by amplitudes  in all three situations.  
To motivate it, consider first the case of  the dilogarithm function, defined for $|z| <1$ by the sum
$$\Li_2(z) = \sum_{n\geq 1} {z^n \over n^2}  \ . $$
 It is an iterated integral over the projective line minus three points, and is the universal function describing amplitudes at one loop.  Chen's general theory of iterated integrals  naturally associates to it a coproduct $\Delta^{\mathrm{it}}$ satisfying
$$\Delta^{\mathrm{it}} \,  \Li_2(z) = 1 \otimes \Li_2(z) + \Li_1(z) \otimes \log(z) + \Li_2(z) \otimes 1\ , $$
where $\Li_1(z) = - \log(1-z)$. 
Suitably interpreted, this coproduct  encodes both the differential equation $ {\partial \over \partial z} \Li_2(z) =  \Li_1(z)\,  d\log(z)$, and also the action of monodromy
$\Li_2(z) \mapsto \Li_2(z) +  2  \pi i   \log(z)$ as $z$ winds around the point $z=1$ in the positive direction.  It is well-known that Feynman integrals and amplitudes of different orders can be  related both with respect to differentiation, and  also with respect to branch cuts, and so it comes as no surprise that the  coproduct $\Delta^{\mathrm{it}}$ has found many uses in high-energy physics via the so-called method of symbols. 

Now consider  the much more subtle situation when  $z=1$. Then $\Li_2(1)= \zeta(2)$ is simply a number and all the  structure described above seems to be lost.  It can, however, be retrieved by replacing $\zeta(2)$ with a `motivic period' $\zetam(2)$, which as a first approximation, can be simply  thought of as   a matrix of numbers (in this case, a two by two matrix).  It now satisfies a \emph{coaction}, rather than a coproduct:
$$\Delta \zetam(2) = \zetam(2) \otimes 1\ . $$
The dilogarithm can also be promoted a motivic version  $\Li_2^{\mm}(z)$ in a similar manner, and  has  a   coaction:
$$\Delta \Li^{\mm}_2(z) =  \Li_2^{\mm}(z)  \otimes 1  + \Li_1^{\mm}(z) \otimes  \log^{\uu}(z) + 1 \otimes  \Li_2^{\uu}(z)$$
 which is valid both for $z$ viewed as a  function, and also for any algebraic values of $z$, including $z=1$ (in which case, $\Li_2^{\mm}(1) = \zetam(2)$ and  $\log^{\uu}(1)= \Li_2^{\uu}(1)= 0$).  The quantities on the right-hand side of the tensor product are of a different nature from those on the left, and could be called  \emph{unipotent de Rham} periods.  This coaction is a much deeper structure than the coproduct  $\Delta^{\mathrm{it}}$.
 Motivic periods have a natural homomorphism $\per$ to the complex numbers: for example,  $\per\, (\Li_2^{\mm}(z)) = \Li_2(z)$.

If we imagine that Feynman integrals  and scattering amplitudes more generally have canonical `motivic' versions, as one certainly expects,  then they inherit a coaction, and it is natural to ask how this coaction relates to the  structure of amplitudes. Indeed, any of the three situations described above should generate a space $H$ of motivic periods, and  a corresponding  algebra $A$ of de Rham periods.  A general `coaction principle'   would be the equation 
\begin{equation} \label{introcoact}   \Delta H \subset H \otimes A \ .
\end{equation} 
In other words, the class of amplitudes is stable under the coaction. An equivalent way to phrase this is in terms of group theory.  Indeed, $A$ is naturally a Hopf algebra, and defines  a group $C$ whose points are the homomorphisms from $A$ to any commutative ring. The equation  $(\ref{introcoact})$ is equivalent
to a group action
\begin{equation}     C \times H \To H\ . 
\end{equation}
In other words, the space of amplitudes in the theory are stable under the action of a group, which could be called a  `cosmic' Galois group, to borrow a phrase from  Cartier \cite{Folle}.

What possible evidence is there for such a structure? Taking each of the three situations in turn, we find the following:

\begin{enumerate}

\item   In quantum field theory, Panzer and Schnetz computed every known amplitude in massless $\phi^4$ theory, and, assuming the `period conjecture' replaced them with their motivic versions \cite{BrMTZ}, generating an algebra $H$.  Remarkably, they found that the coaction principle $(\ref{introcoact})$ holds in every case. Evidence in  \cite{CutsCoproducts, CutsCoproducts2}  suggests that Feynman amplitudes of small graphs with non-trivial masses and momenta also satisfy a similar property. 

\item In string perturbation theory, Stieberger and Schlotterer  \cite{SS} replaced the multiple zeta values in the expansion of  the genus zero open string with their motivic versions \cite{BrMTZ}.  They found that the coaction  gives rise to a spectacular factorisation of the amplitude, which is invisible  on the level of numbers. Similarly, the closed genus zero string can be expressed in terms of single-valued multiple zeta values \cite{ClosedString}, whose motivic versions are also known to be closed under the coaction $(\ref{introcoact})$.  
\item Several features of the hexagon bootstrap \cite{Bootstrap} are expressible in terms of  a coaction principle $(\ref{introcoact})$. 
Here, $H$ is a subspace of the space of motivic iterated integrals on the moduli space  of curves $\mathcal{M}_{0,6}$ of genus $0$ with six marked points \cite{NotesMot} \S10.6.2, and   equation $(\ref{introcoact})$ is equivalent to  `last $n$ entries' constraints on the iterated integrals.    \end{enumerate}

In all these settings, we believe that a version of the coaction principle $(\ref{introcoact})$ holds, after possibly enlarging the space of integrals or amplitudes under consideration. 
It is also important to note that in the first setting, the coaction principle holds graph by graph, whereas in the third setting, it   operates on the entire amplitude, i.e., the sum of all graphs (these  two  statements are by no means equivalent).

 In this paper, we concentrate only on the first setting. Our first goal, then,  is to define canonical motivic Feynman integrals for a large class of graphs in perturbative quantum field theory in an even number of space-time dimensions.  We then develop tools to prove that, after enlarging the  space of motivic integrals under consideration slightly, the coaction principle $(\ref{introcoact})$ does indeed hold. This is just the first step in a programme to study amplitudes via the representation theory of groups.

\subsection{Statement of results}
To any  Feynman graph  $G$ one associates a Feynman integral,  which is given  by a possibly divergent projective integral of the general parametric form
\begin{equation} \label{introIG}
I_G(q,m) = \int_{\sigma_G}  \omega_G \quad \hbox{ where }  \quad \omega_G=  {P(\alpha_e) \Omega_G  \over \Psi_G^A\, \Xi_G(q,m)^B}\ .
\end{equation} 
Here $\Psi_G$, $\Xi_G(q,m)$, and $P(\alpha_e)$ are certain  polynomials in variables $\alpha_e$ indexed by the edges of $G$, $\Omega_G$ is defined in $(\ref{OmegaGdef})$,  $A,B \in \Z$, and $\sigma_G$ is the domain where all
$\alpha_e \geq 0$. 
These quantities are involved in  predictions for particle collider experiments.
It is immediate from this integral representation, when it converges, that amplitudes are families of periods \cite{KoZa}, depending on kinematic data such as particle
momenta  $q=(q_i)_i$ and masses $m=(m_e)_e$.   
A  deep idea, originating with  Grothendieck, is that there should exist a Galois theory of periods \cite{An1,An2}, extending the classical Galois theory of algebraic numbers. We shall apply these ideas to the integrals of the form $(\ref{introIG})$.

The first problem, when trying to set up a Galois theory of periods, is that one immediately runs into difficult conjectures concerning  motives.
A simple way around this is to work in a category of  systems of realisations\footnote{We shall abusively use the word  `motive' to signify an object in such a category which is the image of the cohomology of an algebraic variety.}
and the second part of these notes \cite{NotesMot}  explain how this can be done without difficulty. 
In brief, the objects in a   category of realisations $\HH(S)$  on a smooth scheme $S$ over $\Q$ are triples $\V=(\V_{B}, \V_{dR}, c)$ where  $\V_B$ is a local system of $\Q$-vector spaces over $S(\C)$;
$\V_{dR}$ is an algebraic vector bundle with integrable connection on $S$
and regular singularities at infinity, and $c$ is a  Riemann-Hilbert correspondence  between $\V_B$ and $\V_{dR}$.  This data should define a variation of mixed Hodge structures on $S$.
A further subtlety, which has nothing to do with 
questions about motives, is how to interpret any given family of integrals, such as $(\ref{introIG})$,  as a period of the cohomology of an algebraic variety since this can involve choices.  It turns out that it  can be done canonically for Feynman amplitudes.

\begin{thm}  For any Feynman graph $G$ with generic kinematics $q,m$, there is a canonical way to associate to a  convergent integral $(\ref{introIG})$:

(i).   an object
$\mot_G $ in $\HH(S)$, where  $S$ is a Zariski open in  a space of kinematics,

(ii).  a  de Rham class $[\omega_G]$ in the generic fiber of $(\mot_G)_{dR}$,

(iii). a Betti class $[\sigma_G]$ in a certain (Euclidean) fiber  of $(\mot_G)^{\vee}_B$, 

\noindent
such that the integral $(\ref{introIG})$ is the period
$$ \sigma_G(c(\omega_G)) = I_G(q,m) \ .$$
\end{thm} 

The object $\mot_G$ is defined as a compatible system of cohomology groups  (namely,    Betti  and de Rham)  of a family of pairs of algebraic varieties 
\begin{equation} \label{intropair} 
(\Pro^G \backslash Y_G ,  D \backslash (D \cap Y_G))
\end{equation}
where $\Pro^G$ is a  blow up of projective space along linear subvarieties, $Y_G$ is the strict transform of the zero loci of $\Psi_G$ and $\Xi_G(q,m)$, and 
$D$ is a certain strict normal crossing divisor independent of $q,m$. This theorem generalises a  result in \cite{BEK} which treats a family of cases with no kinematic dependence ($B=0$ and $S$ is a point).

This theorem enables us to replace the Feynman integral  $I_G(q,m)$ with a canonical `motivic version' $I^{\mm}_G(q,m)=[\mot(G), \sigma_G, \omega_G]^{\mm}$ which is 
defined as a matrix coefficient of the torsor of isomorphisms between two fiber functors on $\HH(S)$. The integral itself $(\ref{introIG})$ can be retrieved from it by applying the period homomorphism.  The motivic integral now carries the action  of an affine group scheme which is the Tannaka group of $\HH(S)$ with respect to the de Rham fiber functor at the generic point.  This group factors through a certain quotient which acts faithfully
on the motivic periods of $\mot_G$ relative to $\sigma_G$,  where $G$ has  at most $Q$ external momenta and $M$  possible non-zero masses, hereafter called `of type $(Q,M)$'. This quotient is denoted by $C_{Q,M}$ and could be called a \emph{cosmic Galois group}, following \cite{Folle}. 
In this way, every convergent integral $(\ref{introIG})$ is replaced by a finite-dimensional representation of $C_{Q,M}$, and  this enables us to assign an array of new invariants
to amplitudes. Examples include: a weight filtration\footnote{This provides a rigorous meaning to many statements in the physics literature referring
to the `transcendental weight' of amplitudes}, a canonical connection (differential equation with regular singularities),
Hodge polynomials, Galois conjugates,  and various measures of complexity such as the  rank, dimension, and unipotency degree. These ideas are explored in the second half of these notes \cite{NotesMot}. In particular, we generalise the notion of the `$f$-alphabet' decomposition of multiple zeta values, which has various applications in physics, 
to all periods. 
 The next step in this programme is to formulate conjectures which relate topological invariants of graphs to the above-mentioned invariants of their motivic periods.

We can subsequently define $\HF^{\mm}_{Q,M}$ to be the vector space spanned by  all motivic periods of \motic\!\!\footnote{This notion is defined in \S\ref{sect: Section3}   and generalises
 the notion of one-particle irreducible.}  Feynman graphs of type $(Q,M)$ relative to $\sigma_G$.  The group action 
\begin{equation} \label{introCQMaction} 
C_{Q,M} \times \HF^{\mm}_{Q,M} \To \HF^{\mm}_{Q,M}
\end{equation} 
can be expressed equivalently as a coaction 
\begin{equation} \label{introDeltacoaction} 
 \Delta: \HF^{\mm}_{Q,M} \To \HF^{\mm}_{Q,M} \otimes_{k_{Q,M}} \HF^{\dR}_{Q,M}
 \end{equation}
where $\HF^{\dR}_{Q,M}$ is the ring generated by the de Rham periods of $\mot_G$. Note that  the formula we gave for  the motivic dilogarithm  in the first paragarph involved the unipotent coaction and unipotent de Rham periods for simplicity (the full coaction involves powers of $\Lef^{\dR}$, see \cite{NotesMot}, last lines of \S10). 
A key point is that there is a general formula for this coaction in terms of the cohomology
$(\mot_G)_{dR}$, and this can be computed explicitly
in some cases.  The apparently unphysical case of graphs with no masses or momenta plays a special role in this theory. Indeed,  $\HF^{\mm}_{0,0}$ 
is an algebra, and  each $\HF^{\mm}_{Q,M}$ is a  module over it:
$$\HF^{\mm}_{0,0}  \times \HF^{\mm}_{Q,M} \To \HF^{\mm}_{Q,M}\ .$$

\subsection{Product structure and stability} 
There is \emph{a priori} no reason whatsoever for  the action of $C_{Q,M}$ to preserve the space of motivic  amplitudes, which form a small subfamily of 
integrals $(\ref{introIG})$ with highly specific numerators.\footnote{From now on, we shall  loosely call \emph{amplitude} a  Feynman  integral of the form $(\ref{introIG})$ with a specific numerator,  which arises, for instance,  from the Feynman rules of a given quantum field theory,  as opposed to an arbitrary integral with that shape. In the literature, the word amplitude  is often reserved to describe the sum of all Feynman integrals at a given loop order. We shall call the  latter  the `full amplitude' in accordance with some authors. }  However, using the formalism of motivic periods,  Panzer  and Schnetz  
computed the Galois action on a huge family of amplitudes in $\phi^4$ theory of type $(0,0)$ 
and verified, astonishingly,  that they are preserved by $C_{0,0}$ in all cases.    The motivation for these lectures was to try to understand this extraordinary fact.

The theory outlined in these notes is best explained by  the following conjecture.

\begin{conj}   \label{conjIntro}
The  motivic periods of a graph $G$ of type $(Q,M)$ are regularised  versions of  motivic integrals of the form  $(\ref{introIG})$. Those   of weight
$\leq k$
are in the algebra generated by regularised  motivic periods of graph minors   of $G$ with at most $k+1$ edges. 
\end{conj}

This conjecture means  that the Galois conjugates of a motivic amplitude of low weight of a graph  should be a regularised motivic period of its sub-quotient graphs of the form $(\ref{introIG})$. Since there 
are few graphs with a bounded number of edges, this  provides a  constraint on amplitudes to all orders. We call this the small graphs principle.
The upshot is that the topology of a graph  constrains the Galois theory of its amplitudes.  In the case $(Q,M)=(0,0)$, this theory partially  explains the observations of Panzer and Schnetz. 

What is presently  lacking for a proof of this conjecture is a suitable  notion of regularisation for motivic periods.\footnote{Unfortunately, when writing the technical background \cite{NotesMot} for these notes, I had not forseen that convergent Feynman integrals might require  the theory of limiting mixed Hodge structures, and a corresponding notion of limiting motivic period, and so it was not discussed.  I do not believe that this, or conjecture  \ref{conjIntro}, should pose  any major difficulties.}  In these notes, we prove a weaker, but more precise version of this conjecture, in which we replace the word `regularised' with `affine', which has a technical meaning (theorem \ref{thmMainAffine}).  It implies the

\begin{thm} The vector space $W_k \HF_{Q,M}^{\mm}$ is finite-dimensional.  In particular, the vector space generated by convergent   integrals $(\ref{introIG})$ which are of bounded weight\footnote{let us call a period  of weight $\leq k$
if it is the image of a motivic period of weight $\leq k$ under the period homomorphism}, for $G$ of any fixed type $(Q,M)$,  is finite-dimensional. 
\end{thm} 

This theorem is non-trivial since   there are infinitely many graphs, and therefore implies infinitely many relations between periods of different graphs.

We also show that the weight-graded pieces of $\mot(G)$ in weight $\leq k$ are indeed generated by minors  of $G$ with at most $\leq k+1$ edges,  and we completely determine $W_2 \HF^{\mm}_{0,0}$. These simple facts already imply strong and concrete constraints on amplitudes to all loop orders - for example,   the last statement  gives a non-trivial and rather subtle  condition  on  amplitudes which has been 
verified numerically for nine loop scalar graphs by Panzer and Schnetz.

At the heart of this theory is a  set of identities for graph polynomials. For the usual graph (Kirchhoff) polynomial it takes the form of 
 a partial factorisation 
$$\Psi_G = \Psi_{\gamma} \Psi_{G/\gamma} + R^{\Psi}_{\gamma, G}$$
where $\gamma \subset G$ is any subgraph (defined by a subset of edges of $G$), and $R^{\Psi}_{\gamma, G}$ is a remainder term
of higher degree in the variables of $\Psi_{\gamma}$ than $\deg(\Psi_{\gamma})$. This identity has been known for some time and is used in the parametric
theory of renormalisation \cite{Angles}, although  only  in the special case when $\gamma$ is a divergent subgraph.  It generalises in two different ways for
the graph polynomial $\Xi_G(q,m)$, which we call the ultraviolet and infra-red factorisations. This requires some genericity assumptions on the external kinematics. The infra-red factorisation  identities are new.

The geometric incarnation of these identities implies that the open strata  of $(\ref{intropair})$
$$D_i \backslash (Y_G \cap D_i)$$
are products of varieties of the same type. This was already observed in \cite{BEK} in the case $(Q,M)=(0,0)$ mentioned above. 
Such a family of varieties defines a type of operad in the category of schemes over $\Z$.
Although we shall barely mention operads in these notes,  the notion of `operad in the category of motives'  imposes  strong constraints
on its possible periods, and merits further  study.\footnote{A similar example of such a system of stratified varieties with product structure are the moduli spaces $\mathfrak{M}_{g,n}$ of curves of genus $g$ with $n$ marked points. } A similar theory to the one described here should
hold more generally for any family of varieties with this product-structure.

We are still very far from exploiting all the consequences of this geometric structure underlying amplitudes.
In fact, the product  structure on the faces of $(\ref{intropair})$ is such a  rigid constraint  that it almost completely determines the polynomials $\Psi_G$ and $\Xi_G(q,m)$, as we prove in \S\ref{sect:  moticHopfalgebra}.

\subsection{Contents}
In \S\ref{sect: Section1} we recall some basic notions relating to Feynman graphs and graph polynomials.  In \S\ref{sect: Section2} we
prove factorisation theorems for graph polynomials. 
In \S\ref{sect: Section3} we study the notion of a \motic subgraph of a Feynman graph. These are in one-to-one correspondence
with the locus where the domain of integration $\sigma_G$ meets the singularities of the Feynman integrand and subsume both ultraviolet and infra-red type 
divergences. In \S\ref{sect:  moticHopfalgebra} we consider a Hopf algebra of graphs where the coproduct ranges over \motic subgraphs and characterize graph polynomials
by their partial factorisation properties.  The motic Hopf algebra is a  generalisation of the Connes-Kreimer coproduct  for scalar graphs since it also  takes into account  certain infra-red subdivergences. 
In \S \ref{sect: LinearBlowups} we study blow-ups of projective space along linear subspaces, and define some affine models which generalise the partial compactifications of  moduli spaces $\mathfrak{M}_{0,n}^{\delta}$ of the author's thesis.  In \S\ref{sect: SectionGraphMotives} we define the graph motive and prove its recursive product structure.   In  \S\ref{sect: MotAmp} we define the motivic amplitude of a Feynman graph, and prove some stability results  in 
\S\ref{sect: Weightsstability}. 
In \S \ref{sect: ConstantCosmic} we focus on the case of graphs with no kinematic dependence and prove some modest results in the direction of the coaction conjecture of \cite{PanzerSchnetz}. Although they  are of limited physical significance,  this family of  graphs plays an important and central  role in the theory.
In \S \ref{sect: Finalsection} we discuss applications of the cosmic Galois group in the case of graphs with general kinematics and state some conjectures
and problems for further study.

The appendix   \S \ref{sect: nonglobalexample} provides a fully worked example of such a graph  using the methods described here. It requires some technical tools
which are set out in  \S  \ref{sect: App1}.

\subsection{Relation to other work}
A number of expressions  in this paper have appeared in  the literature with possibly different meanings. They are listed below:
\begin{itemize}
\item `Cosmic Galois group'.  The phrase `cosmic Galois group' was invented by P. Cartier. Later, Connes and Marcolli made a precise definition of a cosmic Galois group in relation to  renormalisation, in the papers 
\cite{CM1, CM2}. It is not clear if it is at all related to the groups  defined here.
\item  `Motivic amplitude'.  This phrase occurs with a different  definition  in \cite{Cluster}, where it means a certain tensor of elements in a field, and only makes
sense  in the mixed Tate case. This notion can be retrieved as a  very special case of the symbol (defined in \cite{NotesMot}) of   de Rham motivic amplitudes. 
\item  `Motivic multiple zeta value'. There are two versions of motivic multiple zeta values in use, an earlier one due to Goncharov, for which the 
motivic version of $\zeta(2)$ vanishes, and which  do not posses a period map,  and another  for which $\zetam(2)$ is non-zero. 
Sending $\zetam(2)$ to zero would destroy much of the  structure in amplitudes discovered in \cite{PanzerSchnetz}. 
 \end{itemize}
It is also important to emphasize that motivic amplitudes do not form a Hopf algebra. The main coaction formula $(\ref{introDeltacoaction})$ is asymmetric - on one side we have motivic periods which have
a well-defined map to numbers, and on the other,  de Rham periods which do not. However we can associate symbols to  de Rham periods  (in the  differentially unipotent case), and also single-valued periods \cite{NotesMot}. 

\subsection{Acknowledgements} This work was partly funded by ERC grant 257638.
I owe many thanks to the  participants of  my lecture series at IHES for their interest and many pertinent questions, especially Pierre Cartier and Joseph Oesterl\'e.   Many thanks to Oliver Schnetz and Erik Panzer for their remarkable computations,  which  motivated this project, and also  to Cl\'ement Dupont for helpful discussions.

\section{Scalar Feynman graphs and Symanzik polynomials}\label{sect: Section1}
We first recall some basic definitions of  scalar Feynman graphs, before describing the corresponding integrals in parametric form.  The reader should be aware that our conventions occasionally differ from the  standard ones in a  few  minor details.

\subsection{Feynman graphs} \label{sect: FeynmanGraphs}
A   Feynman graph is a   graph $G$ defined by 
$$( V_G, E_G, E^{\ext}_G)$$
where $V_G$ is the set of vertices of $G$, $E_G$ is the set of internal edges of $G$, 
and $E^{\ext}$ is a set of external half-edges (also known as legs). Their endpoints are encoded by maps
 $\partial: E_G \rightarrow   \mathrm{Sym}^2 \,V_G$ and  
$\partial: E^{\ext}_G\rightarrow  V_G$.
We shall assume that the vertices  with external legs  (image of $E^{\ext}_G$ in $V_G$) lie in a single connected component of $G$.

 A Feynman graph  additionally comes with kinematic data:
\begin{itemize}
\item a particle mass $m_e \in \R $ for every internal edge $e\in E_G$, 
\item a  momentum $q_i  \in \R^d$ for every external half-edge $i \in E^{\ext}_G$,
\end{itemize} 
where $d\geq 0$, the dimension of space-time,  is fixed. The  internal edges of $G$ are labelled  if there is a bijection $E_G \leftrightarrow S$ with a fixed set $S$. 
The external half-edges will be oriented inwards,   so all momenta are  incoming and are  subject to momentum conservation 
 \begin{equation} \label{eqn: momcons}
 \sum_{i \in E^{\ext}_G} q_i = 0 \ .
 \end{equation} 
 Some of the internal masses $m_e$ will be zero. Let $M_G \subset E_G$ denote the set of  internal edges $e$ of $G$ for which
$m_e\neq 0$. In our figures, the mass-carrying edges in $M_G$ will be drawn with a doubled edge.

In this paper, a  subgraph $H$ of $G$ will be a graph  defined by a triple $(V_H, E_H, E_H^{\ext})$ where $V_H \subset V_G, E_H \subset E_G$ and either $E^{\ext}_H = E_G^{\ext}$ or
$E^{\ext}_H = \emptyset$.    Note that for $H$ to be a Feynman subgraph, the extra condition that  the vertices $E_H^{\ext}$  lie in a single connected component of $H$ must also  hold.  
This guarantees  that momentum conservation holds for every component of $H$.
  The particle masses of a Feynman subgraph $H\subset  G$ are determined by  the following condition:
$$\begin{array}{lccl}
\hbox{either}    &   E^{\ext}_H= E^{\ext}_G  & \hbox{ and}  &M_H = M_G    \\
 \hbox{or }\qquad   &    E^{\ext}_H= \emptyset  & \hbox{ and}&  M_H = \emptyset   \ .\end{array}
$$
In the  former case, $H$ contains all massive edges of $G$, and inherits the corresponding  masses.
In the  latter case,  $H$ is viewed as a massless diagram and  $m_e=0$ for all $e\in E_H$.
  Thus in these notes,  a Feynman subgraph   either meets all external legs in a single connected component, and contains all massive edges; or  is
  considered to be massless  with no external momenta.

 Note that  all external legs correspond to a potentially non-zero momentum;  external legs which would ordinarily be considered to have zero incoming momentum
will  simply be omitted. This forces  our  graphs $G$  to have vertices of  varying degrees: a graph $G$ is said to be in $\phi^n$ if every vertex has degree $\leq n$.

A tadpole, or self-edge, is a subgraph of $G$ of the form $(\{v\},\{v,v\}, \emptyset)$.

We shall use the following notation 
for the basic combinatorial invariants of $G$:
\begin{eqnarray}
h_G  &= &  h^1(G) \quad \hbox{ the loop number of } G \nonumber \\
\kappa_G  &= &  h^0(G) \quad \hbox{ the number of connected components  of } G \nonumber \\
N_G & = & |E_G| \quad \,\,\,\,\hbox{ the number of edges of } G\ .  \nonumber 
\end{eqnarray}
They do not depend on the external legs of $G$. 
Euler's formula  states that
\begin{equation} \label{eqn: Eulerformula} N_{G} - V_G = h_G - \kappa_G\ .
\end{equation}
We  define the following equivalence relation on Feynman graphs. If a vertex $v\in V_G$  has several incoming momenta $q_1,\ldots, q_n$
we can replace it with a single incoming momentum $q_1+ \ldots + q_n$:
\begin{center} 
\fcolorbox{white}{white}{
  \begin{picture}(192,72) (115,-27) 
    \SetWidth{1.0}
    \SetColor{Black}
    \Line[arrow,arrowpos=0.4,arrowlength=5,arrowwidth=2,arrowinset=0.2](128,29)(160,-19)
    \Line[arrow,arrowpos=0.4,arrowlength=5,arrowwidth=2,arrowinset=0.2](144,29)(160,-19)
    \Line[arrow,arrowpos=0.4,arrowlength=5,arrowwidth=2,arrowinset=0.2](176,29)(160,-19)
    \Line[arrow,arrowpos=0.4,arrowlength=5,arrowwidth=2,arrowinset=0.2](192,29)(160,-19)
    \Line[arrow,arrowpos=0.4,arrowlength=5,arrowwidth=2,arrowinset=0.2](288,29)(288,-19)
    \Text(118,31)[lb]{{\Black{$q_1$}}}
    \Text(137,31)[lb]{{\Black{$q_2$}}}
    \Text(154,25)[lb]{{\Black{$\ldots$}}}
    \Text(166,31)[lb]{{\Black{$q_{n-1}$}}}
    \Text(190,31)[lb]{{\Black{$q_n$}}}
     \Text(230,0)[lb]{\Large{\Black{$\sim$}}}
    \Text(262,31)[lb]{{\Black{$q_1+\ldots + q_n$}}}
    \Vertex(160,-19){2}
    \Vertex(288,-19){2}
  \end{picture}
}
\end{center}Our notion of Feynman subgraph respects this equivalence relation.  
The  graph polynomials defined below will only depend on  equivalence classes.

We say that a Feynman graph is \emph{of type} $(Q,M)$ if it is equivalent to a graph with at most $Q$ external kinematic parameters, and at most  $M$ non-zero particle masses.  We shall call a graph  \emph{one-particle irreducible}, or 1PI, if  every connected component is  $2$-edge connected (i.e. deleting any edge  causes the loop number to drop). 
\begin{example} \label{example: Dunce}The following  Feynman graph will be  our basic example to illustrate the ideas in this paper. It will be referred to several times throughout this text.
\begin{center} 
\fcolorbox{white}{white}{
  \begin{picture}(192,105) (19,-5)
    \SetWidth{1.0}
    \SetColor{Black}
    \Line[arrow,arrowpos=0.5,arrowlength=5,arrowwidth=2,arrowinset=0.2](32,41)(80,41)
    \Vertex(80,41){3}
    \Line[double,sep=2](80,41)(128,73)
    \Line(80,41)(128,9)
    \Arc(197.643,41)(76.643,155.322,204.678)
    \Arc[clock](68,41)(68,28.072,-28.072)
    \Vertex(128,9){3}
    \Vertex(128,73){2.828}
    \Line[arrow,arrowpos=0.5,arrowlength=5,arrowwidth=2,arrowinset=0.2](162,93)(128,73)
    \Text(100,61)[lb]{{\Black{$1$}}}
    \Text(100,16)[lb]{{\Black{$2$}}}
    \Text(114,38)[lb]{{\Black{$3$}}}
    \Text(139,38)[lb]{{\Black{$4$}}}
    \Text(24,38)[lb]{{\Black{$q_1$}}}
    \Text(164,92)[lb]{{\Black{$q_2$}}}
  \end{picture}
}
\end{center} 
This graph has a single non-zero mass, namely $m_1$, and  $m_2=m_3=m_4=0$. Momentum conservation demands that 
$q_1  =- q_2 .$  The bottom right vertex, which meets edges $2,3,4$ has zero incoming external momentum.

\end{example}

\subsection{Graph polynomials} Let $G $ be a Feynman graph.  Recall that a tree is a connected graph $T$ with $h_T=0$.  A forest is any graph $T$ with $h_T=0$. 
\begin{defn} A spanning $k$-tree of $G$ is a subgraph $T =T_1 \cup \ldots \cup T_k \subset G$ which has exactly $k$ components $T_i$ such that $T_i$ is a tree and $V_T = V_G$.

\end{defn}
A spanning $1$-tree is simply called  a spanning tree.

\begin{defn} Let $G$ be a \emph{connected} Feynman graph. \label{defn: Symanzik}
The Kirchhoff polynomial (or 1st Symanzik polynomial) is the  polynomial in $\Z[\alpha_e, e\in E_G]$ defined by 
\begin{equation} \label{eqn: psidefn}
\Psi_G = \sum_{T \subset G}  \prod_{e \notin T} \alpha_e\ , \end{equation} 
where the sum is over all spanning trees $T$ of $G$. If  $G$ has several connected components $G_1,\ldots, G_n$ we shall  define
\begin{equation} \label{eqn: psiproddefn}
\Psi_G = \prod_{i=1}^n \Psi_{G_i}
\ . \end{equation} 
Note that one sometimes takes $(\ref{eqn: psidefn})$ as the general definition of graph polynomial. It  differs from $(\ref{eqn: psiproddefn})$ and vanishes if $G$ has more than one connected component.

The second Symanzik polynomial is defined for  \emph{connected}  $G$ by 
\begin{equation} \label{eqn: phidefn}
\Phi_G(q) = \sum_{T_1 \cup T_2 \subset G}  \,(q^{T_1})^2\prod_{e \notin T_1 \cup T_2} \alpha_e \ ,
\end{equation}
where the sum is over all spanning $2$-trees $T=T_1 \cup T_2$ of $G$ and  $q^{T_1} =\sum_{i\in E^{\ext}_{T_1}}q_i $ is the total momentum entering $T_1$. It equals $ -q^{T_2}$  by momentum conservation $(\ref{eqn: momcons}).$
If  $G$ has several connected components $G_0,G_1, \ldots, G_n$, then by our definition of a Feynman graph, exactly one component, say $G_0$, 
contains all external momenta.

 In this case we  define
\begin{equation} \label{eqn: phiproddefn}
 \Phi_G(q) = \Phi_{G_0}(q) \prod_{i\geq 1}  \Psi_{G_i}
\ . \end{equation}

\end{defn} 

Applying   Euler's formula $(\ref{eqn: Eulerformula})$ to $G$ and  a spanning $\kappa_T$-tree $T\subset  G$  implies  that 
\begin{equation} \label{eqn: Edgecomplementnumber}
N_G - N_T =h_G+\kappa_T-\kappa_G \ .
\end{equation} 
Since the right-hand side is independent of $T$, 
 $\Psi_G$ and $\Phi_G$ are homogeneous with respect to Schwinger parameters and  have the following degrees (in the $\alpha_e$):
\begin{eqnarray} \label{eqn: degreesofPsiandPhi}
\deg\, \Psi_G  & =  & h_G \\ 
\deg\, \Phi_G(q) & =  & h_G + 1  \nonumber 
\end{eqnarray} 
These equations can also be deduced from the contraction-deletion relations which are stated below.
It is a crucial fact for the arithmetic of Feynman integrals that the coefficients of every monomial in $\Psi_G$ are only $0$ or $1$. Furthermore, $\Psi_G$ and $\Phi_G$ are of  degree at most one in every Schwinger parameter $\alpha_e$, but these two facts  play a minor role in these notes. See \cite{BrFeyn}, \cite{Modularphi4} for applications of these facts.

\begin{defn} Let $G$ be a  Feyman graph. Define
\begin{equation} 
\label{eqn:  Xidefn}
\Xi_G(q,m) = \Phi_G(q) +  \big(\sum_{e \in E_G} m^2_e \alpha_e \big) \Psi_G\ .
\end{equation}
By $(\ref{eqn: degreesofPsiandPhi})$, the polynomial $\Xi_G$ is homogeneous in the $\alpha_e$ of degree $h_G+1$.
\end{defn} 

 Since the graph polynomials $\Psi_G$, $\Phi_G(q)$, $\Xi_G(q,m)$ only depend on total momentum flow, they  are well-defined on equivalence classes of graphs. 
\begin{example} Let $G$ be the graph of example $\ref{example: Dunce}.$ Then we have 
\begin{eqnarray}
\Psi_G & = & \alpha_1 \alpha_3 + \alpha_1 \alpha_4 + \alpha_2 \alpha_3 + \alpha_2 \alpha_4 + \alpha_3 \alpha_4  \nonumber \\
\Phi_G(q) & = & q_1^2 ( \alpha_1 \alpha_2 \alpha_3 + \alpha_1 \alpha_2 \alpha_4 + \alpha_1  \alpha_3  \alpha_4)  \nonumber  \\
\Xi_G(q,m ) & = & q_1^2 ( \alpha_1 \alpha_2 \alpha_3 + \alpha_1 \alpha_2 \alpha_4 + \alpha_1  \alpha_3  \alpha_4) +m_1^2 \alpha_1 \Psi_G  \nonumber 
\end{eqnarray} 
\end{example}

\begin{rem} The reader is warned that in  the literature, the polynomial $(\ref{eqn:  Xidefn})$ is often written with the opposite sign (or equivalently, all momentum terms $q$ are replaced by $i q$). To find one's bearings, consider the following familiar bubble diagram with two equal masses $m_1=m_2=m$:
\begin{center}
\fcolorbox{white}{white}{
  \begin{picture}(160,64) (19,-27)
    \SetWidth{1.0}
    \SetColor{Black}
    \Line[arrow,arrowpos=0.5,arrowlength=5,arrowwidth=2,arrowinset=0.2](32,0)(80,0)
    \Vertex(80,0){3}
    \Vertex(128,0){3}
    \Text(24,0)[lb]{\Large{\Black{$q$}}}
    \Line[arrow,arrowpos=0.5,arrowlength=5,arrowwidth=2,arrowinset=0.2,flip](128,0)(176,0)
    \Arc[double,sep=2](104,-8)(25.298,18.435,161.565)
    \Arc[double,sep=2,clock](104,8)(25.298,-18.435,-161.565)
    \Text(164,2)[lb]{\Large{\Black{$-q$}}}
    \Text(100,22)[lb]{\Large{\Black{$1
$}}}
    \Text(100,-28)[lb]{\Large{\Black{$2$}}}
  \end{picture}
}
\end{center}
 It satisfies
$$\Xi_G (q,m) = q^2 \alpha_1 \alpha_2 + m^2 (\alpha_1+ \alpha_2)^2\ . $$
Its discriminant is $4m^2+q^2$, and so its Landau singularity occurs at $q^2 =- 4m^2$. With our chosen  sign convention, it will lie outside the Euclidean region where $q, m$ are real. This, and other,  typical physical infra-red singularities (where, for example, terms in $\Xi_G(q,m)$ cancel out  altogether)  will be excluded from the present set-up (although certain other types of infra-red singularities will   be allowed).   However, they can  still be treated   in the present theory on a graph by graph basis after analytic continuation in the space of kinematics $(q,m)$. That the analytic continuation exists follows from the fact that the discriminant locus is algebraic of codimension $\geq 1$, and its complement is connected in the analytic topology. Therefore there exists a path from a point in the Euclidean region to an open subset of the region  $q \in i\R$,  along which the integral is  analytic. 
 \end{rem}

\subsection{Feynman integral in projective space} Let $d\in 2 \N$ be an even integer, which denotes the dimension of space-time. 
Here  it  will always be fixed, and will be dropped from the notation. 
Our version of the Feynman integral in parametric form differs marginally the usual presentation. 
In order to kill two birds (namely, the case with no kinematics, and the case with non-trivial kinematics) with one stone we shall take the following definition, after omitting certain pre-factors (see  \cite{Angles} \S3 of a more rigorous derivation from first principles):
 \begin{equation}
\label{eqn: projectiveFeynmanint}
I_G(q,m) = \int_{\sigma} \omega_G(q,m)
\end{equation} 
where 
\begin{equation} \label{eqn:  omegaGdef}
\omega_G(q,m) =  {1\over \Psi_G^{d/2}} \Big( { \Psi_G \over \Xi_G(q,m)} \Big)^{N_G - h_G d/2}  \Omega_G
\end{equation} 
and
\begin{equation} \label{OmegaGdef}
\Omega_G = \sum_{i=1}^{N_G} (-1)^i \alpha_i \,  d\alpha_1 \wedge \ldots \wedge \widehat{d\alpha_i} \wedge \ldots \wedge d \alpha_{N_G}
\end{equation}
where $\widehat{d\alpha_i}$ means that the term $d\alpha_i$ is omitted. Note that 
the form $\omega_G$ is homogeneous of degree $0$, which follows from $(\ref{eqn: degreesofPsiandPhi}).$
Finally, let 
$$\sigma \subset \Pro^{N_G-1}(\R)$$
be the coordinate simplex  defined in projective coordinates by 
\begin{equation} \label{eqn: sigmadef}
\sigma = \{(\alpha_1: \ldots : \alpha_{N_G} )  \in  \Pro^{N_G-1}(\R) :  \  \alpha_i \geq 0 \}\ .
\end{equation}
The integral $(\ref{eqn: projectiveFeynmanint})$ may not necessarily converge.  Necessary and sufficient conditions for its convergence, in a certain kinematic region, will be stated below.
The derivation of the parametric form of the Feynman integral $(\ref{eqn: projectiveFeynmanint})$ from its momentum space representation
 using the Schwinger trick is nicely explained in \cite{PanzerPhd}.

\begin{rem}
The integral $(\ref{eqn: projectiveFeynmanint})$ is a drastic simplification in certain situations. For example, if $d=4$ and $G$ is primitive overall log-divergent ($N_G = 2h_G$), then the second factor in $\omega_G(q,m)$ drops out and it has no dependence on external masses or momenta. For such a graph, 
 we obtain 
$$I_G = \int_{\sigma} {\Omega_G \over \Psi_G^2}\ .$$ 
In the case of the wheel with three spokes,  this equals $6 \zeta(3)$ which is  its residue (coefficient of  $1/\varepsilon$ in dimensional regularisation).  The full vertex function is
$$\int_{\sigma}  { \log (\Xi_G(q,m)/ \Xi_G(q_0,m_0))  \over \Psi_G^2} \Omega_G$$
where $q_0,m_0$ is a chosen renormalisation point. Such integrals, and their derivation, are discussed at length in \cite{Angles}. They can also be viewed as period integrals either by writing  the $\log$ in the numerator as  an integral, or by differentiating with respect to a renormalisation scale as in \cite{Angles} to make the integrand algebraic.  It is highly likely that the theory described in this paper also extends to this situation,  but   a discussion of renormalisation would have made the present paper overly lengthy. 
\end{rem} 

\subsection{Edge subgraphs and their quotients}
Let $G=(V_G, E_G, E^{\ext}_G)$ be a Feynman graph.
A set of internal edges $\gamma \subset E_G$  defines a  subgraph of $G$ as follows. Write $E_{\gamma} = \gamma$ and let $V_{\gamma}$  be the set of endpoints of elements of $E_{\gamma}$.

\begin{defn}  \label{defn: momentumspanning} A set of edges $\gamma \subset E_G$ is \emph{momentum-spanning} if $\partial E^{\ext}_G \subset V_{\gamma}$, and
the vertices $E^{\ext}_G$ lie in a single connected component of the graph $(V_{\gamma}, E_{\gamma})$. 
\end{defn}

We define the  subgraph associated to $\gamma\subset E_G$  by 
\begin{equation} \label{eqn: Edgesubgraphdefn} 
 ( V_{\gamma}, E_{\gamma},E^{\ext}_{\gamma} ) 
\end{equation}
where $E^{\ext}_{\gamma} = E^{\ext}_G$ if $\gamma$ is momentum-spanning, and $E^{\ext}_{\gamma} = \emptyset$  otherwise. Thus the Feynman graph $ (\ref{eqn: Edgesubgraphdefn}) $ 
inherits all external momenta of $G$ if it is momentum-spanning and has no external momenta otherwise. We shall call $(\ref{eqn: Edgesubgraphdefn})$
the \emph{edge-subgraph} associated to $\gamma$, and denote it  also by $\gamma$ when no confusion arises.

The quotient of $G$ by an edge-subgraph $\gamma$ is   defined by 
$$ G/ \gamma = (V_G/\! \sim,\,  (E_G \backslash \gamma)/\!\sim ,\,  E^{\ext}_G /\!\sim ) $$
where $\sim$ is the equivalence relation on vertices of $G$ where  two vertices are equivalent if and only if they are vertices of  the same connected component of $\gamma$, and the induced equivalence relation on edges (unordered pairs of vertices).  It is a Feynman graph. Every connected component of $\gamma$ corresponds to a unique  vertex in  $G/\gamma$.
Note that  $\gamma$ is momentum-spanning if and only if  $G/\gamma$ is equivalent to a graph with no external momenta (by momentum conservation). In this case we can take $E^{\ext}_{G/\gamma} = \emptyset$.

In this way,  exactly one of the two Feynman graphs $\gamma$ and $G/\gamma$ is  equivalent to a Feynman graph with non-zero  external momenta: 
if $\gamma$ is momentum spanning it is $\gamma$, otherwise it is $G/\gamma$.

\subsection{Contraction-deletion}

Let $G = (V_G, E_G, E^{\ext}_G)$  be a Feynman graph. The \emph{deletion} of an edge $e$ in $G$ is the graph $G\backslash e$ defined by 
deleting the edge $e$ but retaining  its endpoints:
$$ G\backslash e = (V_G,\,  E_G \backslash \{e\},\,  E^{\ext}_G)\ .$$
In general, it is not  a union of  Feynman graphs since momentum conservation may not hold on each of its connected components.
 
One sometimes encounters the following variant of the previous notion of graph-quotient.  It will be denoted by a double slash to distinguish it from the ordinary quotient.   For an edge-subgraph $\gamma$,
let
$G\q \gamma$  be the empty graph  if  $h_\gamma>0$ and  $$G \q \gamma = G/\gamma$$
if $\gamma$ is a forest. In the case of a single edge $e$, $G\q e$ is empty whenever $e$ is a tadpole.

It follows from Euler's formula $(\ref{eqn: Eulerformula})$ that 
\begin{equation} \label{eqn:  additivityofh}
h_G = h_{\gamma} + h_{G/\gamma} 
\end{equation}
for any  edge-subgraph $\gamma \subset G$ (which is not necessarily connected).

\begin{lem} (Contraction-deletion) Let $G$ be connected, and $e\in E_G$. Then 
\begin{eqnarray} \label{eqn: ContractionDeletion}
\Psi_G & = &  \Psi^0_{G\backslash e} \alpha_e + \Psi_{G\q e} \\ 
\Phi_G(q) & = &  \Phi^0_{G\backslash e}(q) \alpha_e + \Phi_{G\q e}(q) \ ,  \nonumber  
\end{eqnarray} 
where $\Psi^0_{G\backslash e}$ is given by the right-hand side of  $(\ref{eqn: psidefn})$: it is $\Psi_{G\backslash e}$ if $G\backslash e $ is connected and $0$ otherwise.  Likewise $\Phi^0_{G\backslash e}(q)$  is given by the right-hand side of  $(\ref{eqn: phidefn})$: it is equal to   $\Phi_{G\backslash e}(q)$ if $G\backslash e$ is connected and equal to   $\Psi_{G_1} \Psi_{G_2} (q^{G_1})^2=\Psi_{G_1} \Psi_{G_2} (q^{G_2})^2$ if $G\backslash e $ has two connected components $G_1, G_2$. 
\end{lem}
\begin{proof} Let $T$ be a spanning $k$-tree of $G$ (where $k\in \{1,2\}$).  The edge $e$ is not an edge of $T$ if and only if $T$ is a spanning $k$-tree of $G \backslash e$. 
By the definition of the graph polynomials, this gives rise to the first terms in the right hand sides of $(\ref{eqn: ContractionDeletion})$. Note that if $e$ is a tadpole, this is the only case which can occur. Now suppose that $e$ is not a tadpole.  If $e$ is  an edge of $T$, then $T/e$ is a spanning $k$-tree of $G/e$.  Conversely,
if $T'$ is a spanning $k$-tree of $G/e$, then there is a unique component of $T'$ which meets the vertex in $G/e$ defined by the endpoints of $e$. It follows that the inverse image of $T'$ in $G$, together with the edge $e$, 
is a spanning $k$-tree in $G$. This establishes a bijection between the set of  spanning $k$-trees in $T$ which contain $e$ and those of $G/e$.  The rest follows from definition  $\ref{defn: Symanzik}.$  \end{proof}

\begin{cor} \label{cor: Restrictiontoalpha=0}
It follows that from $(\ref{eqn: ContractionDeletion})$ and $(\ref{eqn:  Xidefn})$ that
\begin{eqnarray}
\Psi_G \big|_{\alpha_e=0}   & = &  \Psi_{G\q e}  \nonumber  \\
\Xi_G(q,m) \big|_{\alpha_e=0}   & = &  \Xi_{G\q e}(q,m) \ . \nonumber 
\end{eqnarray} 

\end{cor}

\subsection{Generic kinematics  and non-vanishing of graph polynomials}
We establish some non-vanishing results for graph polynomials which  hold for generic momenta. These will be  important for the sequel.

\begin{lem}  \label{lem: PsiGvanishing} A connected graph $G$ has a spanning tree. Equivalently,  
$\Psi_G \neq 0$. 
\end{lem} 
\begin{proof} 
Let $G$ be a connected graph with  $\Psi_G =0$. Then $\Psi_{G\q e}=0$ for all $e$ by $(\ref{eqn: ContractionDeletion})$. By repeatedly contracting  edges with distinct endpoints,  we obtain a graph  $G'$ with a single vertex such that $\Psi_{G'}=0$. It has a unique spanning tree consisting of this vertex,
so $\Psi_{G'} = \prod_{e\in E_{G'}} \alpha_e$, which is non-zero, a contradiction.
\end{proof} 
Consider the following condition on external momenta
\begin{equation} \label{eqn: genericmomenta} 
(\sum_{i \in I} q_i)^2  \neq 0 \quad \hbox{ for all } I \subsetneq E_G^{\ext}\ .
\end{equation} 
It respects the equivalence relation of Feynman graphs. 
\begin{lem}  \label{lem: PhiGvanishing} Let  $G$ be a  Feynman graph with  non-trivial external momenta (in other words, there exists a vertex $v\in \partial E^{\ext}_G$ such that the total momentum  $q^{\{v\}}$ entering $v$   is non-zero).  Then with condition $(\ref{eqn: genericmomenta})$, 
$ \Phi_G(q) \neq 0\ .$
\end{lem} 
\begin{proof} 
By momentum conservation $(\ref{eqn: momcons})$, there exist at least two vertices $v_1, v_2$ with  non-zero total incoming momenta $q_1, q_2$ respectively.  
Since $G$ is a Feynman graph, $v_1,v_2$ lie in the same connected component $G_0$. By the previous lemma,  there exists a spanning tree $T$ in $G_0$. Since $T$ is connected, there is a (shortest) path from $v_1$ to $v_2$ contained in $T$. Delete any edge $e'$ in this path  to obtain a spanning $2$-tree $T\backslash e' = T_1 \cup T_2$ such that $v_1 \in V_{T_1}$ and $v_2 \in V_{T_2}$. It contributes a non-zero monomial $(q^{T_1})^2 \prod_{e\notin T_1 \cup T_2} \alpha_e $ to the second Symanzik polynomial $\Phi_{G_0}(q)$ by $(\ref{eqn: genericmomenta})$. It cannot cancel out since all signs in the definition of  $\Phi_{G_0}(q)$ are positive.  Now apply $(\ref{eqn: phiproddefn})$ and the previous lemma to deduce  that $\Phi_G(q) \neq 0$.  
\end{proof} 
Now consider the following  condition on momenta and masses:
\begin{equation} \label{eqn: genericmassmomenta} 
(\sum_{i \in I} q_i)^2 + m_e^2  \neq 0 \quad \hbox{ for all } I \subsetneq E_G^{\ext}  \hbox{ and } e\in E_G \ .
\end{equation}

\begin{lem}  \label{lem:  Xivanishing} If   $(\ref{eqn: genericmomenta})$  and $(\ref{eqn: genericmassmomenta})$ hold then $\Xi_G(q,m)=0$ if and only if $G$ has no massive edges, and no incoming momenta (i.e., 
$q^{\{v\}}=0$ for all $v\in \partial  E^{\ext}_G$). 
\end{lem}

\begin{proof} Let $e \in E_G$ such that $m_e \neq 0$. Note that $\Phi_{G}(q)$ is of degree at most one  in $\alpha_e$.  If $G\backslash e$ is connected, then by 
 $(\ref{eqn: ContractionDeletion})$, the coefficient of $\alpha_e^2$ in $\Xi_G$ is $m_e^2 \Psi_{G\backslash e}$ which is non-zero by lemma $\ref{lem: PsiGvanishing}$.
 In particular, $\Xi_G(q,m)\neq 0$. In the opposite case, $e$ is a bridge in $G$, and $G\backslash e$ has two connected components $G_1, G_2$. Then $\Psi_G = \Psi_{G_1} \Psi_{G_2}$ and  by $(\ref{eqn: ContractionDeletion})$ the coefficient of $\alpha_e$ in $\Xi_G(q,m)$ is $((q^{G_1})^2 + m_e^2)\Psi_{G_1} \Psi_{G_2}$. This is non-zero
 by $(\ref{eqn: genericmassmomenta})$ and  so $\Xi_G(q,m) \neq 0$.  Finally, if all edges of $G$ are massless, then $\Xi_G(q,m) = \Phi_G(q)$ and we can appeal
 to the previous lemma. 
  \end{proof}

\begin{example} Consider the following Feynman graph 
\begin{center}
\fcolorbox{white}{white}{
  \begin{picture}(192,86) (35,-10)
    \SetWidth{1.0}
    \SetColor{Black}
    \Line[arrow,arrowpos=0.4,arrowlength=5,arrowwidth=2,arrowinset=0.2](48,70)(80,38)
    \Line[arrow,arrowpos=0.4,arrowlength=5,arrowwidth=2,arrowinset=0.2](48,54)(80,38)
    \Line[arrow,arrowpos=0.4,arrowlength=5,arrowwidth=2,arrowinset=0.2](48,22)(80,38)
    \Line[arrow,arrowpos=0.4,arrowlength=5,arrowwidth=2,arrowinset=0.2](48,6)(80,38)
     \Line[double,sep=2](80,38)(144,38)
    \Line[arrow,arrowpos=0.4,arrowlength=5,arrowwidth=2,arrowinset=0.2](176,70)(144,38)
    \Line[arrow,arrowpos=0.4,arrowlength=5,arrowwidth=2,arrowinset=0.2](176,54)(144,38)
    \Line[arrow,arrowpos=0.4,arrowlength=5,arrowwidth=2,arrowinset=0.2](176,22)(144,38)
    \Line[arrow,arrowpos=0.4,arrowlength=5,arrowwidth=2,arrowinset=0.2](176,6)(144,38)
    \Vertex(80,38){3}
    \Vertex(144,38){3}
    \Text(34,34)[lb]{\Large{\Black{$q_i$}}}
     \Text(56,34)[lb]{\Large{$\vdots$}}
    \Text(180,34)[lb]{\Large{\Black{$q_j$}}}
      \Text(164,34)[lb]{\Large{$\vdots$}}
    \Text(110,42)[lb]{\Large{\Black{$1$}}}
  \end{picture}
}
\end{center}
where the momenta entering on the left are $q_i$, for $i\in I$, and those on the right $q_j$, $j\in J$.
Then
 $\Xi_G(q,m) = \big( (\sum_{i\in I} q_i)^2 +m^2\big)\alpha_1$  is identically zero if 
 $(\sum_{i\in I} q_i)^2 +m^2=0$, where $m=m_1$.
  These examples
 imply, by contraction and deletion,  that the conditions $(\ref{eqn: genericmomenta})$ and $(\ref{eqn: genericmassmomenta})$ are optimal. 
  \end{example}

\subsection{Space of generic kinematics} \label{sect:  GenKin}
The previous discussion motivates the following definition. In order to allow the possibility of  masses and momenta taking values in different fields, we work with  affine spaces.  Suppose that we wish to consider
processes in $d\in 2\N$ spacetime dimensions with $Q$ external momenta $\underline{q} = (q_1,\ldots, q_Q)\in \A^{Qd}$, which are subject to momentum conservation 
\begin{equation} \label{eqn: spacegenmom} q_1+ \ldots + q_Q=0\ ,\end{equation}
 in a theory with $M$  possible non-zero particle masses $\underline{m}=(m_1,\ldots, m_M) \in \GG_m^M$, where $\GG_m = \A^1 \backslash \{0\}$. The graph polynomials and hence the  integral  $(\ref{eqn: projectiveFeynmanint})$  
are invariant under the action of the orthogonal group in $d$ dimensions. Therefore set
$$s_{i,j} = s_{j,i}= q_i . q_j \qquad \hbox{for } 1\leq i \leq j \leq Q$$
to be the Euclidean inner product of the momenta $q_i,q_j$, and write
$$s_I = \sum_{i,j\in I} s_{i,j} =( \sum_{i\in I} q_i)^2 \qquad \hbox{ for  }  I \subset \{1,\ldots, Q\}\ .$$
Condition $(\ref{eqn: spacegenmom})$ implies that the $s_{i,j} \in \A^{\binom{Q+1}{2}}$ lie in a subspace isomorphic to
$\A^{\binom{Q}{2}}$ since we can solve for $q_Q$. It is parameterised, for example, by the 
$s_{i,j}$ for $1\leq i\leq j \leq Q-1$.   The Feynman amplitude is a function only of the $s_{i,j}$ and $m_k$. 

\begin{defn} \label{defn:  OpenUKinematic} Define a space of generic kinematics 
\begin{equation}  \label{eqn:  KQMdefn}
K_{Q,M}^{\gen} \subset   K_{Q,M}=   \A^{\binom{Q}{2}} \times \GG_m^{M}
\end{equation}
with coordinates $(\underline{s}, \underline{m})$, to be the open complement of the  union of the spaces
\begin{equation} \label{eqn: finalgenericmomcondition}
  s_I + m^2_j =0    
   \end{equation}
for all $I \subsetneq \{1,\ldots, Q\}$ and $j\in \{0,1,\ldots, M\}$,  where we set $m_0=0$. Compare $(\ref{eqn: genericmomenta})$. It is an affine scheme defined over $\Z$ of dimension $M+\binom{Q}{2}$.
 Define the \emph{Euclidean region} to be its set of real points $K_{Q,M}^{\gen}(\R)$.
 Define  a  subspace 
$$ U_{Q,M}^{\gen} \subset K_{Q,M}^{\gen}(\C)$$
to be the open region (in the usual topology) of $K_{Q,M}^{\gen}(\C)$ defined by 
\begin{multline}  \label{Ugendef}
 U_{Q,M}^{\gen} = \{ (\underline{s}, \underline{m})  \in  K_{Q,M}^{\gen}(\C):  \ \qquad \mathrm{Re}\,  s_I >0 \hbox{ for all }   I \subsetneq \{1,\ldots, Q\} \ , \\ \!\!\!\! \hbox{and }   \quad \mathrm{Re}\, m^2_j >0  \quad \hbox{for all }  \  j\in \{1,\ldots, M\}\} \ .
 \end{multline}
The  region $U_{Q,M}^{\gen}$ contains the Euclidean region $K_{Q,M}^{\gen}(\R)$.
\end{defn}
Note that  $K_{1,M}^{\gen} = K_{0,M}^{\gen} = \GG_m^M$, and in particular $K_{1,0}^{\gen} = K_{0,0}^{\gen}= \Spec(\Z)$.

\begin{defn} Let  $k_{Q,M} $  denote the field of fractions of $\Or(K_{Q,M})$. 
It is isomorphic to $\Q((s_{i,j})_{1\leq i \leq j<Q}, (m_k)_{1\leq k \leq M})$. In particular, $k_{0,0} = \Q$. 
\end{defn}
\section{Partial factorization theorems} \label{sect: Section2}

The factorization theorems presented below are crucial to the construction of the cosmic Galois group. The so-called ultraviolet factorisations
are used in the theory of renormalisation, the infra-red factorisations below are apparently new.
\subsection{UV factorizations}
Let $G$ be a connected Feynman graph, and  let  $\gamma \subset E_G$ be an edge-subgraph   with connected components $\gamma_1,\ldots, \gamma_n$.

\begin{lem}  \label{lem: Bijectionspanningtrees} The map $T \ \mapsto \ ( T/(T\cap \gamma) ,  T\cap \gamma_1  ,   \ldots,   T\cap \gamma_n)$ 
is a bijection from:
$$\{\hbox{Spanning } k\hbox{-trees } T \hbox{ such that } \gamma_i\cap T \hbox{ is connected for all }i=1,\ldots, n\}$$
to
$$ \{\hbox{Spanning }k\hbox{-trees in } G/\gamma\}  \ \times\   \prod_{i=1}^n \{\hbox{Spanning trees in } \gamma_i\}  $$
\end{lem} 
\begin{proof} Let $T$ be any subgraph of $G$ such that $T\cap \gamma$ is a union of  trees. Then 
\begin{equation}  \label{eqninproof: hT}
h_{ T/(T\cap \gamma)} = h_T  \quad \hbox{ and  } \quad \kappa_{ T/(T\cap \gamma)} = \kappa_T\ .
\end{equation}
The first formula follows from $(\ref{eqn:  additivityofh})$, the second is clear.
If  $T$ is a spanning $k$-tree such that each $T\cap \gamma_i$ is  connected,  then each $T\cap \gamma_i$ is a spanning tree in $\gamma_i$  and 
 $T/(T\cap \gamma)$ is a spanning $k$-tree in $G/\gamma$  by  $(\ref{eqninproof: hT})$.

In the other direction, suppose that $S\subset G/\gamma$ is a spanning $k$-tree, and let $T_i\subset \gamma_i$ be spanning trees.
There is a unique subgraph $T$ of $G$
such that $T\cap \gamma_i= T_i$ and $T/(T\cap \gamma)=S$.   By $(\ref{eqninproof: hT})$, $h_T = h_S=0$ and $\kappa_T=\kappa_S=k$, and since $T$ meets every vertex of $G$  it follows that $T$ is a spanning $k$-tree.
\end{proof}

The following factorisation formulae are essentially well-known \cite{Angles}.

\begin{prop} Let $G$ be connected, $\gamma \subset E_G$ as above. Then
\begin{eqnarray} \label{eqn: UVfactorizations}
\Psi_G  &=&   \Psi_{\gamma} \Psi_{G/\gamma} + R^{\Psi}_{\gamma, G}   \\
\Phi_G(q) &=& \Psi_{\gamma}  \Phi_{G/\gamma}(q) + R^{\Phi,\UV}_{\gamma, G}(q)  \nonumber 
\end{eqnarray}
where  the degree of $R^{\Psi}_{\gamma,G}$ and $R^{\Phi, \UV}_{\gamma,G}(q)$  in the variables $\alpha_e$, $e\in E_{\gamma}$ is strictly greater than 
$$\deg  \Psi_{\gamma} = \deg  \prod_{i=1}^n \Psi_{\gamma_i} = h_{\gamma}\ .$$
\end{prop}
\begin{proof} We shall prove both formulae simultaneously. Let $k=0$ (resp.  $1$).  
By $(\ref{eqn: psidefn})$ and $(\ref{eqn: phidefn})$, the set of monomials in  $\Psi_G$ (resp.  $\Phi_G(q)$)
are in one-to-one correspondence with the set of spanning $k$-trees   $T\subset G$. 
The latter can be partitioned into two subsets: those for which   $T\cap \gamma_i$ is connected for all $i$,
and those for which $T\cap \gamma_i$ is  not connected for some $i$. The former class is in one-to-one correspondence, by  lemma $\ref{lem: Bijectionspanningtrees}$, with the monomials in 
$ \Psi_{\gamma_1}\ldots \Psi_{\gamma_n} \times \Psi_{G/\gamma}$ (resp. $ \Psi_{\gamma_1}\ldots \Psi_{\gamma_n} \times \Phi_{G/\gamma}(q)$). 
The latter class correspond to monomials in the remainder terms $R^{\Psi}_{\gamma, G}$ (resp. $R^{\Phi, \UV}_{\gamma, G}(q)$). To see this, 
observe by  $(\ref{eqn: Edgecomplementnumber})$ applied to $\gamma_i \cap T \subset \gamma_i$ that  for each $i$, 
$$N_{\gamma_i}- N_{T\cap \gamma_i}=  h_{\gamma_i} + \kappa_{T\cap \gamma_i} -1\ .$$ 
Thus the degree of the monomial $\prod_{e\notin T} \alpha_e$ in the variables $\alpha_e$ for $e\in E_{\gamma}$ is 
$$\sum_{i=1}^n ( N_{\gamma_i} - N_{T\cap \gamma_i} )= h_{\gamma}+   \sum_{i=1}^n  (  \kappa_{T\cap \gamma_i} -1)\ .$$
This is strictly greater than $h_{\gamma}$ whenever some $T\cap \gamma_i$ is not connected.
\end{proof}

Equivalently, setting $\alpha'_e = \lambda \alpha_e$ for $e\in E_{\gamma}$ and $\alpha'_e = \alpha_e$ otherwise, we have
\begin{eqnarray} 
 \Psi_G  (\alpha'_e)   & \equiv &   \lambda^{h_{\gamma}}   \Psi_{\gamma} (\alpha'_e)  \Psi_{G/\gamma} (\alpha'_e) \pmod{\lambda^{h_\gamma+1}}   \nonumber \\
 \Phi_G  (\alpha'_e)(q)   & \equiv &   \lambda^{h_{\gamma}}  \Psi_{\gamma} (\alpha'_e)  \Phi_{G/\gamma} (q)(\alpha'_e) \pmod{\lambda^{h_\gamma+1}}   \nonumber
  \end{eqnarray} 

\begin{rem} One can show \cite{Angles} that the formulae $(\ref{eqn: UVfactorizations})$, in the special case when $\gamma$ is a divergent subgraph,
are sufficient to recover some of the main theorems of the theory of renormalization. The full strength of the factorisation formulae $(\ref{eqn: UVfactorizations})$, for $\gamma$ an arbitrary subgraph, will only manifest itself in the motivic period.
\end{rem}

\subsection{IR factorizations}
Let $G$ be a connected Feynman graph.

With generic momenta $(\ref{eqn: genericmomenta})$, 
      $\gamma\subset E_G$ is momentum-spanning\footnote{If wants to consider non-generic
momentum configurations, one could take $(\ref{eqn:  PhiG/gammanvanishing})$ as  the definition of momentum-spanning. But in this case the factorisation theorems stated below will fail without some additional assumptions on momenta. See example \ref{ex:  Circle}.} if and only if 
\begin{equation} \label{eqn:  PhiG/gammanvanishing}
\Phi_{G/\gamma}(q) = 0\ .
\end{equation} 
This follows immediately from lemma  $\ref{lem: PhiGvanishing}$. In this situation, the second factorization formula $(\ref{eqn: UVfactorizations})$ is degenerate. It turns out that the remainder term $R^{\Phi, \UV}_{\gamma, G}$ can be further factorized  via 
the following  formula, which  is apparently new.
\begin{prop} Let $\gamma\subset E_G$ be a momentum spanning edge-subgraph. Then 
\begin{equation}\label{eqn: IRfactPhi}
\Phi_G(q) =   \Phi_{\gamma}(q)   \Psi_{G/\gamma} + R^{\Phi,\IR}_{\gamma,G}(q)
\end{equation}
where  the degree of  $R^{\Phi,\IR}_{\gamma,G}(q)$ in the variables $\alpha_e$, $e\in E_{\gamma}$ is strictly greater than 
$$  \deg \Phi_{\gamma}(q) = h_{\gamma} +1 \ .$$
\end{prop}
\begin{proof}  Suppose that $\gamma$ has  connected components $\gamma', \gamma_1, \ldots, \gamma_n$  such that $\gamma'$ is  momentum spanning.
Monomials in $\Phi_G(q)$ are in one-to-one correspondence with spanning 2-trees $T=T_1 \cup T_2$ such that $(q^{T_1})^2 \neq 0$.
For such a $2$-tree, $T\cap \gamma'$ cannot be connected because each component $T_i$  intersects $\gamma'$ non-trivially (otherwise,  $V_{\gamma'} \cap  V_{T_i} =\emptyset$ for some $i$, 
which implies that  $q^{T_{i}} = 0$ because $\gamma'$ is momentum-spanning).

Partition the set of spanning 2-trees such that $(q^T_1)^2\neq 0$ into two classes:  those such that  $T\cap\gamma'$  has 2 components and $T\cap \gamma_i$ is connected for all $i$ (call this class $C_1$), and those for which $T\cap \gamma'$ or some $T\cap \gamma_i$ has  strictly more components ($C_2$). 

There is a bijection from the first set $C_1$ to
$$\{\hbox{Spanning } 2\hbox{-trees in } \gamma'\} \times\{\hbox{Spanning trees in } G/\gamma\} \times \prod_i \{\hbox{Spanning trees in } \gamma_i\}\ . $$
It is given by the map
$$T \mapsto   ( T\cap \gamma'  \ ,  \  (T \cup \gamma)/\gamma \ , \  T\cap \gamma_1,\ \ldots,\  T\cap \gamma_n)\ .$$
The proof is similar to lemma $\ref{lem: Bijectionspanningtrees}$, on noting that  $(T\cup \gamma)/\gamma$ is the one-vertex join of $T_1/(T_1 \cap \gamma)$ and $T_2 / (T_2 \cap \gamma)$ along the vertex given by the image of $\gamma$, and has exactly one connected component. One checks  that 
given   a spanning $2$-tree $T' \subset \gamma'$, and spanning trees  $S \subset G/\gamma$ and $T_i \subset \gamma_i$, there is a unique graph  $T\subset G$ 
such that $T\cap \gamma'= T'$, $(T\cup \gamma) /\gamma = S$ and $T_i = T\cap \gamma_i$, 
and that it   has exactly two connected components.

This gives a one-to-one correspondence between the set $C_1$ and monomials in 
$$\Phi_{\gamma}(q) = \big( \Phi_{\gamma'}(q)\prod_{i=1}^n \Psi_{\gamma_i} \big)  \Psi_{G/\gamma} \ .$$
Spanning $2$-trees  $T$ in the set $C_2$ are such that $T\cap \gamma'$ has at least 3 components, or some $T\cap \gamma_i$ has at least 2 components.
In this case,  the degree of the monomial $\prod_{e\notin T} \alpha_e$ in the variables $\alpha_e$ for $e\in E_{\gamma}$ is, by equation $(\ref{eqn: Edgecomplementnumber})$ applied to $T\cap \gamma' \subset \gamma'$ and $T\cap \gamma_i \subset \gamma_i$
$$N_{\gamma'} - N_{T\cap \gamma' } + \sum_{i=1}^n (N_{\gamma_i} - N_{T \cap \gamma_i }) = h_{\gamma}+  (\kappa_{T\cap \gamma'}-1)  + \sum_{i=1}^n  (  \kappa_{T\cap \gamma_i} -1)\ ,$$
which is strictly greater than $h_{\gamma}+1$, and contributes to $R^{\Phi,\IR}_{\gamma,G}(q)$.
\end{proof}

One can derive the contraction-deletion relations $(\ref{eqn: ContractionDeletion})$ from the factorizations $(\ref{eqn: UVfactorizations})$  and $(\ref{eqn: IRfactPhi})$ by setting $\gamma=e$ in the former and $\gamma = G\backslash e$ in the former and latter. 

\begin{example} \label{ex:  Circle} (Degenerate momenta).
 Consider the following Feynman graph
\begin{center}
\fcolorbox{white}{white}{
  \begin{picture}(110,97) (166,-58)
    \SetWidth{1.0}
    \SetColor{Black}
    \Arc(220,-10)(32,263,623)
    \Line[arrow,arrowpos=0.5,arrowlength=5,arrowwidth=2,arrowinset=0.2](165,35)(195,10)
    \Line[arrow,arrowpos=0.5,arrowlength=5,arrowwidth=2,arrowinset=0.2](165,-55)(195,-30)
    \Line[arrow,arrowpos=0.5,arrowlength=5,arrowwidth=2,arrowinset=0.2](275,35)(245,10)
    \Line[arrow,arrowpos=0.5,arrowlength=5,arrowwidth=2,arrowinset=0.2](275,-55)(245,-30)
    \Vertex(195,10){2}
    \Vertex(245,10){2}
    \Vertex(245,-30){2}
    \Vertex(195,-30){2}
      \Text(170,15)[lb]{{\Black{$q_1$}}}
              \Text(265,15)[lb]{{\Black{$q_2$}}}
            \Text(265,-40)[lb]{{\Black{$q_3$}}}
         \Text(170,-40)[lb]{{\Black{$q_4$}}}
           \Text(217,25)[lb]{{\Black{$1$}}}
              \Text(256,-13)[lb]{{\Black{$2$}}}
   \Text(180,-13)[lb]{{\Black{$4$}}}
               \Text(217,-52)[lb]{{\Black{$3$}}}
  \end{picture}
}
\end{center}
 which satisfies 
 $$\Phi_G(q) = (q_2+q_3)^2 \alpha_1\alpha_3 + (q_1+q_2)^2 \alpha_2 \alpha_4 + q_1^2 \alpha_1 \alpha_4 + q_2^2 \alpha_1 \alpha_2+ q_3^2 \alpha_2\alpha_3 +q_4^2 \alpha_3\alpha_4$$
 and impose the condition 
 $(q_2+q_3)^2 = (q_1+q_4)^2=0$. In this case the subgraph $\gamma$ defined  by the two edges $2$ and $4$ satisfies  $\phi_{G/\gamma}(q)=0$, but is not momentum spanning according to our stricter definition, because the incoming momenta do not all lie in the same connected component.  To leading order in the subgraph variables $\alpha_2, \alpha_4$ we have
 $$\Phi_G(q) = \alpha_2 ( q_3^2 \alpha_3+q_2^2 \alpha_1) + \alpha_4 (q_1^2 \alpha_1 +  q_4^2 \alpha_3) + R $$ 
 where $R = (q_1+q_2)^2 \alpha_2 \alpha_4$, and the leading terms do not factorize. If however, one further imposes the conditions
 $q_2+q_3=0$ and $q_1+q_4=0$ (so that the subgraph $\gamma$ now satisfies momentum conservation in each connected component), we obtain
 $$ \Phi_G(q) = (q^2_2 \alpha_2 + q^2_4 \alpha_4)(\alpha_1+\alpha_3) + R $$
 and a factorization  formula for  the leading term is restored.
  These types of phenomena  suggest our results  generalise, but  will not be considered in these notes.
\end{example}

\subsection{Factorization formulae for $\Xi$} A UV-factorization formula for $\Xi_G(q,m)$ follows immediately  from $(\ref{eqn: UVfactorizations})$.
The IR-factorization formula requires a further constraint on the distribution of masses.
\begin{defn} A subgraph $\gamma\subset G$ is  \emph{mass-spanning}  if it contains all massive edges of $G$: for every edge $e\in E_G$ such that $m_e \neq 0$,   $e\in E_{\gamma}$.

We shall say that a subgraph $\gamma$ is \emph{mass-momentum spanning} (or simply $m.m.$ for short) if it is both mass and momentum-spanning.
\end{defn}

 For generic kinematics $(\ref{eqn: genericmassmomenta})$, 
a subgraph  $\gamma\subset E_G$ satisfies
\begin{equation}  \label{eqn:  mmequivalentXi}\gamma \hbox{ is } m.m.  \quad \Longleftrightarrow  \quad 
\Xi_{G/\gamma}(q,m)=0 \ .
\end{equation}
This is a direct consequence of  lemma $\ref{lem:  Xivanishing}$.

\begin{thm} Let $G$ be a connected Feynman graph, and let  $\gamma\subset E_G$ be an edge-subgraph with any number of connected components. Then
\begin{equation} \label{eqn: XiUVfact}
\Xi_G(q,m) =    \Psi_{\gamma} \,  \Xi_{G/\gamma}(q,m) + R^{\Xi, \UV}_{\gamma, G} (q,m)
\end{equation} 
where $R^{\Xi, \UV}_{\gamma, G}(q,m)$ has degree $> h_\gamma$ in the $\alpha_e$, $e\in E_{\gamma}$.
Now suppose that $\gamma$ is a mass-momentum subgraph. In this case,  
\begin{equation} \label{eqn: XiIRfact}
\Xi_G(q,m) =      \Xi_{\gamma}(q,m)\,   \Psi_{G/\gamma}+ R^{\Xi, \IR}_{\gamma, G} (q,m)
\end{equation} 
where $R^{\Xi, \IR}_{\gamma, G}(q,m)$ has degree $> h_\gamma+1$ in the $\alpha_e$, $e\in E_{\gamma}$.
\end{thm} 
\begin{proof} For the proof of $(\ref{eqn: XiUVfact})$ combine $(\ref{eqn: UVfactorizations})$ with the definition $(\ref{eqn:  Xidefn})$ and set
$$R^{\Xi,\UV}_{\gamma, G}(q,m) = R^{\Phi,\UV}_{\gamma, G}(q) + (\sum_{e \in E_{\gamma}} m_{e}^2 \alpha_e) \Psi_{\gamma}\Psi_{G/\gamma} + (\sum_{e \in E_G} m_e^2 \alpha_e) R_{\gamma,G}^{\Psi}$$
For $(\ref{eqn: XiIRfact})$, combine the factorization formula for $\Psi$ with the $\IR$-factorization formula $(\ref{eqn: IRfactPhi})$,  use the condition  $m_e\neq 0 \Rightarrow e\in E_{\gamma}$, and set 
$$R^{\Xi, \IR}_{\gamma, G}(q,m) = R^{\Phi, \IR}_{\gamma, G}(q)  + (\sum_{e \in E_G} m_e^2 \alpha_e)   R_{\gamma,G}^{\Psi}\ .$$
The degree of $R^{\Xi,\IR}_{\gamma, G}$ is indeed of degree $>h_{\gamma}+1$ in the variables $\alpha_e$, for $e\in E_{\gamma}$. 
\end{proof}

Note that the factorisation formula for $\Psi_G$, which is symmetric with respect to $\gamma$ and $G/\gamma$, occurs in both the UV and IR-factorizations of $\Xi_G$. 

\section{\Motic subgraphs} \label{sect: Section3}

For want of a better adjective, the invented word $\motic$\!\! will be used to define a class of  subgraphs of a Feynman graph. It pertains to the word mote, meaning a speck or particle (leading to notions of indivisibility), and the fact that its letters stand for `members of the inner circle', leading to the idea of connectedness. The \emph{motive} of a graph will be constructed out of its \motic\!\! subquotients.

\subsection{Definition of \motic subgraphs}
Let $G$ be a  Feynman graph. Recall that at most one component of $G$ carries non-trivial kinematics.
\begin{defn} An edge-subgraph $\Gamma \subset G$ is \motic if, for every edge-subgraph $\gamma \subsetneq \Gamma$ which is mass-momentum spanning in $\Gamma$, one has 
$h_{\gamma} < h_{\Gamma}.$
\end{defn}

Recall from \S\ref{sect: FeynmanGraphs} that any edge subgraph $\Gamma \subset G$  which is mass-momentum spanning inherits all masses and external momenta from $G$. When it is not mass-momentum spanning then it is considered to be a Feynman graph with zero internal masses and no external momenta. Every  subgraph of such a graph is  trivially mass-momentum spanning.  In particular, if $G$ has no  kinematics
$$ \Gamma \subset G \hbox{ \motic } \quad \Longleftrightarrow \quad \Gamma \hbox{ is } 1  \hbox{-particle irreducible} \ .$$
Another example of a motic subgraph is a minimal mass-momentum spanning subgraph $\Gamma \subset G$ (related to the notion of `infra-red' graph in \cite{Speer}).

\begin{rem} \label{rem: moticbyedgecutting}
An edge subgraph $\Gamma \subset G$ is \motic if (and only if)  every mass-momentum spanning edge subgraph  of $\Gamma$ of the form $\Gamma \backslash e$, where  $e\in E_{\Gamma}$,  satisfies $h_{\Gamma \backslash e} < h_{\Gamma}$. To see that a subgraph with this property is indeed \motic\!\!, let $\gamma \subset \Gamma$ be any edge subgraph which is mass-momentum spanning in $\Gamma$, and   choose $e \in E_{\Gamma} \backslash E_{\gamma}$. Then $ \Gamma \backslash e$ is also mass-momentum spanning in $\Gamma$ and contains $\gamma$, so we have $h_{\gamma} \leq h_{\Gamma \backslash e} < h_{\Gamma}$.  Thus a graph is \motic when cutting an edge either causes the loop number to drop, or breaks the property of being mass-momentum spanning.

\end{rem} 
It follows from the definition that a subgraph $\Gamma \subset G$ is \motic if and only if every connected component of $\Gamma$ is  a \motic subgraph of  $G$.

\begin{example}
Below are the six \motic subgraphs of example of $\ref{example: Dunce}$. All subgraphs are mass-momentum spanning and give rise to an infra-red (and possibly also ultra-violet) sub-divergence, except for the final subgraph given by the edges $3,4$,
which corresponds to a purely ultra-violet sub-divergence. 
\begin{center} 
\fcolorbox{white}{white}{
  \begin{picture}(292,105) (-30,-5)
    \SetWidth{1.0}
    \SetColor{Black}
    \Line[arrow,arrowpos=0.5,arrowlength=5,arrowwidth=2,arrowinset=0.2](-42,51)(-20,51)
    \Vertex(-20,51){3}
    \Line[double,sep=2](-20,51)(28,83)
    \Line(-20,51)(28,19)
    \Arc(97.643,51)(76.643,155.322,204.678)
    \Arc[clock](-32,51)(68,28.072,-28.072)
    \Vertex(28,19){3}
    \Vertex(28,83){3}
    \Line[arrow,arrowpos=0.5,arrowlength=5,arrowwidth=2,arrowinset=0.2](48,98)(28,83)
    \Text(0,71)[lb]{{\Black{$1$}}}
    \Text(0,26)[lb]{{\Black{$2$}}}
    \Text(14,48)[lb]{{\Black{$3$}}}
    \Text(39,48)[lb]{{\Black{$4$}}}
       \Line[arrow,arrowpos=0.5,arrowlength=5,arrowwidth=2,arrowinset=0.2](58,51)(80,51)
    \Vertex(80,51){3}
    \Line[double,sep=2](80,51)(128,83)
    \Line(80,51)(128,19)
    \Arc(197.643,51)(76.643,155.322,204.678)
    \Vertex(128,19){3}
    \Vertex(128,83){3}
    \Line[arrow,arrowpos=0.5,arrowlength=5,arrowwidth=2,arrowinset=0.2](148,98)(128,83)
    \Text(100,71)[lb]{{\Black{$1$}}}
    \Text(100,26)[lb]{{\Black{$2$}}}
    \Text(126,48)[lb]{{\Black{$3$}}}
       \Line[arrow,arrowpos=0.5,arrowlength=5,arrowwidth=2,arrowinset=0.2](158,51)(180,51)
    \Vertex(180,51){3}
    \Line[double,sep=2](180,51)(228,83)
    \Line(180,51)(228,19)
    \Arc[clock](168,51)(68,28.072,-28.072)
    \Vertex(228,19){3}
    \Vertex(228,83){3}
    \Line[arrow,arrowpos=0.5,arrowlength=5,arrowwidth=2,arrowinset=0.2](248,98)(228,83)
    \Text(200,71)[lb]{{\Black{$1$}}}
    \Text(200,26)[lb]{{\Black{$2$}}}
    \Text(228,48)[lb]{{\Black{$4$}}}
  \end{picture}
}

\end{center}  
 
\begin{center} 
\fcolorbox{white}{white}{
  \begin{picture}(292,105) (-60,-5)
    \SetWidth{1.0}
    \SetColor{Black}
    \Line[arrow,arrowpos=0.5,arrowlength=5,arrowwidth=2,arrowinset=0.2](-42,51)(-20,51)
    \Vertex(-20,51){3}
    \Line[double,sep=2](-20,51)(28,83)
    \Arc(97.643,51)(76.643,155.322,204.678)
    \Arc[clock](-32,51)(68,28.072,-28.072)
    \Vertex(28,19){3}
    \Vertex(28,83){3}
    \Line[arrow,arrowpos=0.5,arrowlength=5,arrowwidth=2,arrowinset=0.2](48,98)(28,83)
    \Text(0,71)[lb]{{\Black{$1$}}}
    \Text(14,48)[lb]{{\Black{$3$}}}
    \Text(39,48)[lb]{{\Black{$4$}}}
       \Line[arrow,arrowpos=0.5,arrowlength=5,arrowwidth=2,arrowinset=0.2](58,51)(80,51)
    \Vertex(80,51){3}
    \Line[double,sep=2](80,51)(128,83)
    \Vertex(128,83){3}
   \Line[arrow,arrowpos=0.5,arrowlength=5,arrowwidth=2,arrowinset=0.2](148,98)(128,83)
    \Text(100,71)[lb]{{\Black{$1$}}}
    \Arc(247.643,51)(76.643,155.322,204.678)
   \Arc[clock](118,51)(68,28.072,-28.072)
    \Vertex(178,19){3}
    \Vertex(178,83){3}
    \Text(164,48)[lb]{{\Black{$3$}}}
    \Text(189,48)[lb]{{\Black{$4$}}}
  \end{picture}
}

\end{center}  
 For motivation, the reader may like to check that the graph polynomial $\Xi_G(q,m)$ vanishes on setting $\alpha_e=0$ for all $e\in E_{\gamma}$,
 for these \motic subgraphs $\gamma\subset G$.

 \end{example}

\subsection{Properties of \motic subgraphs}
Throughout this section, we use the abbreviation  $m.m.$  to stand for mass-momentum spanning. 
\begin{lem}  \label{lem: propertiesofmm} Let $\alpha \subset \beta \subset G$ be edge subgraphs.

$(i)$. $\alpha$ is m.m. in $\beta$, and $\beta$ is m.m. in $G  \Longleftrightarrow \alpha$ is m.m. in $G$.

$(ii)$.  $\beta$ is m.m. in $G \Longleftrightarrow \beta/\alpha$ is m.m. in $G/\alpha$. 

\end{lem} 

\begin{proof}
Part $(i)$ is clear.
For $(ii)$, observe that $(G/\alpha)/(\beta/\alpha) \cong G/\beta$ and hence $\Xi_{G/\beta} = \Xi_{(G/\alpha)/(\beta/\alpha)}$. Now apply $(\ref{eqn:  mmequivalentXi})$. 
\end{proof}
Note that the intersection of two $m.m.$ subgraphs is not necessarily $m.m.$ (in example \ref{ex:  Circle}, consider the edge subgraphs spanned by edges $1,2,3$ and $2,3,4$.)

\begin{rem}  \label{rem: moticintrinsic} The definition of a \motic subgraph is  intrinsic in the following sense.
If $ H \subset G$ is an edge subgraph, and  $\gamma \subset H$, then $\gamma$ is \motic in $G$ if and only if it is \motic in $H$.
This follows immediately from the definition if $H$ is $m.m.$ in $G$, by  lemma $\ref{lem: propertiesofmm}$ $(i)$, 
since $\alpha \subset \gamma$ is $m.m.$  in $H$ if and only if it is $m.m.$ in $G$. 
In the case when $H$ is not $m.m.$ in $G$, then neither is $\gamma \subset H$ by the same lemma. 
It is  \motic if and only if it is 1-particle irreducible, which is an intrinsic property.
\end{rem} 

The main properties of \motic subgraphs are summarised below.
\begin{thm} \label{thm: moticproperties} \Motic graphs have the following properties.  Let $G$ be a  Feynman graph and let 
$\alpha, \beta \subset G$ be edge subgraphs.

\vspace{0.02in}

 $(i)$. (Quotients) If $\beta$ is \motic in $G$,   then  $(\beta \cup \alpha)/\alpha  $ is  \motic in $G/\alpha$.

  $(ii).$ (Extensions) Let $\alpha \subset \beta$. If $\alpha$ is \motic in $G$ and $\beta/\alpha$  is \motic    in $G/\alpha$,  then $\beta$ is \motic in $G$.

\vspace{0.02in}

$(iii).$ (Unions) If $\alpha,\beta \subset G$ are  \motic subgraphs then $\alpha \cup \beta \subset G$ is \motic\!\!.

\vspace{0.02in}

$(iv).$ (Contraction of edges)  Let $e\in E_G$.  If $(\alpha \cup e)/e$ is   \motic in $G/e$,  then at least one of $\alpha$ or $\alpha \cup e $ is \motic in $G$.   Thus there is a surjective map 
$$\alpha \mapsto (\alpha \cup e)/e: \{\motic \hbox{subgraphs of } G\} \To \{\motic \hbox{subgraphs  of } G/e\}\ .$$
It is not injective:  it can happen that both $\alpha$ and $\alpha \cup e$  are \motic\!\!. 
\end{thm} 

\begin{proof}  $(i)$. First consider the case when $\alpha \subset \beta$, and let $\beta$ be \motic in $G$. There is a one-to-one correspondence 
\begin{eqnarray} 
\{\hbox{Edge subgraphs }\gamma \hbox{ s.t. }  \alpha\subset \gamma \subset \beta\}  & \leftrightarrow & \{\hbox{Edge subgraphs of } \beta/\alpha\} \nonumber \\
 \gamma & \mapsto & \gamma/ \alpha \nonumber
 \end{eqnarray}
 By lemma $\ref{lem: propertiesofmm}$ $(ii)$, this bijection preserves the subset of $m.m.$ subgraphs. Now by $(\ref{eqn:  additivityofh})$, 
 we have $h_{\beta}= h_{\alpha}+ h_{\beta/\alpha}$ and $h_{\gamma} = h_{\alpha} + h_{\gamma/\alpha}$, whence
 $$h_{\beta} - h_{\gamma} =  h_{\beta/\alpha} - h_{\gamma/\alpha}\ .$$
 Thus $\gamma/\alpha \subset \beta/\alpha$ is strict and $m.m.$ if and only if $\gamma \subset \beta$ is. In this case the left-hand side is 
 strictly positive and $ h_{\beta/\alpha} > h_{\gamma/\alpha}$. This proves that $\beta/\alpha$ is \motic in $G/\alpha$.

Now consider the general case when $\beta\subset G$ is \motic but  does not necessarily contain $\alpha$.
The quotient $G/\alpha$ is obtained by successively contracting edges in $e\in E_{\alpha}$. For every such edge which is an edge of $\beta$,
we can invoke the case proved above.  If $e$ has no common vertices with $\beta$, or a single vertex in common with $\beta$, then the image
$\overline{\beta}= (\beta \cup e)/ e$ of $\beta $ in $G/e$ is isomorphic to $\beta$, and the proof is straightforward. 
It remains to consider the  case when $e$ meets $\beta$ in two vertices $v_1,v_2 \in V_{\beta}$, and hence 
  $\overline{\beta}$  is obtained from $\beta$ by identifying $v_1$ and $v_2$.  We wish to show that $\overline{\beta}$ is \motic\!\!. 
Since $E_{\beta} = E_{\overline{\beta}}$ there is a one-to-one correspondence  $\gamma \leftrightarrow \overline{\gamma}$ between edge subgraphs of $\beta$ and those of $\overline{\beta}$.
First of all, suppose that  $v_1, v_2$ lie in two distinct connected components of $\beta$. Euler's formula $(\ref{eqn: Eulerformula})$ implies that $h_{\overline{\beta}}= h_{\beta}$.  Since
$\beta$ is \motic\!\!, all momentum-bearing vertices lie in a single connected component of $\beta$. Thus 
a strict  subgraph $\overline{\gamma} \subsetneq \overline{\beta}$ is $m.m.$ if and only if $\gamma \subsetneq \beta$ is $m.m.$ and we have
$$h_{\overline{\gamma}} = h_{\gamma} < h_{\beta} = h_{\overline{\beta}}$$
using the fact that $\beta$ is \motic\!\!. This proves that $\overline{\beta}$ is \motic also.  Now suppose that $v_1,v_2$ lie in the same connected component of $\beta$, and hence $h_{\overline{\beta}}= h_{\beta}+1$ by $(\ref{eqn: Eulerformula})$. 
Let 
 $\overline{\gamma} \subsetneq \overline{\beta}$  be an  $m.m.$ edge-subgraph. If  $\gamma$ is $m.m.$ in $\beta$,  we deduce that
 $$h_{\overline{\gamma}} \overset{(\ref{eqn: Eulerformula})}{\leq} h_{\gamma}+1 < h_{\beta} + 1 =  h_{\overline{\beta}}\ ,$$
 since the middle inequality follows from the fact that 
 $\beta$ is \motic\!\!. 
Now suppose that $\gamma$ is not $m.m.$ in $\beta$. Since $\gamma$ is  mass-spanning, 
the set of momentum bearing vertices of $\beta$ do not all lie in the same connected component of $\gamma$. Since  $\overline{\gamma}$ is $m.m.$,  their images  in $\overline{\gamma}$  lie in the same connected component, and we have
$\kappa_{\gamma } > \kappa_{\overline{\gamma} } $ 
and hence $h_{\overline{\gamma}} =h_{\gamma} $ by $(\ref{eqn: Eulerformula})$.  But then $h_{\overline{\gamma}} =h_{\gamma }\leq h_{\beta} < h_{\overline{\beta}}$ since $h_{\overline{\beta}}= h_{\beta}+1$.  In both cases we have shown that 
$h_{\overline{\gamma}}< h_{\overline{\beta}}$, which proves that $\overline{\beta}$ is \motic\!\!.

$(ii)$.  Let $\alpha \subset \beta$ be \motic  and $\beta/\alpha \subset G/\alpha$ be \motic also. Let $\gamma \subsetneq \beta$ be a strict $m.m.$ edge subgraph. Denote its image  in $\beta /\alpha $ by 
$$\overline{\gamma}= (\gamma\cup \alpha)/\alpha\ . $$
By lemma $\ref{lem: propertiesofmm}$ $(i)$,  $\gamma$ and $\gamma \cup \alpha$ are $m.m.$ in $\beta$ and hence by 
lemma $\ref{lem: propertiesofmm}$ $(ii)$, $\overline{\gamma}$ is $m.m.$ in $\beta/\alpha$.   By $(\ref{eqn:  additivityofh})$, we have
$$h_{\gamma} = h_{\gamma \cap \alpha} + h_{\gamma / (\gamma\cap \alpha)}  \ \leq\  h_{\gamma \cap \alpha} + h_{\overline{\gamma}} \  \leq  \ h_{\alpha} + h_{\beta/\alpha} = h_{\beta}$$
The first inequality holds because $\overline{\gamma}$ is obtained from $\gamma / (\gamma\cap \alpha)$ by identifying vertices and therefore 
$h_{\gamma / (\gamma\cap \alpha)}\leq h_{\overline{\gamma}}$. Suppose by contradiction that $h_{\gamma}= h_{\beta}$.  Then 
$$h_{\overline{\gamma}}= h_{\beta/\alpha} \quad , \quad h_{\gamma\cap \alpha} = h_{\alpha} \quad \ , \quad  h_{\gamma / (\gamma\cap \alpha)}=h_{\overline{\gamma}}\ .$$
Since $\beta/\alpha$ is \motic and $\overline{\gamma} \subset \beta /\alpha$ is $m.m.$,  the first equality implies that $\overline{\gamma}=\beta/\alpha$.  
Suppose that  $\alpha \cap \gamma$ is $m.m.$ in $\alpha$. Then the second equality would imply that 
$\alpha \cap \gamma = \alpha$, since $\alpha$ is \motic\!\!, and this, together with $\overline{\gamma}=\beta/\alpha$, 
contradicts the fact that $\gamma$ is a strict subgraph of $\beta$. 

Now consider the case when  $\alpha \cap \gamma$ is $m.m.$ in $\alpha$. 
 If  $\alpha$ is not $m.m.$ in $G$, then  every subgraph of $\alpha$ is automatically $m.m.$ in $\alpha$ and there is nothing to prove. 
 Therefore suppose that $\alpha$ is $m.m.$ in $G$, and hence in $\beta$.  
   Consider the set $Q$ of momentum-bearing vertices in $\beta$, and suppose that they lie in $k \geq 1$ different connected components of $\alpha \cap \gamma$. It suffices to show that $k=1$, since in that case $\alpha\cap \gamma$ is $m.m.$ from the definition (as both $\gamma$ and $\alpha$ are $m.m.$ in $\beta$). Since $\gamma$ is momentum-spanning, the  image of  $Q$ in $\gamma/ (\gamma \cap \alpha)$  consists of exactly $k$ vertices lying  in a single connected component. On the other hand,  since $\alpha$ is momentum-spanning, the image of $Q$ in $\overline{\gamma}$ reduces to a single vertex.  So $\overline{\gamma}$ is obtained from $\gamma /(\gamma \cap \alpha)$ by identifying these $k$ connected vertices, and possibly identifying further vertices. If $k>1$ we would have   $h_{\gamma / (\gamma\cap \alpha)}< h_{\overline{\gamma}}$ by $(\ref{eqn: Eulerformula})$ contradicting the third equality  above. Hence $k=1$, as required. 
  
     $(iii)$.  By $(i)$, the graph   $ \overline{\beta}=(\alpha \cup \beta) /\alpha $ is \motic in $G/\alpha$.  Since  $\alpha \subset \alpha \cup \beta$,  
     it follows from the extension property $(ii)$ that $\alpha \cup \beta$ is \motic in $G$.

 $(iv)$. Let   $\alpha \subset G$ be an edge subgraph, and  $e \in E_G$. Let   $\Gamma =\alpha \cup e$
 and suppose that $\Gamma/e$ is \motic in $G/e$.  Suppose that $e$ is not a tadpole.
   By the proof of $(i)$ above,  there is a one-to-one correspondence between $m.m.$ subgraphs $\gamma\subset \Gamma$ which 
 contain the edge $e$ and $m.m.$ subgraphs $\gamma/e$ of $\Gamma/e$.  If $\gamma\subsetneq \Gamma$ is strict and contains the edge $e$, then 
 $h_{\gamma} = h_{\gamma/e}< h_{\Gamma/e} = h_{\Gamma}$ by $(\ref{eqn:  additivityofh})$. The strict inequality in the middle follows since $\Gamma/e$
 is \motic\!\!.  By remark \ref{rem: moticbyedgecutting}, $\Gamma$ will be \motic if 
  $\Gamma \backslash e$ is not $m.m.$, or if it is $m.m.$ and $h_{\Gamma \backslash e} < h_{\Gamma}$.

   Therefore let us suppose that $\Gamma \backslash e$ is $m.m.$  and $h_{\Gamma \backslash e}= h_{\Gamma}$.  The latter equality implies that $e$ is a bridge in $\Gamma$, so we can write $\Gamma \backslash e = \Gamma_1 \cup \Gamma_2$, where  $\Gamma_1,\Gamma_2$ are disjoint.
   The quotient $\Gamma /e$ is the one-vertex join $\Gamma_1 . \Gamma_2$. Now there is a bijection between the subgraphs  $\gamma_1 \cup \gamma_2$
   of $\Gamma_1 \cup \Gamma_2$ and $\gamma_1 . \gamma_2$ of $\Gamma_1. \Gamma_2$. This bijection preserves the number of loops and the  property of being $m.m.$.  The momentum spanning property  follows from the fact that  $\Gamma\backslash e$ is momentum-spanning by assumption, so all momenta
   flow into only one of the parts $\Gamma_i$. Thus if $\Gamma_1 . \Gamma_2$ is \motic if and only if $\Gamma_1 \cup \Gamma_2$ is. Since $\Gamma/e$ is \motic
   we deduce that $\Gamma \backslash e$ is. It remains to consider the case when $e$ is a tadpole. Then $\Gamma\backslash e = \Gamma/e$ is
   \motic (so too is $\Gamma$).
              \end{proof}

\section{The \motic Hopf algebra of graphs} \label{sect:  moticHopfalgebra}
The notion of \motic subgraph gives rise to a Hopf algebra which encapsulates both infra-red and ultra-violet divergences of Feynman graphs. 

\begin{defn} Let $\FF$  denote the free  $\Z$-module generated by disjoint unions of \motic\!\!\footnote{$G$ is \motic if it is \motic as a subgraph of itself.} Feynman graphs, modulo the equivalence relation $G\sim G \cup \{v\}$ where $v$ is an isolated vertex. It is bigraded by 
$$(\hbox{number of edges}, \hbox{number of loops})\ .$$ \end{defn}

The disjoint union of graphs defines a 
  commutative ring structure on $\FF$, whose unit is the empty graph $1$. 
Define  a coproduct on $\FF$ by the formula:
\begin{eqnarray}
 \Delta : \FF &  \To&  \FF  \otimes_{\Z} \FF   \\
   G & \mapsto & \sum_{\gamma\subset E_G } \gamma \otimes G/\gamma \nonumber
 \end{eqnarray} 
 where $G$ is a connected graph and   the sum is over all \motic subgraphs $\gamma$ of $G$. This is a generalisation of the  Connes-Kreimer coproduct for scalar graphs. The map $\Delta$ extends to a unique
 homomorphism on  $\FF$.

\begin{thm} $\FF$   is a connected bigraded Hopf algebra.
\end{thm}
\begin{proof} 
The coassociativity of $\Delta$ is a consequence of   properties $(i)$ and $(ii)$ of  theorem  $\ref{thm: moticproperties}$ by a standard argument (see, for example, \cite{BlochKreimer}).  The augmentation map $\varepsilon$ is the projection $\FF \rightarrow \FF_0 \cong \Z$. Denote its kernel by $I$. A generator  $G$ is \motic\!\!, so 
$\Delta (G) = 1\otimes G + G \otimes 1 \pmod {I \otimes \FF + \FF \otimes I}$. From this follows the 
equations $(\varepsilon \otimes \id)\Delta =  \id $ and $(\id \otimes \varepsilon )\Delta = \id $.
The antipode $S$ is constructed by the usual recursive formula for the antipode in a connected graded commutative Hopf algebra, and is defined over $\Z$.
The fact that the loop number and edge number are gradings follows from $(\ref{eqn:  additivityofh})$ and $N_{\gamma} + N_{G/\gamma} = N_G$. 
\end{proof} 
The grading by loop number is of importance when considering  the geometry of graphs in momentum space and the perturbative expansion, but  the edge number grading will be more relevant in these notes,  since we shall focus on the geometry of graphs in parametric space and their corresponding motives. 

Let us denote by $\FF_{Q,M}$ the free submodule of $\FF$ generated by \motic  Feynman graphs of type $(Q,M)$, i.e., with $Q$ external 
momenta and $M$ possible non-zero  masses. Then $\FF_{Q,M}$ is a module over $\FF_{0,0}$, where multiplication is that of  $\FF$:
\begin{equation}
\label{equationff00mult} 
\FF_{0,0} \otimes_{\Z} \FF_{Q,M} \To \FF_{Q,M}\ .
\end{equation}
 In particular, $\FF_{0,0}$ is a commutative ring. 
An important feature of the \motic  coproduct is the `all or nothing' property of masses and momenta:
\begin{equation}\label{Deltaallornothing} 
\Delta  \FF_{Q,M} \subset \big(\FF_{Q,M} \otimes_{\Z} \FF_{0,0}\big) \  \oplus  \ \big(\FF_{0,0} \otimes_{\Z} \FF_{Q,M}\big) \ .
\end{equation} 
The terms landing in the first factor correspond to subgraphs which are $m.m.$, and the terms in the second factor are those which are not $m.m.$.
In particular, $\FF_{0,0} \subset \FF$ is a Hopf subalgebra. It was  defined in  \cite{BlochKreimer} and 
 called the \emph{core} Hopf algebra.

There is a variant of this construction in which one considers graphs whose edges have distinct labels. We leave the details to the reader. 

\begin{rem} A more fundamental structure underlying the space of graphs should be the structure of an operad. For graphs of type $(0,0)$ it is clear how to define this by insertion into vertices of graphs, but is more delicate for graphs with masses and momenta. See \cite{FeynCat} for some related categorical notions. 
\end{rem} 

\subsection{Coradical filtration} The coradical filtration $C_i \FF$ is defined as follows.
Let $$\Delta'= \Delta - 1\otimes \id - \id \otimes 1$$ denote  the reduced coproduct.  Then $C_0 \FF=\Z$ and  
  $x\in C_n \FF$, for $n\geq 1$, if and only if $(\Delta')^{n} x=0$. The space $C_1 \FF$ consists of  primitive  linear combinations of (unions of) graphs. 

\begin{lem} \label{lem: coradicaldegree} The coradical filtration of a \motic  Feynman graph $G$ is 
$\leq h_G$ if $G$ is of type $(0,0)$ and  $ \leq h_G+1$ otherwise.
\end{lem} 

\begin{proof}
A \motic  graph of type $(0,0)$ is 1PI and hence satisfies $h_{\gamma}\geq 1$ (for every edge $e$ of $\gamma$ we have  $0\leq h_{\gamma \backslash e}<h_{\gamma}$). 
  Let $G$ be a generator of  $\FF_{0,0}$.   Then $(\Delta')^n G\subset \FF_{0,0}^{\otimes n+1}$ and each component  is \motic so has loop number $\geq 1$.
  Since the loop number is a grading, this forces $(\Delta')^n$  to vanish if $n\geq h_G$.
  
  In the general case, let $G$ be of type $(Q,M)$. Then  by $(\ref{Deltaallornothing})$
$$(\Delta' G)^n \subset \FF_{0,0}^{\otimes n} \otimes_{\Z} \FF_{P,Q}$$
and the same  argument shows that 
 $(\Delta' G)^n$  vanishes if $n \geq h_G+1$, the only difference being that the component  in $\FF_{P,Q}$ can satisfy  $h_{\gamma}=0$. 
\end{proof}

\subsection{\Motic descendents of graphs}  \label{sect:Moticdescendants}
  If $G$ is  a \motic graph,  define 
\begin{eqnarray}  \label{twooperators}
d_e (G ) & = &  G/e \qquad \qquad \hbox{ if } e\in E_{G}   \\ 
d_{\gamma} ( G)  & =  &  \gamma \otimes G/\gamma  \qquad \hbox{ if } \gamma \subsetneq E_{G}  \hbox{ is  \motic\!\!} \ .  \nonumber
\end{eqnarray} 
Iterating these operations generates a cascade of tensor products of \motic graphs. 
\begin{defn} \label{defnDescendants} Let $T(\FF) = \bigoplus_{m \geq 0 } \FF^{\otimes m}$ denote the tensor algebra on $\FF$.  The grading by edge numbers induces a  grading 
we shall call the \emph{degree}:
\begin{equation} \label{degreedefn} \deg (\gamma_1 \otimes \ldots \otimes \gamma_n) = \sum_{i=1}^n (N_{\gamma_i} - 1) = N_{\cup \gamma_i} - n\ .
\end{equation}
Given a \motic graph $G$, define  the set of  \emph{\motic descendants} $D(G)$ of  $G$ to be the set of tensor products of graphs (generators in $T(\FF)$) obtained by  repeatedly
applying operators of the form  $\id^{\otimes r} \otimes d_{\bullet} \otimes \id^{\otimes s}$   to $G$, where $d_{\bullet}$ is one of the two operations $(\ref{twooperators})$ above.
Since these strictly decrease the degree, the set $D(G)$ is finite.
\end{defn}

The degree $(\ref{degreedefn})$ is the usual grading in the literature on  the bar construction.  It will correspond to the dimension of facets in the Feynman polytope and 
also to the cohomological degree of  our graph motives (to be defined below).

\begin{rem} The operator $d_e$  is not to be confused with   the contraction of $e$, defined  by 
$c_e G= G\q e.$ Then  the  equation 
\begin{equation} \label{dediffeqn} \Delta c_e = ( c_e\otimes \id + \id \otimes c_e)  \Delta 
\end{equation} 
is not always satisfied, so $\FF$ is not a differential Hopf algebra in general. The reason is the failure of the 
map in theorem  $\ref{thm: moticproperties}$ $(iv)$ to be injective. More precisely, if $G$ is a Feynman graph, and
$\Gamma$ an edge subgraph such that $\Gamma$ and $\Gamma\backslash e$ are both \motic in $G$ with $h_{\Gamma} = h_{\Gamma \backslash e}$,
then $c_e$ is not necessarily zero on $(\Gamma\backslash e)/\Gamma$, and $(\ref{dediffeqn})$ fails. This can only occur if $\Gamma$ is $m.m.$ and $\Gamma \backslash e$ is not (for example, let $\Gamma$ be the subgraph spanned by edges $1,3,4$ in  example \ref{example: Dunce}, and $e=1$). 
 For $G$ of type $(0,0)$
equation $(\ref{dediffeqn})$ is  satisfied and $\FF_{0,0}$ can be made into a differential graded Hopf algebra \cite{BlochKreimer}. 
\end{rem}

\subsection{Uniqueness of graph polynomials} \label{sectUnique}
This section  is not essential and can be skipped. 
It shows that the graph polynomials $\Psi_G$ and $\Xi_G$ are nearly uniquely determined by the factorization and edge-contraction formulae.

\begin{lem} \label{lem: existenceofmotic} Let $\Gamma \subset G$ be an edge-subgraph. Then there exists a 1PI and hence \motic subgraph $\gamma \subseteq \Gamma$ such that
$h_{\gamma} = h_{\Gamma}$. 
\end{lem} 
\begin{proof} Let  $e$  be an edge of  $\Gamma$. If $h_{\Gamma\backslash e} = h_{\Gamma}$ then replace $\Gamma$ with $\Gamma\backslash e$. Repeat
until we obtain a graph $\gamma$ whose loop number drops whenever any edge is cut.
\end{proof}

\begin{prop} For every labelled  \motic Feynman graph $G$, let 
$$P_G, C_G \in \Q[\alpha_e, e \in E_G]$$
be homogeneous polynomials of degrees $h_G, h_G+1$ respectively, which respect the equivalence 
relation of \S\ref{sect: FeynmanGraphs} and take the same values on $G$ and $G \cup \{v\}$ where $v$ is an isolated vertex, and 
 satisfy the following properties:
\begin{enumerate}
\item Partial factorisations:
\begin{eqnarray}
P_{G} & \equiv_{{\gamma}} &  P_{\gamma} P_{G/\gamma}  \nonumber \\
C_{G} & \equiv_{{\gamma}} &  P_{\gamma} C_{G/\gamma}  \quad  \hbox{ if } \gamma \hbox{ not } m.m.  \nonumber \\
C_{G} & \equiv_{{\gamma}} &  C_{\gamma} P_{G/\gamma}  \quad \hbox{ if } \gamma \hbox{ is } m.m.  \nonumber 
\end{eqnarray} 
where  $A_G \equiv_{\gamma} B_{\gamma} C_{G/\gamma}$ for homogeneous polynomials $A,B,C$  means that  $A_G - B_{\gamma} C_{G/\gamma}$ is of degree $> \deg B_{\gamma}$
in the variables $\alpha_e$, for $e\in E_{\gamma}$. 
\item Edge contraction:
\begin{eqnarray}
P_{G}\big|_{\alpha_e=0}  & =  &   P_{G\q e}  \nonumber \\
C_{G}\big|_{\alpha_e=0}  &= &  C_{G\q e}    \nonumber 
\end{eqnarray}

\item  Initial conditions: if $G$ has a single edge then $P_G=\Psi_G$ and $C_G = \Xi_G(q,m)$. 
If $G$ is a  banana graph (a connected graph with 2 vertices) with all edges  massive, then the coefficient of $\prod_{e\in E_G} \alpha_e$ in $C_G$ is 
$$ q^2 + \sum_{e\in E_G} m_e^2\ .$$

\end{enumerate}
With these assumptions,  $P_G = \Psi_G$ and $C_G = \Xi_G(q,m)$. 
\end{prop}

\begin{proof} 
The proof proceeds by induction on the number of edges. Assume for now that the theorem is true for all \motic graphs $G$ such that $h_G\leq 1$. For the induction step,  suppose that $G$ is \motic and satisfies $h_G \geq 2$. 
 For any edge $e_0\in E_G$, there exists a   non-trivial \motic subgraph $\gamma \subset G$ such that $e_0\notin E_{\gamma}$ and $h_{\gamma} = h_G-1$, by lemma \ref{lem: existenceofmotic}.  It may or may not be $m.m.$. We obtain by induction hypothesis
 $$P_{G} = \Psi_{\gamma} \Psi_{G/\gamma} +R_{\gamma, G/\gamma}^P $$
on application of $(1)$, 
where the degree of $R_{\gamma, G/\gamma}^P$ in the variables $\alpha_e$, for $e\in \gamma$ is $\geq h_{\gamma}+1$ and hence equal to $h_G$. 
Thus $R_{\gamma, G/\gamma}^P$ does not depend on the variables $\alpha_e$, for $e\in G/\gamma$. We deduce that  $P_G$ is of degree $\leq 1$ in $\alpha_{e_0}$
and the coefficient of $\alpha_{e_0}$ is $\Psi_{\gamma} \Psi^{e_0}_{G/\gamma}$. The constant term in $\alpha_{e_0}$ is uniquely determined from $(2)$ and induction hypothesis
since $G\q e_{0}$ has fewer edges. This proves that $P_G=\Psi_G$ for all $G$.  
For the polynomial $C_G$, an application of $(1)$ and the induction hypothesis gives either
\begin{eqnarray} C_{G}& =&  \Xi_{\gamma}(q,m) \, \Psi_{G/\gamma} +R^{C,\IR}_{\gamma, G/\gamma} 
 \nonumber \\ 
\hbox{or }\quad C_{G}   &= & \Psi_{\gamma}\,  \Xi_{G/\gamma}(q,m) +R_{\gamma, G/\gamma}^{C,\UV}  \quad   \nonumber
\end{eqnarray} 
depending on whether $\gamma$ is $m.m.$ or not.  The former case proceeds as for $P_G$. In the latter case, the term 
$R_{\gamma, G/\gamma}^{C,\UV}$ is of degree at most one in $\alpha_{e_0}$, and hence the coefficients of $(\alpha_{e_0})^k$ for $k\geq 2$
are uniquely determined by induction. The coefficients of $(\alpha_{e_0})^0$ are determined by contracting the edge $e_0$ via $(2)$. Thus the only
undetermined term in $C_G$ is the unique   monomial  $\prod_e \alpha_e$ of degree exactly one in every $\alpha_e$ for $e\in E_G$. It can only occur in $C_G$  if $h_G+1=\deg C_G = N_G$. If $G$ has more than one component, it   necessarily contains a self-edge, say $e'$, which is a \motic subgraph and not $m.m.$. Applying $(1)$ to this subgraph gives
$$C_G = \alpha_{e'} C_{G/e'} (q, m) + O(\alpha_{e'}^2)$$
and so the term linear in $\alpha_{e'}$ is again determined by induction.  The only remaining case is when $G$ is connected. Then $N_G = h_G+1$ implies that it has two vertices by Euler's formula $(\ref{eqn: Eulerformula})$, so it is a banana graph. Furthermore, every edge is massive otherwise we could construct a non-trivial \motic $m.m.$ subgraph $\gamma$, and determine $C_G$ using the third formula of $(1)$. Thus we are reduced to the case of $(3)$.

It remains to check the cases when $G$ does not have a non-trivial \motic subgraph, i.e., $h_G\leq 1$.
First suppose that $G$ is a forest $(h_G=0)$.  Then $P_G$ is  constant (of degree $0$), and by  contracting edges we deduce that $P_G= P_{\{v\}}=1$, where $v$ is an isolated vertex. The polynomial $C_G$ is homogeneous of degree one, i.e.,  $C_G = \sum_e \lambda_e \alpha_e$. The coefficient  $\lambda_e$ is uniquely
determined by contracting all edges except $e$, and we are reduced to the case of a graph with a single edge $(3)$.
Now suppose that $h_G=1$. The same argument proves
the statement for $P_G$. The polynomial $C_G$ is homogeneous of degree two and hence of the form 
$C_G = \sum_{e,f} \lambda_{e,f} \alpha_e\alpha_f$.  By contracting all edges except $e,f$ we reduce to a two-edge graph.  By edge contraction, $\lambda_{e,e}$
and $\lambda_{f,f}$ are determined via $(3)$. If $e, f$ form a massive $2$-edge banana, the coefficient $\lambda_{e,f}$ is determined by the second part of $(3)$. In all other cases, $G$ has  a non-trivial \motic $m.m.$  subgraph and we can reduce to a graph with fewer edges by the partial factorization formulae as above.  \end{proof}

\begin{rem}  
Properties $(1)$ and $(2)$ are essential requirements for the product-structure on graph hypersurfaces and hence for our results on the 
action of the cosmic Galois group. Thinking  of the data of $P$ and $C$ as Feynman rules,  the proposition tells us  how restrictive these requirements are.
The polynomial $P_G$ is essentially uniquely determined,  but there is nonetheless a small amount of  freedom to modify the polynomial $C_G$
by adding a term
 $$ s_n \prod_{e\in E_G} \alpha_e$$
 to every massive banana graph $G$ with $n$ loops, where $s_n$ is a new parameter. By properties $(1)$ and $(2)$, the coefficients $s_n$  will infiltrate the $C_G$ for all other Feynman graphs. In this way one could, suprisingly,  modify the Feynman integrands by essentially a single quantity  $s_n$  (which  could depend on the labelling of the edges $E_G$)
at every loop order without affecting the mathematical structures studied in this paper.\footnote{This provides Feynman integrals with a new and natural parametrisation, distinct from the kinematic parameters,  which may be useful for setting up differential equations.}
 Observe also that the concept of mass and momenta only enter via the initial conditions $(3)$. One could  allow for more diverse families of polynomials $P_G,C_G$ either  by restricting the set of graphs under consideration, or by   allowing the polynomials $P_G$, $C_G$ to have higher degrees. \end{rem} 

  An interesting question would be  to study  similar partial  factorisation properties of Feynman  integrands for gauge theories
 and see what restrictions this imposes on the set of possible Feynman rules.

\section{Linear blow-ups in projective space} \label{sect: LinearBlowups}
We study  blow-ups of coordinate linear subspaces of projective space.
The role of local coordinates is emphasised owing to their close  relation to sector decompositions in the physics literature.

\subsection{Iterated blow-ups}
Let $S$ be a finite set, and let $\Pro^S=\Pro(\Q^S)$ denote projective space over $\Q$ of dimension $|S|-1$ with projective coordinates $\alpha_s$, $s\in S$.
Every subset $I\subset S$ defines a linear subspace  
$$L_I \cong  \Pro^{I^c}  \subset \Pro^S $$
defined by the vanishing of coordinates $\alpha_i$, $i\in I$. The notation $I^c$ denotes the complement  $S\backslash I$ when the set $S$ is unambiguous. We have
$$L_{I_1 } \cap L_{I_2} = L_{I_1 \cup I_2}\ .$$
Now let $B\subset 2^S$   be a set of subsets of $S$ with the property that 
\begin{equation} \label{eqn:  Bclosedunderunions}
I_1,  I_2  \in B  \quad \Longrightarrow \quad  I_1 \cup I_2  \in B \ 
\end{equation}
and satisfying $S\in B$. Define the iterated blow-up $P^B$  of $\Pro^S$ along $B$
$$\pi_B: P^B \To \Pro^S$$
by the following standard procedure:

\vspace{0.05in}
$(0)$  First blow up all subspaces $L_I$, for $I\in B$, such that  $\dim(L_I)=0$, in any order,  to obtain a space $P_0\rightarrow \Pro^S$.

\vspace{0.02in}
$(1)$  Blow up all strict transforms of $L_I$ in $P_0$, for $I\in B$, such that  $\dim(L_I)=1$, in any order, to obtain a space $P_1 \rightarrow P_0$.

\vspace{0.02in}
$(k)$ At the $k^\mathrm{th}$ stage, blow up the strict transforms of $L_I$ in $P_{k-1}$, for $I\in B$  such that $\dim(L_I)=k$, to obtain $P_k \rightarrow P_{k-1}$.
\vspace{0.05in}

\noindent 
 Finally, define $P^B= P_{|S|-1}$.
The key point is that since $B$ is closed under unions  $(\ref{eqn:  Bclosedunderunions})$, the strict transforms of $L_I$ in $P_{k-1}$ for $I\in B$ are disjoint and can be blown up in any order, and thus $P^B$ is well-defined. 
 The scheme $P^B$ has a  divisor
\begin{equation}\label{eqn: Ddivisordefn}
 D = \pi^{-1} (\bigcup_{i\in S} L_i)
 \end{equation}
given by the total transform of the coordinate hyperplanes $L_i$.  If $U_I\subset L_I$ denotes the open  where $\alpha_j\neq 0$ for all $j\notin I$, the irreducible 
components of $D$ are 
\begin{eqnarray} \label{DiDI}
D_i& =& \overline{\pi^{-1}(U_i)} \quad , \hbox{ for all } i \in S\ .   \\ 
D_I&=&  \overline{\pi^{-1}(U_I)}  \quad , \hbox{ for all } I \in B, \hbox{ where }2\leq  |I| \leq |S|-1\ , \nonumber
\end{eqnarray} 
where the closure is with respect to the Zariski topology.  By taking repeated intersections, $D$ defines a stratification on $P^B$.   For every $I\in B$ define
\begin{eqnarray}
B^I & = &  \{ J \in B \hbox{ such that }  J\subseteq I\} \nonumber \\
B_I & =  & \{ J \backslash I  \hbox{ where }  J \in B \hbox{ and }  J\supseteq I\} \ , \nonumber
\end{eqnarray} 
and for every $i\in S$, set 
\begin{eqnarray}
B_i & = & \{ J\backslash (J\cap \{i\}) \hbox{ for } J \in B\}\ .
\end{eqnarray}

\begin{thm} \label{prop: PBstructure} The space $P^B$ is a smooth scheme over $\Z$ and is well-defined (it does not depend on the order of blow-ups at  each stage of the above procedure). 

The divisor $D$ is strict normal crossing, and there are canonical isomorphisms
\begin{eqnarray}\label{eqn: DIprodstructures}
D_I  &= & P^{B^I} \times P^{B_I}   \qquad \hbox{ where } I \in B\\
D_i & = & \Spec\, \Z \times  P^{B_{i}} \ .\nonumber 
\end{eqnarray}
Consider any two components $D_I,D_J$ (of either type $(\ref{DiDI})$), where    $I,J \in B$ or  a singleton in $S$.
Then $D_I \cap D_J$ is non-empty if and only if either
$$I\subset J \quad \hbox{or} \quad J\subset I $$
or 
$$I \cap J = \emptyset \quad\hbox{and}\quad  I\cup J \notin B\ .$$ 
The latter case only arises when at least one of $I,J$ is a singleton in $S$, by $(\ref{eqn:  Bclosedunderunions})$.\end{thm}

A proof  is outlined  below. 

\subsection{$B$-polytope} \label{sect: Bpolytope} Define  a compact real manifold with corners
$$\widetilde{\sigma}_B \subset P^B(\R)$$
to be the closure, in the analytic topology, of $\pi_B^{-1}(\overset{\circ}{\sigma})$, where $\overset{\circ}{\sigma} \subset \Pro^S(\R)$ is the open coordinate simplex defined by 
$\alpha_i >0$, for all $i\in S$.  Its facets inherit the following product structure from the isomorphisms  $(\ref{eqn: DIprodstructures})$ of theorem \ref{prop: PBstructure}. 
\begin{cor} \label{cor: facetstructure} The facets of $\widetilde{\sigma}_B$ satisfy
\begin{eqnarray}
\widetilde{\sigma}_B \cap D_i(\C)  & =& \{pt\} \times \widetilde{\sigma}_{B_{i}}  \quad\qquad \hbox{ for } i \in S\nonumber \\
\widetilde{\sigma}_B \cap D_I(\C)   & =  &  \widetilde{\sigma}_{B^{I}}  \times \widetilde{\sigma}_{B_{I}}  \quad \qquad \hbox{ for } I \in B, 2 \leq |I| \leq |S|-1\ . \nonumber 
\end{eqnarray}
\end{cor} 
\noindent The poset structure on the faces of $\widetilde{\sigma}_B$, with respect to inclusion,  is identical to the poset structure on the stratification of $P^B$ generated by the divisor $D$.

\subsection{Local coordinates and  theorem $\ref{prop: PBstructure}$} \label{sect:  localcoords}
The space $P^B$ will be covered by explicit coordinate charts of the form
$$\A^n = \Spec \Z[\beta_1,\ldots, \beta_n]\ .$$
These charts are obtained by iterating the following basic example.

\begin{example} \label{example: oneblowup} Consider a single blow-up in affine space. If $J\subset \{1,\ldots, n\}$ let $L_J $ denote the zero locus of $\alpha_j$, $j\in J$. 
 Let $ I=\{1,\ldots, m\}  \subset \{1,\ldots, n\}$.  The blow-up $A \rightarrow \A^n$ of $\A^n$ along $L_I$ has local coordinates
$$\beta_1= {\alpha_1 \over \alpha_m},\  \beta_2 = {\alpha_2 \over \alpha_m}, \ldots,\  \beta_{m-1} = {\alpha_{m-1} \over \alpha_m}, \  \beta_m = \alpha_m, \ \beta_{m+1} = \alpha_{m+1}, \ 
\ldots , \ \beta_n = \alpha_n \  . $$
This means that there is an affine chart $\A^n \subset A$ with a morphism
$$\pi : \Spec \Z[\beta_1,\ldots, \beta_n] \To \Spec \Z[\alpha_1,\ldots, \alpha_n]$$
defined by $\pi^*(\alpha_i ) = \beta_i \beta_m$ for $1\leq i< m$ and $\pi^*(\alpha_i) = \beta_i$ for $i\geq m$.  It is  an isomorphism on the opens defined by $\beta_m\neq 0$ and $\alpha_m\neq 0$ respectively. 
  The exceptional divisor in these coordinates is given by the equation $\beta_m=0$, and the strict transform of $L_i$ is given by $\beta_i=0$ for all $i\neq m$. 
The strict transform of $L_m$ in $A$ does not meet this coordinate chart,  and hence neither does the strict transform of any $L_J$,  for $m\in J \subsetneq I$.

More generally,  for any choice of element $j\in I$, we have local coordinates
$$\beta_i = {\alpha_i \over \alpha_j}  \hbox{ for } i \in I\backslash \{j\}  , \quad   \beta_j=\alpha_j  ,\quad \beta_i= \alpha_i 
\hbox{ for } i \notin I$$
and hence a local chart $\Spec \Z[\beta_1,\ldots, \beta_n]$ on $A$ as above. These form  an affine covering of $A$. Note that if $J_1, J_2 \subset I$ and $J_1 \cup J_2= I$, then the strict transforms
of $L_{J_1} $ and $L_{J_2}$ do not intersect in $A$, since this is true in every coordinate chart.
\end{example}

We now define a scheme, denoted $P^B$, explicitly using such  affine charts. It will turn out to be isomorphic to the space defined in the previous section. 

For every nested sequence (or flag)
\begin{equation} \label{eqn: Flagdefn} 
\FF : \qquad \quad \emptyset =I_0 \subsetneq  I_1 \subsetneq I_2 \subsetneq \ldots \subsetneq I_{k} \subsetneq  I_{k+1}= S
\end{equation}
where each $I_r \in B$, and every choice of elements
\begin{equation} \label{eqn:  jnchoices} 
c: \qquad \qquad j_n \in I_{n} \backslash I_{n-1} \quad \hbox{ for } \quad 1\leq n \leq k+1\ ,
\end{equation} 
which are  maximal in the sense that  $\FF, c$ cannot be made larger  (i.e., there exists no $I\in B$ and $0\leq i \leq k$ such that $I_i \subsetneq I \subsetneq I_{i+1}$
and $j_{i+1} \notin I$),
define an affine
$$ \A^{\FF,c} = \Spec \Z[\beta^{\FF,c}_i, i\neq j_{k+1}]$$
and   a morphism 
$\pi: \A^{\FF,c} \rightarrow \Spec \Z[\alpha_i, i\neq j_{k+1}]$
where the right-hand side is the open subset  $\alpha_{j_{k+1}}=1$ in $\Pro^S$. The morphism $\pi$ is  defined as follows.  Its inverse $(\pi^{-1})^*$ on the open 
 $\alpha_{j_n} \neq 0$ for  all $1\leq n \leq k$ is given by 
\begin{eqnarray} \label{eqn: localcoords}
\beta^{\FF,c}_i &= & {\alpha_{i} \over \alpha_{j_n} } \quad \hbox{ for }\quad   i \in I_n \backslash  ( \{j_{n}\} \cup I_{n-1} ) \ , \  \quad 1\leq n \leq k+1 \ ,  \\
\hbox{ and } \qquad \beta^{\FF,c}_{j_n} &= & {\alpha_{j_n} \over \alpha_{j_{n+1}} } \quad \hbox{ for }   \quad 1\leq n \leq k+1 \ , \nonumber 
\end{eqnarray}
where we set $\alpha_{j_{k+2}} = 1$;  the map $\pi^*$  is obtained by writing the  $\alpha$'s in terms of the $\beta^{\FF,c}$'s in  the previous equations, i.e.,
\begin{eqnarray}  \label{newpistarequations}
\pi^*(\alpha_{j_n}) & = & \beta^{\FF,c}_{j_n} \beta^{\FF,c}_{j_{n+1}} \ldots  \beta^{\FF,c}_{j_{k+1}}   \\ 
\pi^*(\alpha_i) & =  &  \beta^{\FF,c}_i\beta^{\FF,c}_{j_n} \beta^{\FF,c}_{j_{n+1}} \ldots  \beta^{\FF,c}_{j_{k+1}} \qquad \hbox{ for }   i \in I_n \backslash  ( \{j_{n}\} \cup I_{n-1} ) \nonumber \ . 
\end{eqnarray} 
  The morphism $\pi$ restricts to  an isomorphism between the  open subsets defined by  $\beta^{\FF,c}_{j_n} \neq 0$ and  $\alpha_{j_n} \neq 0$ for  all $1\leq n \leq k$.
  A simple example is given in \S\ref{sectAffinecovering}.

  By $( \ref{eqn: localcoords})$, the coordinate rings of the $\A^{\FF,c}$ are contained in the fraction field of $\Pro^S$, and glue together to form a
  scheme $P^B$ with a morphism $\pi : P^B \rightarrow \Pro^S$  over $\Spec \Z$. We claim that $P^B$ is indeed the space defined by blow-ups in the first paragraph, and that the 
  $\beta^{\FF,c}_j$ are local coordinates in the neighbourhood of a `corner' 
  $$ (\FF,c) : \qquad   \bigcap_{i \in S\backslash \{j_1,\ldots, j_k\}}  D_i \cap \bigcap_{1\leq n\leq k } D_{I_n}$$
   where  the trace of $D_{I_n}$ on the chart $\A^{\FF,c}$ is given  by $\beta^{\FF,c}_{j_n}=0$ and $D_i$ by $\beta^{\FF,c}_i=0$.  These divisors are clearly  normal crossing,
   and meet according to the rules described in the second half of theorem $\ref{prop: PBstructure}$. For example, 
  divisors of the form $D_{I}$ and $D_J$ with $I,J\in B$  only meet on such a chart  if the sets $I,J$ fit into a flag $\FF$,  in which case $I\subset J$ or $J\subset I$; the 
  remaining cases are left to the reader. 
   
    Next observe that the product structure of these divisors is clear from the structure of  the  coordinates $(\ref{eqn: localcoords})$ which give a canonical isomorphism:
  \begin{equation} \label{eqn: isomDIr}
  D_{I_r}  \cap \A^{\FF,c} = V(\beta^{\FF,c}_{j_r}) \overset{\sim}{\To} \A^{\FF^{I_{r}},c^{I_{r}}} \times \A^{\FF_{I_r},c_{I_r}}
  \end{equation}
 where $\FF^{I_r},  \FF_{I_r}$ are the flags in $I_r$ and $S\backslash I_r$ defined by 
 \begin{eqnarray}
 \FF^{I_r} & :& \qquad   \emptyset =I_0 \subsetneq  I_1 \subsetneq I_2  \ldots \subsetneq I_r \nonumber \\
 \FF_{I_r} & : &\qquad      \emptyset  \subsetneq  I_{r+1}\backslash I_r   \subsetneq  \ldots \subsetneq I_{k} \backslash  I_r \subsetneq S \backslash  I_r    \nonumber 
    \end{eqnarray}
  and $c^{I_{r}}, c_{I_{r}}$ are the obvious restrictions of $c$ to $\FF^{I_k}, \FF_{I_k}$ respectively.
  Gluing the  morphisms $(\ref{eqn: isomDIr})$  together for all $\FF, c$  proves the first  equation of $(\ref{eqn: DIprodstructures})$ since every  pair of flags in $B^{I_r}$, $B_{I_r}$
 is obtained from a flag in $B$ in this way. Similarly,
  $$\A^{\FF,c} \supset V(\beta^{\FF,c}_{i} ) \overset{\sim}{\To}  \A^{\FF_{i},c_{i}}$$
 where $i \in I_r \backslash (I_{r-1}\cup \{j_r\}) $ and $\FF_i$ is the flag in $S\backslash \{i\}$ defined by 
 $$ \FF_i  : \qquad   \quad \emptyset =I_0 \subsetneq  I_1 \subsetneq \ldots \subsetneq  I_{r-1} \subsetneq I_r  \backslash \{i\} \subsetneq \ldots \subsetneq I_{k}\backslash \{i\} \subsetneq  I_{k+1}= S\backslash \{i\}\ ,$$
 which proves the second  equation of $(\ref{eqn: DIprodstructures})$. 
    It remains to show that the spaces $P^B$ as defined above are indeed given by the blow-up procedure. This can be proved by induction.    First of all, if $B$ is empty, then $P^B$ is simply $\Pro^S$ and the coordinate charts $\A^{\FF,c}$ are the usual affine covering of $\Pro^S$. 
    Let   $P_{-1} = \Pro^S$, and   suppose by induction that $P_{n-1}$ is isomorphic to $P^{B(n-1)}$ where $B(n-1)$ is the subset of $B$  consisting  of all $I \in B$ such that $|I| \geq |S|-n$, where $n\geq 0$.  It is stable under unions. Let $I \in B$ of cardinality $|I| = |S|-n-1$. 
 
 \begin{lem} \label{lem: Flagavoidance} Let $\FF, c$ be  a maximal  flag in $B(n-1)$ given by $(\ref{eqn: Flagdefn})$ and $(\ref{eqn:  jnchoices})$ where $I_1,\ldots, I_k \in B(n-1)$. Let $I\in B(n) $ as above.   The strict transform of $L_I$  in    $P^{B(n-1)}$    meets the chart $\A^{\FF,c}$ if and only if 
 $I \subset I_1 \backslash \{j_1\}$.  \end{lem} 
 \begin{proof} If $j_r \in I$ for any $1\leq r\leq k+1$ then the strict transform of $L_I$ is contained in the strict transform of $L_{j_r}$, which does not meet $\A^{\FF,c}$ by inspection of the first equation of $(\ref{newpistarequations})$ (its total transform is contained in the union of the vanishing locus of the $\beta^{\FF,c}_{j_{k}}$, which are all exceptional divisors). 
 On the other hand, suppose that $I$ does not contain any $j_r$, and  is not contained in $I_1$. Let $m$ be the smallest integer such that 
 $I \subset I_{m+1}$. Then $I \subsetneq I_m$ so 
 $$I_{m} \subsetneq I \cup I_{m} \subsetneq I_{m+1}\ .$$
Since $B$ is closed under unions,  $I \cup I_{m}$ is an element of $B(n-1)$ which contradicts the maximality of $\FF, c$ (as $j_{m+1} \notin I$).   Thus for the strict transform of $L_I$ to meet $\A^{\FF,c}$,
 we must have $I \subset I_1$ and $j_1 \notin I$.  Conversely, when  this holds, the  intersection of the strict transform of $L_I$ with $\A^{\FF,c}$ is given by the equations $\beta^{\FF,c}_i=0$ for $i \in I$. 
  \end{proof} 
 
 Now, in the situation of the previous lemma, we can blow up the strict transform $L_I$ explicitly in the affine chart $\A^{\FF,c}$ 
 using example $\ref{example: oneblowup}$. It  is  
 covered by the affine charts $\A^{\FF', c'}$ where $\FF',c'$  extends $\FF, c$ to the left: 
 $$\FF': \qquad \qquad     \emptyset \subsetneq I \subsetneq I_1 \subseteq \ldots \subsetneq I_k \subsetneq S$$
and the  restriction of $c'$ to $\FF$ is  $c$. Proceeding in this way, we can blow up the strict transforms of $L_I$, for $I\in B$ of cardinality $|I| = |S|-n-1$ in any order, to give exactly the space $P^{B(n)}$ together with its affine covering $\A^{\FF',c}$. 
This proves by induction that the two descriptions of $P^B$, by explicit coordinates and also by iterated blow-ups, are one and the same.

\subsection{An affine model} \label{sectAffineModel} Given $B$ satisfying $(\ref{eqn:  Bclosedunderunions})$ with $S\in B$, we shall construct an affine subspace $A^B \subset P^B$ by removing strict transforms of certain linear hyperplanes, in such a way that the product-structure of $P^B$ is also satisfied by  $A^B$.

We can assume that for every $I \in B$, $|I| \geq 2$.  For any subset $J\subset S$, denote by 
$$\alpha_J = \sum_{j \in J} \alpha_j\ .$$
Let $H_J = V(\alpha_J)\subset \Pro^S$ denote the corresponding hyperplane.
Write   $\A^S = \Pro^S \backslash H_S$, and consider the hyperplane complement   $\A^S \backslash \cup_{I \in B} H_I$.  We shall directly 
define a partial compactification of this space as follows. Write $B^+ = B \cup \{ \{i\}, i\in S\}$. 

\begin{defn} Let $R_B = \Z[ b_{I/J} : I \subset S ,  J \in B  \hbox{ such that } \emptyset \neq I \subseteq J ] /\mathcal{I}$, 
where $\mathcal{I}$ is the ideal generated by the following relations:
\begin{eqnarray} \label{bIJrels}
b_{I/J}  &=&  \sum_{i \in I}  b_{\{i\}/J} \\
b_{J/J}  &= &  1  \nonumber \\
b_{I/J} b_{J/K} & = & b_{I/K}  \qquad \hbox{ if  } J, K \in B\nonumber 
\end{eqnarray}
For every $I \in B^+\backslash \{S\}$, let $D_I = V( b_{I/J} \hbox{ for all } I \subsetneq J \in B ) \subset \Spec(R_B)$. 
\end{defn}

There is a homomorphism
\begin{eqnarray}
R_B  & \To&  \Z [\alpha_i  \hbox{ for } i \in S,  \alpha_J^{-1} \hbox{ for } J \in B, |J|\geq 2] \\
b_{I/J} & \mapsto &  { \alpha_I  \over \alpha_J} \nonumber 
\end{eqnarray} 
 which respects the relations $(\ref{bIJrels})$, and hence a morphism 
 \begin{equation}  \label{AtoSpecRmorphism} \A^S\backslash \cup_{I \in B} H_I \To \Spec (R_B)\ .
\end{equation} 
 Let $ \pi : P^B\rightarrow \Pro^S$ denote the blow-up defined in the previous paragraphs. Let $\widetilde{H}_I$ denote the strict transform of $H_I$, for $I \in B^+$. 
 
 \begin{defn} Let $A^B = P^B \backslash  \cup_{I \in B} \widetilde{H}_I$. For every $I \in B^+\backslash \{S\}$, let  $D_I\subset P^B$ denote the divisor defined in \S \ref{sect:  localcoords}, and, by abuse of notation, let us also denote by  $D_I$  its intersection with  $A^B$. Let $D$ denote their union. 
 \end{defn} 
 
 Observe that  $\A^S \backslash  \cup_{I \in B} H_I$ is an open subspace of $ A^B$, since the image in $\Pro^S$ of every exceptional divisor in $P^B$
 is contained in some such hyperplane $H_I$.  More precisely, $\A^S \backslash  \cup_{I \in B} H_I\cong P^B \backslash 
( D \cup   \cup_{I \in B} \widetilde{H}_I )$. 
 
 \begin{thm} \label{thmaffinespace}  The morphism $(\ref{AtoSpecRmorphism})$   extends to a canonical  isomorphism
 $$A^B \overset{\sim}{\To} \Spec R_B \ .$$
 In particular,  $(\ref{AtoSpecRmorphism})$ has Zariski-dense image, and furthermore, $A^B$ is affine and $\Spec R_B$ is smooth over $\Z$.  The divisors  denoted $D_I$ in $A^B$ and $\Spec R_B$ are mapped isomorphically to each other, and there are canonical isomorphisms
 $$ D_i  = \mathrm{Spec}\, \Z \times A^{B_i}  \qquad \hbox{ and } \qquad D_I = A^{B^I} \times A^{B_I}\ .$$
 \end{thm}

If the set $S=\{1,\ldots, n\}$ is ordered, and $B$ consists of all consecutive sets $\{k,k+1,\ldots, k+\ell\}$, for $1 \leq k \leq n-\ell$ and $\ell \geq 2$, then $A^B$ is isomorphic to the affine  space $\mathfrak{M}_{0,n+2}^{\delta}$
defined in the author's thesis\footnote{To write down the isomorphism,  it suffices to observe  that the simplicial coordinates $t_i$ for $1\leq i \leq n-3$ correspond to 
 $b_{I/S}$ where $I= \{1,2, \ldots, i\}$ and $S= \{1,2,\ldots, n-2\}$. }. The strategy of \S\ref{sectOutlineaffine}  gives a new proof of its main properties. 
 On the other hand, if we take $B$ to be all subsets of $S$ of cardinality $\geq 2$, we obtain  an affine algebraic model of a permutohedron. 

\begin{example} Let $S = \{1,2\}$, and let $B=S$. Then there is nothing to  blow up  and $P^B= \Pro^1$.  Its affine version is $A^B = \Pro^1 \backslash \{\alpha_1+ \alpha_2= 0 \} \cong \A^1$.  In the explicit coordinates above,  $R_B= \Z[b_{1/12}, b_{2/12}]/\mathcal{I}$ where $\mathcal{I}$ is the ideal
generated by the relation $b_{1/12}+b_{2/12}=1$, and the composition with inclusion into $\A^2$
$$A^B  \overset{\sim}{\To} \Spec R_B \subset \A^2$$
is given in projective coordinates by  $(\alpha_1 : \alpha_2) \mapsto (\alpha_1/(\alpha_1+ \alpha_2) , \alpha_2 /(\alpha_1 + \alpha_2))$. 
\end{example}

\subsection{Outline of proof of theorem  $\ref{thmaffinespace}$} \label{sectOutlineaffine} Let $(\FF, c)$ denote a maximal pair $(\ref{eqn: Flagdefn})$ and $\A^{\FF,c}$ the corresponding chart of $P^B$. 
Let us denote by $U^{\FF,c}$ the open $\A^{\FF,c} \cap A^B$.  We must first show that the morphism $(\ref{AtoSpecRmorphism})$ canonically extends to a morphism
$$ U^{\FF, c} \To \Spec R_B\ .$$
To see this, compute the strict transform of $H_J$, for $J\in B$ in the coordinates $\beta^{\FF,c}_i$. It follows from the definition of the coordinates $\beta^{\FF,c}$ in $(\ref{eqn: localcoords})$  that 
$$\pi^* \alpha_J = P_J \times \prod_{r \geq \ell}^k \beta^{\FF,c}_{j_r}$$
where $P_J$ is an irreducible polynomial in the $\beta^{\FF,c}_i$, and $\ell$ is the smallest integer such that $J \subseteq I_{\ell}$.  
We have $P_{\{j_n\}} = 1$. 
 The strict transform $\widetilde{H}_J$  of $H_J$ is given locally on $\A^{\FF,c}$ by the zero locus of $P_J$.  It follows from this calculation that 
$$\Or( U^{\FF,c}) = \Spec \big(\Z [ \beta^{\FF,c}] [ P_J^{-1}, \hbox{ for } J \in B, |J|\geq 2]\big)  $$
and that we have explicitly 
\begin{eqnarray} \label{RBtoOU}
R_B & \To &  \Or( U^{\FF,c})\\
b_{I/J} &\mapsto&  {P_I \over P_J} \times \prod_{r=\ell_1}^{\ell_2-1} \beta^{\FF,c}_{j_r} \nonumber
\end{eqnarray}
where  $\ell_1 \leq \ell_2$ are minimal such that $I\subset I_{\ell_1}$, $J \subset I_{\ell_2}$.   The point is that there are no terms $\beta_{j_r}^{\FF,c}$ in the denominator. One can check that $(\ref{RBtoOU})$ is well-defined (respects the relations in the defining ideal of $R_B$), since it is compatible with $\pi^*$ and $(\ref{AtoSpecRmorphism})$. Gluing the resulting maps $U^{\FF,c} \rightarrow \Spec R_B$ together, we deduce  that 
$(\ref{AtoSpecRmorphism})$ extends to a morphism
$$A^B  = \bigcup_{\FF,c}  U^{\FF,c} \To \Spec R_B\ .$$
For the next step, observe that if $j_r \in J \subset I_r$ then by the maximality of $(\FF,c)$ there is no smaller $\ell <r$ such that $J \subset I_{\ell}$ and therefore under $(\ref{RBtoOU})$ 
$$b_{J/ I_r }  \mapsto  {P_J \over P_{I_r}}\ .$$
Setting  $J= \{j_r\}$, we have $b_{j_r/I_r} \mapsto P_{I_r}^{-1}$, which is invertible in $\Or(U^{\FF,c})$.   It follows that  $(\ref{RBtoOU})$ factorizes through the ring
\begin{equation} \label{RBtoOrUFF}  R_B [ b_{j_1/I_1}^{-1}, \ldots, b_{j_{k+1} / I_{k+1}}^{-1} ] \To \Or(U^{\FF,c})\ ,
\end{equation} 
where $I_{k+1}=S$.
Note that the defining relation $b_{j_r/I_r} = b_{j_r/J} b_{J/I_r}$ implies that inverting $b_{j_r/I_r}$  also inverts  $b_{J/I_r}$ and $b_{j_r/J}$ for $J\in B$. 
We claim that $(\ref{RBtoOrUFF})$ is an isomorphism. To see this, we can write down  a  map  in the opposite direction
\begin{eqnarray} \beta^{\FF,c}_{i} & \mapsto&    b_{i/I_n}  b^{-1} _{j_n/I_n} \quad  \hbox{ if  }   i \in I_n \backslash  ( \{j_{n}\} \cup I_{n-1} ) \ , \  \quad 1\leq n \leq k+1 \ ,   \nonumber \\
\hbox{ and } \qquad \beta^{\FF,c}_{j_n} &\mapsto &  b_{j_n/I_{n+1}}  b^{-1}_{j_{n+1}/I_{n+1}}  \quad \hbox{ for }   \quad 1\leq n \leq k+1 \ , \nonumber 
\end{eqnarray}
consistent with $(\ref{eqn: localcoords})$. It clearly lands in $R_B [ b_{j_1/I_1}^{-1}, \ldots, b_{j_{k+1} / I_{k+1}}^{-1} ] $, and one checks that 
it is  indeed the inverse of $(\ref{RBtoOU})$. We omit the details.

For the final step, define   open affine subspaces $ V^{\FF,c} \subset \Spec R_B$ by
$$V^{\FF,c} = \Spec R_B [ b_{j_1/I_1}^{-1}, \ldots, b_{j_{k+1} / I_{k+1}}^{-1} ]\ .$$
We have shown that $U^{\FF, c} \cong V^{\FF,c}$ are canonically isomorphic. It suffices to check that the opens $V^{\FF,c}$ form a covering of $\Spec R_B$. 
To see this, observe that 
$$\bigcup_{i \in I_1}  \Spec R_B [ b_{i/I_1}^{-1} ,  b^{-1}_{j_2/I_2}, \ldots, b^{-1}_{j_{k+1}/ I_{k+1}}] \subseteq \Spec R_B[b^{-1}_{j_2/I_2}, \ldots,b^{-1}_{j_{k+1}/ I_{k+1}}] $$ is an equality, 
which follows from the partition of unity relation  
$\sum_{i\in I_1} b_{i/I_1}=1\ $
in $R_B$. By varying $c$ in $(\FF,c)$, and eliminating the $b_{j_r/I_r}^{-1}$ in turn,  we deduce that 
$$\bigcup_{(\FF,c)} V^{\FF,c} = \Spec R_B \ .$$
This proves that $A^B \overset{\sim}{\rightarrow} \Spec R_B$ is an isomorphism. The next statements follow by transferring information along this isomorphism: $\Spec R_B$ is clearly affine, and $A^B$, defined by the complement of divisors in a blow-up, is  smooth over $\Z$. 

For the last statement concerning the product structure of the divisors $D_I$,  we refer to theorem \ref{thm: recursive}.  Note that $A^B = P^B \backslash Y$, where $Y$ is the strict transform of the zero locus of 
 the polynomials $\phi_B(\underline{n}) = \prod_{J\in B} \alpha^{n_J}_J$, where $\underline{n}= (n_J)_{J\in B}$ are integers $n_J \geq 0$.  They satisfy the same factorisation properties as graph polynomials which are used in the proof
 of theorem $\ref{thm: recursive}$, namely 
 $$\phi_B(\underline{n}) = \phi_{B^I}(\underline{n}') \phi_{B_I} (\underline{n}'') +R$$
 for every $I \in B$,   where $R$ is of higher degree in the variables $\alpha_i, i\in I$ than $\phi_{B^I}$. 

\section{Motives of graphs  with kinematics} \label{sect: SectionGraphMotives}

We define the  motive (or rather, its image in a category of realisations) of a Feynman graph  by  applying the linear blow-up construction of the previous section to the set of \motic subgraphs of a Feynman graph. In the  case where the graph has no kinematic dependence, and is primitive log-divergent, this retrieves the definition of graph motive due to Bloch-Esnault-Kreimer \cite{BEK}.

\subsection{Orders of vanishing} Let $G$ be a Feynman graph, and consider the projective space $\Pro^{E_G}$.
There is a bijection between coordinate linear subspaces
$$\Pro^{E_G} \supset L_I   \hbox{ for } I \subsetneq E_G \quad \longleftrightarrow\quad  \hbox{edge subgraphs } I \subsetneq G\ .$$
Given a homogenous polynomial $P$ in $\Z[\alpha_e, e\in E_G]$, let 
$$v_I (P) =\hbox{order of vanishing of } P \hbox{ along } L_I\ .$$

\begin{lem} Assume generic kinematics $(\ref{eqn: genericmomenta})$ and $(\ref{eqn: genericmassmomenta})$. Let  $\gamma\subset E_G$. Then 
\begin{equation} \label{eqn: valuationforPsi}  v_{\gamma} (\Psi_G) = h_{\gamma} \ , \qquad \qquad \qquad
\end{equation} 
\begin{equation} \label{eqn: valuationformula}
\qquad \qquad \qquad v_{\gamma} (\Xi_G) = \begin{cases} h_{\gamma} + 1  \qquad  \hbox{ if } \gamma \hbox{ is } m.m.  \\ h_{\gamma} \,\,\,\,\quad  \qquad  \hbox{ if } \gamma \hbox{ is  not } m.m. \end{cases} 
\end{equation}
\end{lem} 
\begin{proof}
This follows from  the factorisations $(\ref{eqn: XiUVfact})$ and $(\ref{eqn: XiIRfact})$,  the degree formulae for graph polynomials $(\ref{eqn: degreesofPsiandPhi})$, and lemmas  \ref{lem: PsiGvanishing} and \ref{lem:  Xivanishing} which assert that $\Psi_{G/\gamma}$ and $\Xi_{G/\gamma}(m,q)$ are non-zero, via equation $(\ref{eqn:  mmequivalentXi})$.
\end{proof}

The principle motivation for the definition of \motic subgraphs comes from  the following proposition.  We again assume generic kinematics  $(\ref{eqn: genericmomenta})$, $(\ref{eqn: genericmassmomenta})$.

\begin{prop}  \label{prop: moticasvaluation} An edge subgraph $\Gamma \subset E_G$ is \motic if and only if 
$$v_{\gamma} (\Xi_G) < v_{\Gamma}(\Xi_G)$$
for all strict edge subgraphs $\gamma \subsetneq \Gamma$. In particular,   $v_{\Gamma}(\Xi_G) >0$ if $\Gamma$ is \motic\!\!. 
\end{prop}

\begin{proof}  Let $\gamma \subsetneq  \Gamma$ be a strict edge subgraph.
Let $\mathbb{I}_{\gamma,G}$ equal $1$ if $\gamma$ is mass-momentum spanning in $G$ and $0$ otherwise. 
We have $h_{\Gamma} \geq h_{\gamma}$.

 Equation $(\ref{eqn: valuationformula})$ implies that
\begin{equation} \label{eqn:  vdifference}
v_{\Gamma}(\Xi_G) - v_{\gamma}(\Xi_G) = (h_{\Gamma}- h_{\gamma}) + (\mathbb{I}_{\Gamma, G} - \mathbb{I}_{\gamma, G})\ .
\end{equation}
Furthermore,  $\mathbb{I}_{\Gamma, G} \geq \mathbb{I}_{\gamma, G}$, because if $\gamma$ is mass-momentum spanning in $G$, then so too is $\Gamma$ by lemma $\ref{lem: propertiesofmm}$ $(i)$. 
Thus  $(\ref{eqn:  vdifference})$ is strictly positive for all strict edge subgraphs $\gamma \subset \Gamma$
  if and only if $h_{\gamma} < h_{\Gamma}$ whenever  $ \mathbb{I}_{\gamma, G}=\mathbb{I}_{\Gamma, G}.$ By lemma $\ref{lem: propertiesofmm}$ $(i)$, this is precisely the set of mass-momentum spanning  subgraphs $\gamma\subsetneq \Gamma$. 
\end{proof}

\subsection{Graph hypersurfaces and the motive}
Let $G$ be a graph of type $(Q,M)$. We shall construct various families of schemes over the space of kinematics $K_{Q,M}$, defined in  $(\ref{eqn:  KQMdefn})$. 
In order to lighten the notation, we shall abusively write $\Pro^{E_G}$,  $L_I$, and so on,   to denote the base change of the schemes defined in  the previous section from $\Spec \, \Z$ to $K_{Q,M}$.  
Define  graph hypersurfaces 
$$X_{\Psi_G} = V(\Psi_G) \subset \Pro^{E_G} 
\qquad \hbox{ and } \qquad  X_{\Xi_G(q,m)} = V(\Xi_{G}(q,m)) \subset \Pro^{E_G}\ . $$
These are to be viewed as families of hypersurfaces over $K_{Q,M}$. The former were considered in \cite{BEK} when $Q=M=0$.   Note that the intersection $X_{\Psi_G} \cap X_{\Xi_G(q,m)}$ is given by the zero locus $V(\Phi_G(q)) $ of the second Symanzik polynomial, by $(\ref{eqn:  Xidefn}).$

 If $\Gamma\subset G$ is \motic\!\!,  then by  proposition $\ref{prop: moticasvaluation}$, $v_{\Gamma}(\Xi_G(q,m))>0$ and so the  linear subspace $L_{\Gamma}$ is contained in the graph hypersurface on each fiber, i.e., 
$$L_{\Gamma}  \subset X_{\Xi_G(q,m)}\ .$$
If $\Gamma$ is not $m.m.$ in $G$, then we also have $L_{\Gamma}  \subset X_{\Psi_{G}}$ by $(\ref{eqn: valuationforPsi})$ since $h_{\Gamma} >0$.  Thus the loci $L_{\Gamma}$
meet both the boundary of  the chain of integration $\sigma$,  and the singularities of the Feynman integrand $(\ref{eqn:  omegaGdef})$, causing potential divergences\footnote{As suggested by T. Damour, such loci could be termed \emph{problemotic}.}.

\begin{defn} Let $G$ be a  \motic Feynman graph of type $(Q,M)$.  Define
$$\pi_G: P^G = P^{B_G}\To \Pro^{E_G}$$
where $B_G=\{\Gamma \subsetneq G \hbox{ \motic\!\!}\}$, which  is stable under unions by theorem $\ref{thm: moticproperties}$ $(iii)$. 
\end{defn}

Let us define  $X_G \subset \Pro^{E_G}$ (viewed as a family over $K_{Q,M}$) to be 
\begin{equation}\label{eqn:  XGdefn}
X_G = \begin{cases}  X_{\Psi_G} \cup X_{\Xi_G(q,m)} \quad \hbox{ if } G \hbox{ has non-trivial kinematics } \\
  X_{\Psi_G}  \qquad \qquad \qquad \hbox{ if } G \hbox{ has no masses or momenta } 
\end{cases}  \ .
\end{equation}
Because of exceptional cases when the Feynman integrand $(\ref{eqn:  omegaGdef})$
has no term $\Psi_G$ in the denominator (i.e., when $N_G \geq  (h_G+1)d/2)$), we can also define
\begin{equation}\label{eqn:  X'Gdefn}
X'_G = X_{\Xi_G(q,m)}\ .
\end{equation}
Denote the strict transforms of $X_{\Psi_G}, X_{\Xi_G(q,m)}, X_G, X'_G$  in $P^G $ by 
$$Y_{\Psi_G} \quad  , \quad  Y_{\Xi_G(q,m)} \quad , \quad Y_G \quad , \quad Y'_G $$
respectively. Note that if $G$ has connected components $G_1,\ldots, G_n$, then 
$Y_{\Psi_G} = \bigcup_{i=1}^n Y_{\Psi_{G_i}}$. 
Recall that $D\subset P^G$ is the divisor  defined by $(\ref{eqn: Ddivisordefn})$, base-changed from $\Spec \Z$ to $K_{Q,M}$.
Denote the canonical projection by  $\pi_G : \fP^G \rightarrow \Pro^{E_G} $. 

\begin{defn} \label{defnmotG} Let $G$ be a \motic Feynman graph of type $(Q,M)$. Let
$$\mot_G = H^{N_G-1} ( P^G \backslash Y_G , D  \backslash( D \cap Y_G))_{/S} \ .$$
It is a triple $\mot_G = ((\mot_G)_B, (\mot_G)_{dR},  c)$ in a category $\HH(S)$,
for some Zariski-open $S \subset K_{Q,M}$, defined immediately below,
via the construction of \cite{NotesMot} \S 10.2.  There is also a variant
$$\mot'_G = H^{N_G-1} ( P^G \backslash Y'_G , D  \backslash( D \cap Y'_G))_{/S} \ .$$
\end{defn} 

\subsection{Reminders from \cite{NotesMot}} \label{sectRemindersHH}
Let $S$ be a smooth geometrically connected scheme over $\Q$. Then $\HH(S)$ is the category of triples
$(\V_{B}, \V_{dR}, c)$ where
$\V_B$ is a local system of $\Q$-vector spaces on $S(\C)$, and $\V_{dR}$ is an algebraic vector bundle on $S$
equipped with an integrable connection $\nabla$ with regular singularities at infinity. These are equipped with  a weight  (and for $\V_{dR}$, a Hodge) filtration satisfying the conditions of 
\S7.2 of \cite{NotesMot}, and $c: \V^{\mathrm{an}}_{dR} \overset{\sim}{\rightarrow} \V_B \otimes_{\Q} \Or_{S^{\mathrm{an}}}$ is an isomorphism of analytic vector bundles with connection, respecting the weight filtrations.

\subsection{Recursive structure} Let $D\subset P^G$  denote the divisor  defined by $(\ref{eqn: Ddivisordefn})$. Its irreducible components are
 $D_e$, for $e\in E_G$ an edge, and $D_{\gamma}$ for $\gamma \subsetneq E_G$ a strict \motic subgraph of $G$.  All schemes are viewed over $K_{Q,M}$ where $G$ is of type  $(Q,M)$. 

\begin{thm}\label{thm: recursive} For every edge $e\in E_G$ and $\gamma \subsetneq E_G$ \motic\!\!, we have
\begin{eqnarray}  \label{eqn: Drecursive}
 D_{\gamma}& =&  P^{\gamma} \times P^{G/\gamma}  \ , \\
D_e  &= &  \{pt\} \times P^{G/e}  \nonumber
\end{eqnarray}
and the  strict transforms of the graph hypersurfaces satisfy:
\begin{eqnarray} \label{eqn: Yrecursive}
Y_G \cap D_{\gamma} &  = &  \big(Y_{\gamma} \times P^{G/\gamma} \big)  \cup   \big(P^{\gamma} \times Y_{G/\gamma} \big)\ ,    \\ 
Y_G \cap D_e   &= &  \{pt\} \times Y_{G/ e}  \ .   \nonumber \end{eqnarray}
Note that at most one of $\gamma$ and  $G/\gamma$ has non-trivial momenta and masses. Here, $\{pt\}$ means the family of  points $\Spec \Z\times_{\Spec \Z }\Spec K_{Q,M}$ over $\Spec K_{Q,M}$. 
\end{thm}

\begin{proof}  Recall that $P^G = P^{B_G}$, where $B_G$ consists of the set of strict \motic subgraphs of $G$. 
Let $\gamma \subsetneq  E_G$ be \motic\!\!. By theorem $\ref{prop: PBstructure}$, and the notation preceding it, 
 $$D_{\gamma} = P^{(B_G)^{\gamma}} \times  P^{(B_G)_{\gamma}}\ .$$
By remark $\ref{rem: moticintrinsic}$,  ${(B_G)^{\gamma}}= B_{\gamma}$ since a subgraph of $\gamma$ is \motic if and only if it is \motic in $G$. 
Finally,  $(B_G)_{\gamma}= B_{G/\gamma}$ by theorem $\ref{thm: moticproperties}$  $(i)$ and $(ii)$.  This proves the first line of $(\ref{eqn: Drecursive})$. 
Now consider an affine chart $\A^{\FF,c}$ defined in \S\ref{sect:  localcoords}, where the flag $\FF$ contains $\gamma$ (otherwise $D_{\gamma}$ does not meet $\A^{\FF,c}$, by lemma \ref{lem: Flagavoidance}). In the coordinates  $(\ref{eqn: localcoords})$,  let  $\beta=0$ denote the equation of $D_{\gamma} \cap \A^{\FF,c}$ in $ \A^{\FF,c}$. Write $\pi_G^* \Psi_G$ in these coordinates  and apply $(\ref{eqn: IRfactPhi})$ to obtain 
$$\pi_G^* \Psi_G = \beta^{h_{\gamma}} (\pi_{\gamma}^* \Psi_{\gamma})  (\pi_{G/\gamma}^* \Psi_{G/\gamma} ) + O(\beta^{h_{\gamma}+1})\ ,$$
where the right-hand side is written using  $\A^{\FF,c} \cong  \A^1 \times \A^{\FF^{\gamma}, c^{\gamma}}  \times \A^{\FF_{\gamma}, c_{\gamma}}$
and where the coordinate on the  component $\A^1$ is $\beta$ (see the discussion following   $(\ref{eqn: isomDIr})$). This proves that 
$Y_{\Psi_G} \cap \fD_{\gamma} =( Y_{\Psi_{\gamma}} \times \fP^{G/\gamma} )\cup  (\fP^{\gamma} \times  Y_{\Psi_{G/\gamma}})$. 
A similar argument using  formulae $(\ref{eqn: XiUVfact})$ and $(\ref{eqn: XiIRfact})$ proves that 
\begin{eqnarray} 
Y_{\Xi_G(q,m)} \cap \fD_{\gamma} & = &( Y_{\Xi_{\gamma}(q,m)} \times \fP^{G/\gamma} )\cup  (\fP^{\gamma} \times  Y_{\Psi_{G/\gamma}}) \qquad \hbox{  when 
}  \gamma   \hbox{ is } m.m.  \nonumber \\
Y_{\Xi_G(q,m)} \cap \fD_{\gamma} & =& ( Y_{\Psi_{\gamma}} \times \fP^{G/\gamma} )\cup  (\fP^{\gamma} \times  Y_{\Xi_{G/\gamma}(q,m)})  \qquad \gamma \hbox{  not } m.m.  \nonumber 
\end{eqnarray} 
Since $Y_G = Y_{\Psi_G} \cup Y_{\Xi_G(q,m)}$, this proves the first equation of  $(\ref{eqn: Yrecursive}).$

Now, by theorem  $\ref{prop: PBstructure}$, we have 
$$D_e = \{pt\} \times P^{(B_G)_e}\ .$$
We must show that $(B_G)_e = B_{G/e}$.  The left-hand side is given by the   sets  of edges
 $\gamma \backslash (\gamma \cap \{ e\})$  for  \motic 
 $\gamma \subsetneq E_G$, and this is the set of edges of the subgraph 
    $( \gamma \cup e)/e\subset E_{G/e}$. By theorem \ref{thm: moticproperties} $(i)$ and $(iv)$,   such a graph is \motic in $G/e$,  and
 every \motic subgraph of $G/e$ arises in this way. (Note that the failure of injectivity stated in the last line of theorem \ref{thm: moticproperties} $(iv)$ poses no problem, since we are only concerned with  $\gamma \cup e$). This proves the second line of $(\ref{eqn: Drecursive})$. If $e$ is not a tadpole, it follows from corollary \ref{cor: Restrictiontoalpha=0} that   $$X_{\Psi_G} \cap L_e = X_{\Psi_{G/ e}}\quad  \hbox{ and } \quad X_{\Xi_G(q,m)} \cap L_e = X_{\Xi_{G/ e}(q,m)}\ ,$$
 and we deduce the corresponding statements for their strict transforms.
  If $e$ is a tadpole it is \motic\!\!, and we are reduced to the previous case, since $P^e = \{pt\}$.
   \end{proof}

\begin{rem} 
The analogous statement of theorem \ref{thm: recursive} for $Y'_G = Y_{\Xi_G(q,m)}$ is 
\begin{eqnarray}\label{eqn: Y'recursive}
Y'_G \cap D_e   &= &  \{pt\} \times Y'_{G/ e}   \\
Y'_G \cap D_{\gamma} &  = &  \big(Y'_{\gamma} \times P^{G/\gamma} \big)  \cup   \big(P^{\gamma} \times Y_{G/\gamma} \big) \nonumber \quad \hbox{if } \gamma  \hbox{ is } m.m. \nonumber \\
&  = &  \big(Y_{\gamma} \times P^{G/\gamma} \big)  \cup   \big(P^{\gamma} \times Y'_{G/\gamma} \big) \nonumber \quad \hbox{if } \gamma \hbox{ is not } m.m.
\end{eqnarray}
Note that every time a $Y_{H}$ (with no prime) occurs in these formulae, it is because the graph $H$ (either $\gamma$ or $G/\gamma$) has no masses or momentum dependence.
\end{rem}
By induction, the theorem implies  that the intersection of the strict transform $Y_G$ of the graph hypersurface with
a facet of $D$ of codimension $n$ is of the form 
$$\bigcup_i P^{\gamma_1} \times \ldots \times P^{\gamma_{i-1}} \times Y_{\gamma_i} \times P^{\gamma_{i+1}} \times \ldots \times P^{\gamma_n}$$
where the $\gamma_i$ are quotients of \motic subgraphs of $G$. In fact, $\gamma_1\otimes \ldots \otimes \gamma_n$ is a descendant of $G$ according to  definition \ref{defnDescendants}. 

\subsection{Feynman polytope and Betti class}
Let $G$ be a  Feynman graph of type $(Q,M)$. Define the Feynman polytope, following  \S\ref{sect: Bpolytope},  to be 
$$\widetilde{\sigma}_G := \widetilde{\sigma}_{B_G}  \times U^{\mathrm{gen}}_{Q,M} \subset P^G(\R)\ .$$
It is a constant family of compact manifolds with corners over the locus $U^{\gen}_{Q,M} \subset K^{gen}_{Q,M}(\C)$ where masses and momenta have positive real parts (\S\ref{sect:  GenKin}).

\begin{thm}  \label{thmsigmaGavoidsY} We have 
$\widetilde{\sigma}_G \cap Y_G(\C) = \emptyset$.
A fortiori,  $\widetilde{\sigma}_G \cap Y'_G (\C)= \emptyset$.
\end{thm}
\begin{proof} Since the polytope $\widetilde{\sigma}_G$ is stratified, it suffices to show that $Y_G(\C)$ does not meet any open stratum. We shall do this by induction.
We have an isomorphism
$$   \big(P^G \backslash D \big)  \overset{\sim}{\To} \big( \Pro^{E_G} \backslash \bigcup_{e\in E_G} L_e\big) $$
which induces  a homeomorphism from the big open stratum $\overset{\circ}{\sigma}_G\subset \widetilde{\sigma}_G$ to $\overset{\circ}{\sigma} \times U^{\gen}_{Q,M}$, the open coordinate simplex.  It also sends $Y_G  \backslash (Y_G \cap D)$ to $X_G$. Thus
$$Y_G(\C) \cap \overset{\circ}{\sigma}_G \cong X_G (\C) \cap  ( \overset{\circ}{\sigma} \times U^{\gen}_{Q,M})\ ,$$
and it suffices to show that  $X_{\Psi_G}(\C) \cap (  \overset{\circ}{\sigma} \times U^{\gen}_{Q,M})$ is  empty, and likewise that  $X_{\Xi_G(q,m)}(\C) \cap (\overset{\circ}{\sigma} \times U^{\gen}_{Q,M})$ is non-empty, 
in the case when $G$ has non-trivial kinematics.  Since $\Psi_G$
is non-zero by lemma \ref{lem: PsiGvanishing} and only has positive coefficients, it is clear that
$\Psi_G>0$  when all  $\alpha_i >0.$
Similarly, if $G$ has non-trivial kinematics then 
$\Xi_G(q,m)$ is non-zero  by lemma $\ref{lem:  Xivanishing}$ since $(\ref{eqn: genericmomenta})$ and $(\ref{eqn: genericmassmomenta})$ is automatically satisfied if $(\ref{Ugendef})$   holds. Using the explicit expression   $(\ref{eqn:  Xidefn})$, we have
$ \Real (\Xi_G(q,m))>0$  when all $\alpha_i>0$  and $(q,m) \in U^{\mathrm{gen}}_{Q,M}   $.
This proves  that  $ X_G (\C) \cap (\overset{\circ}{\sigma} \times U^{\gen}_{Q,M}) =\emptyset$.  
Now  theorem $\ref{thm: recursive}$  implies that the facets of $\widetilde{\sigma}_G$ satisfy
$$\widetilde{\sigma}_G \cap D_e (\C) = \{pt\}\times \widetilde{\sigma}_{G/e}  \qquad \hbox{ and } \qquad
\widetilde{\sigma}_G \cap D_{\gamma}(\C) = \widetilde{\sigma}_{\gamma}  \times \widetilde{\sigma}_{G/\gamma} \ .  $$
Using the recursive structure  $(\ref{eqn: Yrecursive})$, we are  reduced to proving an identical statement
for quotients of \motic subgraphs of $G$, and proceed by induction by decreasing dimension of the strata. 
\end{proof} 

\begin{defn}  \label{defnsigmaGclass} Let us write $\sigma_G\subset (P^G \backslash Y_G)(\R)$ (with no tilde) for the intersection
$$\sigma_G = \widetilde{\sigma}_G \cap (P^G \backslash Y_G)(\C) \ .$$
It follows from the previous theorem that  $\sigma_G$ is homeomorphic to $\widetilde{\sigma}_G$, and   its boundary is  contained in $D\cap Y_G(\C)$. It therefore   defines a canonical Betti class
\begin{equation}
[\sigma_G] \ \in\   \Gamma \big(U^{\mathrm{gen}}_{Q,M}, H_{N_G-1} (P^G \backslash Y_G, D \backslash (D \cap Y_G))\big) \ .
\end{equation} 
which we view as a local section of the dual Betti local system 
$$[\sigma_G ] \in \Gamma \big(U^{\mathrm{gen}}_{Q,M}, (\mot_G)_B^{\vee}\big)\ .$$
Likewise  we can replace  $Y_G$  by $Y'_G$ in the case when $N_G\geq( h_G+1)d/2$.
\end{defn}
One reason why we emphasise the action of the de Rham Galois group (as opposed to the Betti Galois group) on motivic Feynman integrals is because  the Betti class is uniformly defined for all graphs and we wish to keep it fixed. Furthermore, we prefer to compute with differential forms rather than homology cycles. 

\subsection{De Rham class and power-counting} Let $G$ be as above.
  The pull-back  of the  form $(\ref{eqn:  omegaGdef})$   along the map
  $\pi_G: P^G \backslash D \backslash Y_G \To \Pro^{E_G}\backslash X_G$  
 $$\pi_G^* (\omega_G(q,m)) \in \Omega^{N_G-1}(P^G \backslash D \backslash Y_G;k_{Q,M})$$
 may have poles along the exceptional divisors $D_{\gamma}\subset P^G$ for $\gamma$ \motic\!\!. 
 
 \begin{defn}The \emph{superficial degree of divergence} of a graph $G$ is
 \begin{equation} 
\sd_G =d \,h_G /2-  N_G \ .
 \end{equation}
 This is an integer, since $d$, the dimension of spacetime,  was assumed to be even.
 \end{defn} 
 The following lemma gives necessary and sufficient conditions for the convergence of the Feynman integral in the enlarged Euclidean region in terms of the superficial degrees of divergences of \motic subgraphs. 
  \begin{lem}\label{lem: valuationsconvergence}
 Let $\gamma \subsetneq E_G$ be \motic\!\!. Then  $\pi_G^* (\omega_G(q,m))$ has a pole along $D_{\gamma}$ of order given by the following formula:
 \begin{equation}\label{eqn: vgamma}
 -v_{\gamma}(\pi_G^* (\omega_G(q,m)))=       \begin{cases}1+ \sd_{\gamma} \qquad \qquad \qquad\hbox{ if } \gamma \hbox{ is not } m.m.  \\ 
1+ \sd_{\gamma}-\sd_G  \qquad  \,\quad \hbox{ if } \gamma \hbox{ is } m.m. 
 \end{cases} \ .
 \end{equation}
 It has no poles along any  divisors of the form $D_e$, where $e\in E_G$ is an edge which is not \motic\!\!. It follows that  the Feynman integral $I_G(q,m)$ $(\ref{eqn: projectiveFeynmanint})$  is absolutely  convergent  in the region $U^{\mathrm{gen}}_{Q,M}$   if and only if 
 \begin{eqnarray} \label{eqn: convineq}
 \sd_{\gamma} &< &0 \qquad \hbox{ for all } \gamma \subsetneq E_G \hbox{ \motic and not } m.m.  \\
 \sd_{\gamma} &< &\sd_G \quad \hbox{ for all } \gamma \subsetneq E_G  \hbox{ \motic and  } m.m. \nonumber   \ .
 \end{eqnarray} 
  \end{lem}
 \begin{proof}
Recall $(\ref{eqn:  omegaGdef})$. From equations $(\ref{eqn: valuationformula})$ and $(\ref{eqn: valuationforPsi})$, we have
  $$-v_{\gamma}\, \Big(  {1\over \Psi_G^{d/2}} \Big( { \Psi_G \over \Xi_G(m,q)} \Big)^{N_G - h_G d/2} \Big)=h_{\gamma} d/2 - \mathbb{I}_{\gamma} \,\sd_G   $$
where $\mathbb{I}_{\gamma}$ is $1$ if $\gamma$ is $m.m.$, and $0$ otherwise. Now consider an affine chart $\A^{\FF,c}$ where 
$$\FF: \qquad \emptyset = I_0 \subsetneq I_1 \subsetneq I_2 \subsetneq \ldots \subsetneq I_k \subsetneq I_{k+1} = E_G$$ is a flag containing $\gamma=I_r$   and $D_{\gamma}$ is given  by the equation $\beta_{j_r}=0$. The chart $\A^{\FF,c}$ lies over the affine $\alpha_{j_{k+1}}=1$.
In  the local coordinates $(\ref{eqn: localcoords})$ we find that 
$$\pi_G^* (\prod_{i\neq j_{k+1}}  d\alpha_{i}) = \beta_{j_1}^{|I_1|- 1}\beta_{j_2}^{|I_2|- 1|} \ldots \beta_{j_k}^{|I_k|-1}\prod_{i\neq j_{k+1}}  d\beta_{i} $$
 vanishes along $\beta_{j_r}=0$ to order $1- |I_r|=1- N_{\gamma}$. Thus 
$$-v_{\gamma}( \pi^*(\omega_G(q,m))) =  1 - N_{\gamma} + h_{\gamma} d/2 - \mathbb{I}_{\gamma} \,\sd_G  $$
  which proves $(\ref{eqn: vgamma})$. For the second part, observe  that $\Psi_G$ and $\Xi_G(q,m)$ vanish along  a coordinate hyperplane $L_e$
  if and only if $e$ is \motic\!\!. Therefore $\omega_G(q,m)$ has
  no pole along $D_e$  if $e$ is not \motic\!\!.
    
  For the last part, suppose that $(\ref{eqn: convineq})$ holds. Then $\omega_G(q,m)$ is continuous and has no poles on the  domain $\sigma_G$, which has compact fibers over $U^{\gen}_{Q,M}$. Its integral   $I_G(q,m)$ is therefore absolutely convergent. Conversely, suppose that  $I_G(q,m)$ is absolutely convergent for some $(q,m)$ in the Euclidean region $K^{\mathrm{gen}}_{Q,M}(\R)$ defined in \S\ref{sect:  GenKin}, where  all momenta $q_i$ and masses $m_e$ are real.   Then $\omega_G(q,m)$ is strictly positive  on $\sigma_G$, and so for any  subset  $U\subset \sigma_G$ we have
  $$  \int_U \omega_G(q,m) \leq I_G(q,m) <\infty\ .$$
   If $\omega_G(q,m)$ had a pole 
  along a  boundary divisor $D_{\gamma}$,  it follows from positivity  by taking $U$ to be a neighbourhood of $D_{\gamma} \cap \sigma_G$   that the left-hand side is infinite. Therefore $\omega_G(q,m)$  has no poles along any $D_{\gamma}$ and hence $(\ref{eqn: convineq})$ holds.
  \end{proof} 
 Note that one can interpret the quantity $\sd_{\gamma} - \sd_G$ as  $-\sd_{G/\gamma}$.
 We say that a Feyman integral is \emph{convergent} if the  conditions  $(\ref{eqn: convineq})$ are satisfied.
In this case, 
  $$\pi_G^* \, \omega_G(q,m)  \in \Omega^{N_G-1}( (P^G \backslash Y_G)_*;k_{Q,M})$$
  where $(P^G \backslash Y_G)_*$ denotes the generic fiber of $P^G \backslash Y_G$ over $K_{Q,M}$. This  in turn defines a relative de Rham cohomology class
  at the generic point 
  $$[\pi_G^*\, \omega_G(q,m) ] \in \Gamma(\Spec k_{Q,M},  (\mot_G)_{dR}) \ . $$
     In the case when $N_G \geq (h_G+1)d/2$, we can replace $Y_G$ by $Y'_G$ in the above.

\begin{rem}
The formulae for the degrees of divergence of a Feynman integral with respect to a subgraph are due to Weinberg \cite{Wein}, and are  known as power-counting. Finding a minimal class of subgraphs which give  necessary conditions for convergence is more subtle  since the  inequalities $(\ref{eqn: convineq})$ are not  independent.   This  problem was studied by  Speer \cite{Speer, SpeerWestwater} for generic
Feynman integrals.
\end{rem}  

\subsection{Convergence of `global periods':  integrals with numerators}
Gauge theories  can  produce Feynman integrals 
with numerators \cite{Corolla}:
  \begin{equation} \label{eqn:  ExFeynmInt}
  \int_{\sigma} \omega(q,m) \quad \hbox{ where }  \quad  \omega(q,m)= { P \,   \Omega_G \over \Psi_G^A\,  \Xi_G(q,m)^B}\ ,
  \end{equation} 
  where $A,B$ are integers, and $P\in \Q[\alpha_e, e\in E_G]$ is homogeneous of degree
  \begin{equation} \label{eqn: degreeforP}
  \deg\, P = A h_G + B(h_G+1) - N_G
  \end{equation} 
to ensure that $\omega(q,m)$ is homogeneous of degree zero. We take $B=0$ if $G$ is of type $(0,0)$, in which case $\Xi_G(q,m)$ vanishes. 
\begin{lem}
Suppose that \begin{equation} \label{eqn: convforP}
v_{\gamma}(P) \geq  Ah_{\gamma} + B(h_{\gamma} + \mathbb{I}_{\gamma})  - N_{\gamma}+1
\end{equation}
for all \motic subgraphs $\gamma\subsetneq E_G$, where $\mathbb{I}_{\gamma}$ is $0$ if $\gamma$ is not $m.m.$, and $1$ if it is $m.m.$.
 Then    $\pi^*_G(\omega(q,m))$ has no poles along $D$ for $(q, m) \in U^{\mathrm{gen}}_{Q,M}$. Therefore 
 $$\pi^*_G(\omega(q,m)) \in \Omega^{N_G-1}(P_G \backslash Y_G ; k_{Q,M})$$
  and the integral $(\ref{eqn:  ExFeynmInt})$ is convergent for all $(q, m) \in U^{\mathrm{gen}}_{Q,M}$.
\end{lem}
\begin{proof}
Similar  to lemma $\ref{lem: valuationsconvergence}$.
\end{proof} 
Thus in this case also
  $[\pi_G^*\, \omega_G(q,m) ] \in  \Gamma(\Spec k_{Q,M}, (\mot_G)_{dR})$.
  The integral $(\ref{eqn:  ExFeynmInt})$ is an example of what we shall call a \emph{Feynman period}.

\begin{example}  Consider the banana graph $G$ with three edges,  no external momenta and no masses. 
  Its graph polynomial is $\Psi_G = \alpha_1 \alpha_2 + \alpha_1 \alpha_3 + \alpha_2 \alpha_3$, and its Feynman amplitude is divergent.
However, the previous lemma provides examples of  Feynman periods such as
$$\int_{\sigma}   { \alpha_1\alpha_2\alpha_3 \over \Psi^3_G} \Omega_G=  \int_{\alpha_1,\alpha_2\geq0} {\alpha_1 \alpha_2 \over (\alpha_1 \alpha_2 + \alpha_1 + \alpha_2)^3} = {1 \over 2} \ . $$
One easily shows (see section $\S\ref{sect: ConstantCosmic}$) that all periods of this graph are rational. 
\end{example}

 \section{Motivic Feynman amplitudes} \label{sect: MotAmp}

Armed with the definition of the motive of a Feynman graph, we can give the definition of the motivic Feynman amplitude and draw some first consequences.

\subsection{Reminders on motivic periods from \cite{NotesMot}}  
Let $S\subset K_{Q,M}$ be Zariski open as in \S\ref{sectRemindersHH},  let $s \in S(\C)$, and denote by   $\Spec k_{Q,M}$ the generic point of $S$.

Then $\HH(S)$ has two fiber functors  (\cite{NotesMot} \S7.2.1)
\begin{eqnarray}
\omega^{\gen}_{dR}  : \HH(S)  & \To &  \mathrm{Vec}_{k_{Q,M}} \nonumber  \\
\omega^s_{B}: \HH(S) & \To&  \mathrm{Vec}_{\Q} \nonumber 
\end{eqnarray} 
where  $\omega^{\gen}_{dR}$  is the fiber of $\V_{dR}$ at the generic point $\Spec k_{Q,M}$ of $K_{Q,M}$,
and $\omega_{B,s}(\V_{B}, \V_{dR}, c) = (\V_B)_s$ is the fiber at $s$. The ring of $\HH(S)$-periods
$$\Pe^{\mm,s, \gen}_{\HH(S)} $$
is  the space  generated by the matrix coefficients (\cite{NotesMot}, \S2.2, \S8.2) of the form  $[V, \sigma, \omega]^{\mm}$ where $V=(V_{B}, V_{dR}, c)$ is an object
of $\HH(S)$,  $\sigma \in \omega_s( V_B)^{\vee}$,  and 
$\omega \in  \omega^{\gen}_{dR}(V)$.  Now suppose that $s_1,s_2 \in U^{\gen}_{Q,M} \cap S(\C)$ are two points  in the region of generic kinematics. 
A path $\gamma\subset U^{\gen}_{Q,M} \cap S(\C)$ from $s_1 $ to $s_2$ yields an isomorphism
$$\Pe^{\mm,s_1,\gen}_{\HH(S)} \cong \Pe^{\mm,s_2,\gen}_{\HH(S)}$$ 
by continuation along paths \cite{NotesMot}, (7.10)
and hence a canonical isomorphism
$$(\Pe^{\mm,s_1,\gen}_{\HH(S)})^{\pi_1(U^{\gen}_{Q,M} \cap S(\C),s_1)}  =  (\Pe^{\mm,s_2,\gen}_{\HH(S)})^{\pi_1(U^{\gen}_{Q,M} \cap S(\C),s_2)}  $$
where the action of the topological fundamental group is on the right,
and commutes with the action of the de Rham Tannaka group (resp. the de Rham coaction).
  If $S' \rightarrow S$ is a smooth morphism and $s'\in S'(\C)$ is in the pre-image of $s$, 
then the  pullback defines  a functor $\HH(S) \rightarrow \HH(S')$ and  hence a homomorphism $\Pe^{\mm,s,\gen}_{\HH(S)} \rightarrow \Pe^{\mm,s',\gen}_{\HH(S')}.$ By taking $\pi_1$-invariants, we can  move the complex point $s\in S(\C)$ to ensure that it lies in the image of $S'(\C)$, and hence take the limit. 

\begin{defn} For  $Q,M\geq 0 $,  following \cite{NotesMot} \S8.2, let   
$$\Pe^{\mm}_{Q,M} =  \varinjlim_S \, (\Pe^{\mm,s,\gen}_{\HH(S)})^{\pi_1(U^{\gen}_{Q,M} \cap S(\C),s)}$$
where the limit ranges over Zariski open  $S\subset K^{\gen}_{Q,M}$ which are defined over $\Q$ and $s\in
U^{\gen}_{Q,M} \cap S(\C)$, ordered with respect to inclusion. 
 \end{defn}

The point of this construction is that the Feynman motivic periods will have poles on the complement of  some unspecified Zariski-open set $S$, but will always be single-valued on the region $U^{\gen}_{Q,M}$ (extended `Euclidean sheet'). The ring $\Pe^{\mm}_{Q,M}$ captures exactly this property.

There is a corresponding ring of de Rham periods $\Pe^{\dR}_{Q,M} =  \varinjlim_S \, \Pe^{\dR,\gen}_{\HH(S)}$.  It  is generated by  equivalence classes of triples  $[V, f, \omega]^{\dR}$ where $f \in \omega^{\gen}_{dR}(V)^{\vee}$ and 
$\omega \in  \omega^{\gen}_{dR}(V)$ and $V$ is an object of $\HH(S)$ for some $S\subset K_{Q,M}$ as above. 

Recall that for $(Q,M)=(0,0)$ the ring $k_{0,0} =\Q$ and $\HH_{0,0} = \HH$ is the category of realisations over $\Q$ considered in \cite{NotesMot} \S2. Its ring of periods $\Pe^{\mm}_{0,0} = \Pe^{\mm}_{\HH}$ is simply the ring of  periods over $\Q$ considered in \cite{NotesMot} \S3.

\subsection{Motivic Feynman amplitudes and motivic Feynman periods}
We make frequent use of the definitions and constructions  from \cite{NotesMot}. 

\begin{defn} Let $G$ be a Feynman graph of type $(Q,M)$ and 
 $I_G(q,m)$  a convergent Feynman amplitude $(\ref{eqn: projectiveFeynmanint})$. Define the \emph{motivic Feynman amplitude} to be 
  $$I^{\mm}_G(q,m) = [ \mot_G, [\pi_G^* \,\omega_G(q,m)], [\sigma_G]]^{\mm}\ \qquad \in \qquad \Pe^{\mm}_{Q,M} \ . $$
  To check that this makes sense, note that the general theory implies that $\mot_G$ is an object of $\HH(S)$ for $S$ some Zariski-open set in $K^{\gen}_{Q,M}$. Since $[\sigma_G]$ is in a fact constant section of $\Gamma(U_{Q,M}^{\gen}\cap S(\C), (\mot_G)_B^{\vee})$,  the fundamental group of 
  $U_{Q,M}^{\gen}\cap S(\C)$ acts trivially upon $I^{\mm}_{G}(q,m)$. Pick any $s \in U_{Q,M}^{\gen}\cap S(\C)$. The element  $I^{\mm}_{G}(q,m)$  therefore can be viewed as an element of $\Pe^{\mm, s, \gen}_{\HH(S)}$, invariant
  under $\pi^{\mathrm{top}}_1(U_{Q,M}^{\gen}\cap S(\C))$. It can therefore be viewed as  an element of $\Pe^{\mm}_{Q,M}$.

  For all $(q,m) \in U^{\gen}_{Q,M}$, its period is given by the convergent Feynman integral
  $$\per (I^{\mm}_G(q,m) ) = I_G(q,m)\ .$$
  \end{defn} 
  \noindent  We can also define (for instance in the case $N_G\geq (h_G+1)d/2$)  a variant:
   $$I^{\mm}_G(q,m)'= [ \mot'_G, [\pi_G^*\, \omega_G(q,m)], [\sigma_G]]^{\mm}   \qquad \in \qquad  \Pe^{\mm}_{Q,M}\ .$$
  
  Since the action of the de Rham Galois group on  motivic Feynman integrals  will generate  new integrands, and since quantum field theory (e.g., gauge theories)
  naturally produce integrands with numerators, we make the following definition.
  
  \begin{defn} \label{defnmotFeynper}
   A \emph{motivic Feynman period} of type $(Q,M)$ is 
\begin{eqnarray}
I^{\mm}_G(\omega) = [ \mot_G, [\omega], [\sigma_G]]^{\mm}  \ \in  \  \Pe^{\mm}_{Q,M}
\end{eqnarray} 
where    $[\omega] \in   \omega_{dR}^{\gen}(\mot_G )$ is  any relative de Rham cohomology class at the generic point. Write
$$I_G(\omega)=\per(I^{\mm}_G(\omega)) \ ,$$
which we call a \emph{Feynman period}. It is a multivalued  meromorphic function on $K^{\gen}_{Q,M}(\C)$. 
Note that graphs for which the Feynman amplitude $\omega_G(q,m)$ does not converge
may still have non-trivial Feynman periods.  The notion of Feynman \emph{period} is therefore more general, but contains,  the notion of Feynman \emph{amplitude}.   Feynman periods may not necessarily always have a physical interpretation. 
  \end{defn}

\noindent One can consider a  variant for $N_G\geq  (h+1)d/2$, by replacing $\mot_G$ by $\mot'_G$.
 
 \begin{rem} Many, but not all,  cohomology classes in $  \omega_{dR}^{\gen} (\mot_G)$ are of the form  $[\pi^*_G(\omega(q,m))]$, where $\omega(q,m)$ is a Feynman integrand
 with numerator $(\ref{eqn:  ExFeynmInt})$. 
 \end{rem}

\subsection{Coaction formula and cosmic Galois group} By \cite{NotesMot} \S8.2, the rings of  periods considered above  admit a coaction 
$$\Delta : \Pe^{\mm}_{Q,M} \To \Pe^{\mm}_{Q,M}  \otimes_{k_{Q,M}}  \Pe^{\dR}_{Q,M}   \  .$$
Applied to a motivic Feynman period we have the formula \cite{NotesMot} $(2.3)$
\begin{equation} \label{eqn: CoactionImG}
\Delta I^{\mm}_G(\omega)   =  \sum_{e_i}   [ \mot_G, [\sigma_G] ,e_i  ]^{\mm}   \otimes [ \mot_G,  e^{\vee}_i ,[\omega] ]^{\dR} \\
\end{equation} 
and the corresponding variant in which $I$ and $\mot_G$ are replaced by $I'$ and $\mot_G'$.  In this formula,
 $e_i$ is a basis of $\omega^{\gen}_{dR} (\mot_G)$ and $e_i^{\vee}$ the dual basis.  The elements on the right-hand side of the tensor product  in $(\ref{eqn: CoactionImG})$ are again motivic Feynman periods
of $G$.

It is customary in the physics literature to restrict  to one-particle irreducible  graphs. 
Likewise, we shall restrict to  the larger class of  \motic graphs.  
\begin{defn}
The comodule of \emph{de Rham Feynman periods} of type $(Q,M)$ is  the  subspace
$\HF^{\dR}_{Q,M}  \subset \Pe^{\dR}_{Q,M}$
spanned by de Rham Feynman periods 
$$[ \mot_G, f,v ]^{\dR}\ ,  \qquad \hbox{ where } v\in  \omega^{\gen}_{dR}(\mot_G)  \hbox{ and } f \in  \omega^{\gen}_{dR}(\mot_G)^{\vee}\ ,$$
and $G$ is \motic of type $(Q,M)$. 
The coproduct is given by the formula 
$$\Delta [ \mot_G,  f,v ]    =  \sum_{e_i}   [ \mot_G,  f,e_i  ]^{\dR}\otimes [ \mot_G, e^{\vee}_i,v  ]^{\dR} \ , $$
where $e_i$ is a  basis of $\omega^{\gen}_{dR} (\mot_G)$ and $e_i^{\vee}$ the dual basis  as above. It is not an algebra since it is not closed under products when $(Q,M)\neq (0,0)$. 

The space of  \emph{motivic  Feynman periods} of type $(Q,M)$ is the subspace
$$\HF^{\mm}_{Q,M}  \subset \Pe^{\mm}_{Q,M}$$
spanned  by motivic Feynman periods 
$I^{\mm}_G(\omega),$ where $G$ is \motic of type $(Q,M)$.  It is a right comodule over $\HF^{\dR}_{Q,M}$. The coaction is given by equation $(\ref{eqn: CoactionImG})$:
\begin{equation} \label{eqn: DeltaonHF} 
\Delta: \HF^{\mm}_{Q,M} \To \HF^{\mm}_{Q,M}     \otimes_{k_{Q,M}} \HF^{\dR}_{Q,M} \ .
\end{equation} 
\end{defn}

Define the \emph{cosmic Galois group}  $C_{Q,M}$ for \emph{convergent} Feynman periods to be
\begin{equation} \label{defnCQM} 
C_{Q,M} = \Spec\, \HF^{\dR}_{Q,M} \ .
\end{equation}
It is an affine group scheme over $k_{Q,M}$, and  acts on $\omega^{\gen}_{dR}(\mot_G)$ for all $G$ \motic of type $(Q,M)$.  The coaction  $(\ref{eqn: DeltaonHF})$ is equivalent to 
a group action 
$$C_{Q,M} \times \HF^{\bullet}_{Q,M} \To \HF^{\bullet}_{Q,M} \qquad \hbox{ with } \bullet = \mm, \dR \ .$$

The  cosmic Galois group itself should be taken with a pinch of salt (in the same way as the Galois group of all algebraic numbers): it is an enormous pro-algebraic group scheme over $k_{Q,M}$, and in practice one is  interested in its finite-dimensional  representations.

\subsection{First applications}  \label{sect: Firstapplications} The notion of motivic Feynman amplitude leads to  a number of immediate consequences. We briefly mention just a few.

\begin{defn} The \emph{representation} associated to a motivic Feynman period  $I^{\mm}_G(\omega)$ is the representation of $C_{Q,M}$ which it generates:
$$V_G(\omega)\subset  \HF^{\mm}_{Q,M}\ . $$
The \emph{Galois conjugates} of a motivic Feynman period are the elements of $V_G(\omega)$, or equivalently, the elements of the comodule 
it generates under the coaction $(\ref{eqn: CoactionImG})$. 
\end{defn} 

\begin{enumerate}
  \setlength\itemsep{1em}
\item (Weight filtration). We can  speak of the weights of Feynman periods. Say that a motivic Feynman period $I^{\mm}_G(\omega)$ is \emph{of weight at most $n$}  if 
$$I^{\mm}_G(\omega) \in W_n \HF_{Q,M}^{\mm}$$
This holds in particular if $\omega \in W_n \omega^{\gen}_{dR} (\mot_G)$. Note that the weight is a filtration, not a grading, except in very special circumstances. We can say that a Feynman period is of weight $\leq n$ if it is the period $\per \,I^{\mm}_G(\omega) $ of a motivic Feynman period  $I^{\mm}_G(\omega)$ of weight $\leq n$.\footnote{The Hodge-theoretic weights here are double those commonly used in physics and the field of multiple zeta values and polylogarithms. For example, $\zeta(3)$ has weight  6 and not 3. If one wishes to divide the weight by two,  one will encounter  half-integral weights. The simplest example of this phenomenon occurs already for 
amplitudes which are given by  elliptic integrals.}

\item (Picard-Fuchs equations). There is an integrable   connection 
$$\nabla: \HF^{\mm}_{Q,M} \To  \HF^{\mm}_{Q,M} \otimes_{k_{Q,M}} \Omega^1_{k_{Q,M}} \ .$$
See \cite{NotesMot} \S7.4 for its compatibilities with the period homomorphism and coaction.
It  follows from general theory and the construction of $\mot_G$ that Feynman periods are solutions to differential equations of Fuchsian type.

\item (Invariants attached to $V_G(\omega)$). The representation $V_G(\omega)$ of a  Feynman period carries  a barrage of new information. This is discussed in \cite{NotesMot}. For example, a motivic Feynman period has a \emph{rank} (the dimension of $V_G(\omega)$); a  \emph{Hodge polynomial}; and a filtration by the \emph{unipotency degree} (this is the coradical filtration on $\HF_{Q,M}^{\mm}$ induced by the coaction $(\ref{eqn: DeltaonHF})$.).

\item (Mixed Tate amplitudes).     Call a motivic Feynman period   \emph{mixed Artin-Tate} if it is equivalent to a motivic period of a
variation of Hodge-Tate type (all Hodge numbers $h_{p,q}=0$ unless $p=q$). Call it \emph{mixed Tate} if it is equivalent to a motivic period of $H$ where  $\gr_{n}^W H$ is zero if $n$ is odd, and  a direct sum of constant Tate objects
$\Q(-k)$ if $k=2n$ is even. 
 The weight filtration is a grading on such motivic periods. 
This case is discussed  in \S\ref{sectsvandsymbols}. 

\end{enumerate}
 \begin{rem} The Betti realisation of $\mot_G$ is a local system on the complex points of some open   $S=K^{\gen}_{Q,M}\backslash L_G$,  where $L_G$ is a closed subscheme we call  the Landau variety\footnote{This is more complex \cite{BrFeyn} than the Landau variety as commonly understood in the 
 physics literature. The Landau equations in the classical  sense describe  a certain subset, but by no means all,  of the  components of  $L_G$.}
   of $G$. The fiber of this local system at a point $s$ carries an action of the fundamental group $\pi_1^{\mathrm{top}}( K^{\mathrm{gen}}_{Q,M}(\C) \backslash L_G (\C), s),$ corresponding to  monodromy of amplitudes, and  commutes with the action of $C_{Q,M}$.  In these notes, it plays a subordinate role  because we do not know how to control $L_G$ in general. 
   \end{rem} 

\subsection{Face relations} The recursive structure of  graph motives gives rise to relations between 
the periods of different Feynman graphs.

Let $G$ be a \motic Feynman graph. From theorem $\ref{thm: recursive}$, the irreducible components of the divisor $D\subset P^G$ define morphisms `inclusion of facets':
\begin{eqnarray}  
i_{e}\quad : \quad \qquad \qquad \qquad  P^{G/e} \backslash Y_{G/e} & \hookrightarrow&  P^G \backslash Y_G  \\ 
i_{\gamma} \quad: \quad (P^{\gamma}\backslash Y_{\gamma})  \times (P^{G/\gamma} \backslash Y_{G/\gamma}) &  \hookrightarrow &  P^G \backslash Y_G 
\nonumber
\end{eqnarray}
where $e\in E_G$ or $\gamma \subsetneq E_G$ is a \motic subgraph. 
They induce morphisms of relative cohomology as described in \cite{NotesMot}, \S10.3.

\begin{thm}  \label{thmFaceRel} The  maps $i_e, i_{\gamma}$ induce morphisms 
\begin{eqnarray} i_e: \mot_{G/e}  &\To &   \mot_{G}  \label{eqn: facemorphismsforG} \\ 
 i_{\gamma} : \mot_{\gamma} \otimes \mot_{G/\gamma} & \To&  \mot_G \nonumber 
\end{eqnarray} 
in the category $\HH(S)$, where $S$ is a Zariski open subset of $K^{\gen}_{Q,M}$ on which the above objects are defined.  On the Betti realisation, this map 
gives
$$\omega_B (i_e)^{\vee}  \,  [\sigma_G]  = [\sigma_{G/e}] \qquad \hbox{ and } \qquad  \omega_B (i_{\gamma})^{\vee}  \, [\sigma_G]  = [\sigma_{\gamma}] \otimes [\sigma_{G/\gamma}] \ .$$
This implies equalities of motivic periods  in $\HF^{\mm}_{Q,M}$
\begin{eqnarray} \label{eqn: Faceequations}
  [\mot_{G/e}, \omega, \sigma_{G/e}]^{\mm }  & = &   [\mot_{G},  \omega_{dR}^{\gen} (i_e)(\omega), \sigma_{G}]^{\mm }   \\ 
 \qquad  {[}\mot_{\gamma}, \omega_1, \sigma_{\gamma}]^{\mm } 
\times  [\mot_{G/\gamma}, \omega_2, \sigma_{G/\gamma}]^{\mm } &  =  &  [\mot_{G},  \omega_{dR}^{\gen}  (i_{\gamma})(\omega_1\otimes \omega_2), \sigma_{G}]^{\mm }  \ . \nonumber
\end{eqnarray} 
\end{thm}
\begin{proof}
Theorem  $\ref{thm: recursive}$ implies that the facet map
$$ i_e: P^{G/e} \backslash Y_{G/e} \overset{\sim}{\rightarrow}  D_e \backslash (Y_{G}\cap D_e) $$
is an isomorphism of stratified spaces, where  the stratification is induced by the divisor $D_e \cap (D \backslash D_e)$  on both sides of the equation.  The face morphisms  defined in \cite{NotesMot} \S10.4,  therefore  define the required morphism
$\mot_{G/e} \rightarrow \mot_G  $ 
in the category $\HH(S)$. In the Betti realisation, this map corresponds to taking the boundary component of a relative homology cycle which is contained in $D_e(\C)$.
By corollary $\ref{cor: facetstructure}$ this gives exactly $\sigma_G \cap D_e(\C)  \cong \widetilde{\sigma}_G \cap D_e(\C) \cong  \widetilde{\sigma}_{G/e}\cong \sigma_{G/e}$.

The corresponding equation in the case of the face map $i_{\gamma}$ follows from a  similar argument, using  the isomorphism of stratified spaces  $$ i_{\gamma}: ( P^{\gamma} \backslash Y_{\gamma})   \times ( P^{G/\gamma} \backslash Y_{G/\gamma}) \overset{\sim}{\rightarrow}  D_{\gamma} \backslash (Y_{G}\cap D_{\gamma})  \ ,$$
 where the stratification on the left is the product of the stratifications induced by the divisors $(\ref{eqn: Ddivisordefn})$ on $P^{\gamma}$ and $P^{G/\gamma}$,
 which follows from theorem $\ref{thm: recursive}$.
\end{proof}

This theorem implies an analogous statement corresponding to the inclusions of faces of higher codimension. We shall call $(\ref{eqn: Faceequations})$
the face relations.\footnote{Although the face relations are stated here as identities between periods of a category of Betti and de Rham realisations, it is obvious that they should be true `motivically' for any reasonable definition of the word, since they come from the morphisms $i_e,i_{\gamma}$ of algebraic varieties.} 
They are preserved by the action of the cosmic Galois group  $C_{Q,M}$. 

\subsection{Multiplicative structure} 
 It is important to observe that in the second equation of $(\ref{eqn: Faceequations})$, one  of $\gamma$ or $G/\gamma$ has no dependence on 
external kinematics $(\ref{Deltaallornothing})$, and hence defines a constant motivic  period in the sense of \cite{NotesMot} \S7.3.1.
Therefore even if one is only interested in processes with $Q$ external momenta and $M$ non-zero particle masses, one is inexorably led to consider the 
case of Feynman integrals with no kinematics (these   only depend on the  graph polynomial $\Psi_G$ and not $\Xi_G(q,m)$).  The  following proposition gives another example of the special role played by  periods of Feynman graphs
of type $(0,0)$.

\begin{prop}
The vector space  $\HF^{\mm}_{Q,M}$ is closed under multiplication by elements of $\HF^{\mm}_{0,0}$. In other words, multiplication defines  a map
$$\HF^{\mm}_{0,0} \otimes_{\Q} \HF^{\mm}_{Q,M} \To \HF^{\mm}_{Q,M}\ .$$
Thus the space of Feynman periods of type $(0,0)$ is a commutative algebra, and the space
of Feynman periods of type $(Q,M)$ is a module over it.
\end{prop} 
\begin{proof} Let $\gamma$ be a \motic graph of type $(0,0)$, and let  $H$ be a \motic Feynman graph of type $(Q,M)$.  Choose any vertex $v$ of $H$. By inserting $\gamma$ into $v$ and attaching the edges of $H$ (both internal and external) which are incident to $v$ to  vertices of $\gamma$ in any way, 
defines a new graph $G$ of type $(Q,M)$, such that $\gamma \subset G$ is \motic and $G/\gamma \cong H$ is also \motic\!\!. By theorem $\ref{thm: moticproperties}$ $(ii)$, 
$G$ is \motic\!\!.  Now apply the  face equation  in the second line of $(\ref{eqn: Faceequations})$. This proves that a product of  motivic periods of $\gamma$
and $H$ are motivic periods of $G$. 
\end{proof}

The same statement is evidently true also for de Rham periods: multiplication yields a map
$\HF^{\dR}_{0,0} \otimes_{\Q} \HF^{\dR}_{Q,M} \To \HF^{\dR}_{Q,M},$ and $\HF^{\dR}_{0,0}$ is an algebra.

\subsection{Single-valued amplitudes and symbols} \label{sectsvandsymbols}
The right-hand side of the coaction involves de Rham Feynman periods. 
It  is an  important problem to determine properties of the right-hand side  of the coaction and try to interpret these quantities physically. 
Note that they  do not have periods.

As a substitute, we have a notion of \emph{single-valued period} \cite{NotesMot} \S 8.3. Restricting to the space of de Rham Feynman amplitudes, it is a homomorphism  $$\s^{\mm}: \HF^{\dR}_{Q,M} \To \Pe^{\mm}_{Q,M} \otimes_{\Q} \overline{\Pe}^{\mm}_{Q,M} \ .$$
Composing with the period homomorphism defines a real analytic single-valued function  on $K^{\mathrm{gen}}_{Q,M}(\C)$, with possible poles, which we call the \emph{single-valued period}.

In the case when a de Rham Feynman period $\xi$ is \emph{unipotent} \cite{NotesMot} \S9.2, we can associate various notions of symbol to it. 
This holds in particular when $\xi$  is mixed Artin-Tate, and therefore can be applied to large classes of Feynman periods. 
The \emph{symbol} is a class in the reduced bar construction 
$$\mathrm{smb}(\xi) \in H^0 (\mathbb{B}(\Omega^1_{k_{Q,M}}))$$
and has an invariant called its \emph{length}. The cohomological symbol $\mathrm{cmb}_{\ell}(\xi)$  is the length-$\ell$ part of the symbol
of length $\leq \ell$ and can be viewed as  an integrable word in logarithmic one-forms  on $k_{Q,M}$. See \cite{NotesMot} \S9.4.
This notion of symbol is a generalisation of the notion of symbol  used extensively in the physics literature.

\subsection{Total Feynman motive}
Using the  face maps we can take the limit over all graphs of type $(Q,M)$  and assemble them into a single object. 

\begin{defn} The \emph{total (convergent) Feynman motive} of type $(Q,M)$ is 
\begin{equation} \label{TotalMotive} \mot_{Q,M} =\varinjlim_G   \mot_G
\end{equation} 
where the limit is over all \motic graphs of type $(Q,M)$ and the morphisms are given by the face maps $i_e$  of $(\ref{eqn: facemorphismsforG})$. It is an ind object
of $\mathrm{Rep} \, (C_{Q,M})$.  
\end{defn}

The periods of $(\ref{TotalMotive})$ do not \emph{a priori}
contain the renormalised amplitudes of graphs. It would be interesting to write down the object which does capture the  amplitudes of all graphs of type $(Q,M)$
after renormalisation.  Note also that $(\ref{TotalMotive})$ has a number of variants: for instance
one can take the limit over all planar graphs, or indeed any family of graphs closed  under the operation of contracting edges.

\section{Weights and stability} \label{sect: Weightsstability}
We now apply results on weight filtrations from \cite{NotesMot} \S 9.

\subsection{Motives of descendants} We  attach a motive
to the  \motic descendants of graphs (defined in  \S\ref{sect:Moticdescendants}),    in the following way.
To a \motic descendent $\gamma_1 \otimes \ldots \otimes \gamma_n$ of $G$ of type $(Q,M)$,  assign the object 
$$\mot(\gamma_1 \otimes \ldots \otimes \gamma_n) = \mot(\gamma_1) \otimes \ldots \otimes \mot(\gamma_n) \qquad \in \qquad \HH(S) \ , $$
where $S$ is some Zariski-open in the space of kinematics $K^{\gen}_{Q,M}$. Define 
a Betti class in  $(\omega_{B}\mot(\gamma_1 \otimes \ldots \otimes \gamma_n) )^{\vee}$ by 
 $$[\sigma_{\gamma_1 \otimes \ldots \otimes \gamma_n}] = [\sigma_{\gamma_1} ] \times \ldots \times [\sigma_{\gamma_n}]\ .$$
Exactly one of the graphs $\gamma_i$ is of type $(Q,M)$; the others are of type $(0,0)$.
There is a canonical  face map $\mot(\gamma_1) \otimes \ldots \otimes \mot(\gamma_n) \rightarrow \mot(G)$ that sends 
$\sigma_{\gamma_1 \otimes \ldots \otimes \gamma_n}$ to $\sigma_G$, defined by inclusion of the corresponding 
face of $\sigma_G$, or by iterating $(\ref{eqn: facemorphismsforG})$. By the face equations $(\ref{eqn: Faceequations})$  the motivic periods of 
descendants of $G$ (i.e.,  the motivic periods $[\mot(\gamma_1\otimes \ldots \otimes \gamma_n), [\sigma_{\gamma_1 \otimes \ldots \otimes \gamma_n}], [\omega]]^{\mm} $) are  also motivic periods of $G$.

The degree $(\ref{degreedefn})$ coincides with the cohomological degree of the corresponding motive, and the dimension of the corresponding
facet in the Feynman polytope.

\begin{lem}  Let $G$ be a Feynman graph of type $(Q,M)$. Then $W_0 \,\mot(G) = \Q(0)$. 
\end{lem} 
\begin{proof} Let $N_G=|E_G| \geq 2$. Apply the results from \cite{NotesMot}, corollary 10.6.  The irreducible components of $D\subset P^G$, defined in $(\ref{eqn: Ddivisordefn})$,
are in one-to-one correspondence with the facets of the Feynman polytope $\widetilde{\sigma}_G$, whose  boundary is homotopic to an $n-1$-sphere, where $n=N_G-1$.  Its cohomology in degree $n-1$ is one-dimensional.  The case when $N_G\leq 1$ is trivial, since $P^G$ is a just a point.  
\end{proof}

\begin{rem} \label{rem:  mot(triv)isQ(0)} If $\deg(G)=0$ then $\mot(G) = H^0(\Spec\, \Q)= \Q(0)$.  It follows that there are only finitely many motives
attached to the set of all possible tensors $\gamma_1 \otimes \ldots \otimes \gamma_n$ of bounded degree. 
\end{rem} 

\subsection{Weight  relations} 
Let $G$ be a Feynman graph of type $(Q,M)$. Recall  that $\mot(G) = H^{N_G-1} (P^G \backslash Y_G, D \backslash (D \cap Y_G))$.
Let $D^{(k)}$ denote the union of the $k$-dimensional facets of $D$. 
Then there is a morphism 
$$W_k H^{k} ( D^{(k)} \backslash D^{(k)} \cap Y_G) \To W_k \mot(G)$$
which is surjective (proposition 10.4 in \cite{NotesMot}).  The dual map on Betti homology
sends the class of the Feynman polytope $[\sigma_G]$ to the  class of the union $[\sigma^{(k)}_G]$ of its $k$-dimensional facets.
This implies the following theorem. 
\begin{thm}  \label{thmkskeleton} Every motivic period of $\mot(G)$ of weight $\leq k$ is equivalent to a motivic period of the form 
$[H^{k} ( D^{(k)} \backslash D^{(k)} \cap Y_G), [\sigma^{(k)}_G],  [\omega]]^{\mm}$.
\end{thm}

Its period is a $k$-dimensional integral 
$$\int_{\sigma_G^{(k)}} \omega$$
where $[\omega] \in  H_{dR}^{k} ( D^{(k)} \backslash D^{(k)} \cap Y_G)$.
By triangulating the domain of integration into affine regions, as discussed in Appendix 1, and taking limits,  this integral
can be written as a sum of regularised limits of integrals over each facet of $D^{(k)}$. Since each facet is isomorphic to a  product of graph hypersurface
complements of motic descendants of $G$ of total degree $\leq k$, this gives a heuristic justification, via the argument of  \S\ref{limitsandreg},
for  conjecture \ref{conjIntro}.  The missing ingredient is to define a notion of regularisation on the level of motivic periods, which is beyond the scope of the notes. This would also have applications to the theory of renormalisation.

\begin{rem} It is not the case that $P^G \backslash Y_G$
are affine,  since there exist  cohomology classes in $P^G \backslash Y_G$ of degree  greater than its dimension.
Therefore proposition 10.7 in \cite{NotesMot} cannot be applied, and we cannot deduce that   the face maps surject onto $W_k \mot_G$, 
which would have implied that all motivic periods of $\mot_G$ of weight $\leq k$ are in the image of the face maps. 
We do not know, therefore,  whether  the motivic periods of weight $\leq k $ of $\mot(G)$ relative to $\sigma_G$ are generated by 
the motivic periods of weight $\leq k$ of its \motic  descendents of degree $\leq k$, as one might hope.
\end{rem}

\subsection{Stability}
A first application  of this  theory  is to show   that the weight-graded parts of the total motive
$\gr^W_k \mot_{Q,M}$ stabilize. 
\begin{thm} \label{thmstabgrad} Let $G$ be a Feynman graph of type $(Q,M)$. Then $\gr^W_k \mot(G)$ is a sub-quotient of 
$$\bigoplus_{\gamma=\gamma_1 \times \ldots \times \gamma_r} \,  \bigoplus_{i_1 + \ldots + i_r\leq k}  \gr^W_k  H^{i_1}(P^{\gamma_1} \backslash Y_{\gamma_1}) \otimes \ldots \otimes 
H^{i_r} (P^{\gamma_r} \backslash Y_{\gamma_r}) \ $$
 where the  direct sum is over  \motic descendents of $G$ of degree $i_1+ \ldots + i_r \leq k$. 
\end{thm}
\begin{proof} Apply corollary of  10.5 of \cite{NotesMot}  to $\mot(G)$. We  deduce that 
$\gr^W_k \mot(G)$ is a sub-quotient of $\bigoplus_{|I|\geq n-k} \gr^W_k H^{n-|I|} (D_I \backslash D_I \cap Y_G)_{/S}$, where $D_I$ are the codimension $|I|$, and hence dimension $n - |I| \leq k$ facets
of $D$. By theorem \ref{thm: recursive},  
$$D_I \backslash ( D_I \cap Y_G)  \cong P^{\gamma_1} \backslash Y_{\gamma_1}  \times \ldots \times 
 P^{\gamma_r} \backslash Y_{\gamma_r} \ ,$$
 where $\gamma_1 \otimes \ldots  \otimes \gamma_r$ is a descendant of $G$ of degree $n - |I|$. 
The statement follows from the Kunneth formula.  (Note that the degree of $\gamma_{i_j}$ is unrelated to $i_{j}$.) \end{proof} 
The theorem gives a constraint on the Hodge polynomials  of the motivic periods of $G$ of weight $\leq k$.
Combined with the Galois coaction, this gives a constraint on motivic periods to  all orders.
For example,  we can deduce the following corollary. 
\begin{cor} Suppose that all 1PI  graphs of type $(0,0)$ up to $N+1$ edges
have mixed Tate cohomology in all degrees. Then the Galois conjugates of any motivic Feynman amplitude of type $(0,0)$
which is  of weight $\leq N$ is mixed Tate. 
\end{cor}

By computation, one knows that the assumption of the corollary is true for N up to about 10.  This corollary therefore already provides a very strong constraint on the 
possible periods which can occur to all orders in perturbation theory.

We can be more precise and try to bound not only the weight-graded parts, but also the extensions between them. 
This is the spirit of conjecture $\ref{conjIntro}$.  In this direction we can prove the following weaker version of the conjecture. 

\begin{thm} \label{thm: finiteness} The space $W_k \HF^{\mm}_{Q,M}$ is finite-dimensional.
\end{thm} 

\begin{rem} A proof of this theorem is given in \S\ref{sectpffinite} using affine models.  A different  way to find an upper bound for    the vector space of periods 
$\per \, W_k \HF^{\mm}_{Q,M}$ is as follows.  Apply theorem $\ref{thmkskeleton}$ to write any motivic period of weight $\leq k$ as a motivic period of a union  $D^{(k)}$ of $k$-dimensional
facets. Using our canonical affine covering of $\Pro^G$, and  triangulating as in Appendix 1, its  period can be written as a sum of  periods of 
affine pieces of each facet. Since the number of graphs with at most $k+1$ edges is finite, and each facet is a product of such graphs,  there are only finitely many such affine pieces. 
This argument gives a crude but effective upper bound for the periods of weight $\leq k$ in terms of  periods of the relative cohomology of (blow-ups of ) graph hypersurfaces of graphs with $\leq k+1$ edges where we now integrate over a cube $[0,1]^n$ (see comments after corollary \ref{correlperiod}). These can in principle be computed for small $k$.
\end{rem}

\subsection{A principle of small graphs}
Combined with the action of the cosmic Galois group, 
$C_{Q,M} \times \HF^{\mm}_{Q,M} \To \HF^{\mm}_{Q,M}$
or rather, the fact that the motivic  periods of any   graph $G$ are closed under the action of $C_{Q,M}$,
 theorem  $\ref{thm: finiteness}$ gives  constraints on Feynman amplitudes to all orders. 
 The space  $W_k \HF^{\mm}_{Q,M}$ is determined from  the finitely many  `small' graphs with at most $k+1$  edges.  This forces constraints on the  Galois conjugates of Feynman periods \emph{of all graphs}. For an example, see \S\ref{sect: smallgraphsillustration}.

\subsection{Affine motive and proof of finiteness theorem $ \ref{thm: finiteness}$} \label{sectpffinite}
Let $G$ be a Feynman graph of type $(Q,M)$. Define the \emph{affine motive} $\mot^{a}_G$ of $G$ in an identical way to definition \ref{defnmotG} except that 
we replace $P^G$ with its affine open $A^G$ of \S\ref{sectAffineModel}:
$$\mot^{a}_G = H^{N_G-1} ( A^G \backslash Y_G , D  \backslash( D \cap Y_G))_{/S}\ ,$$
as an object of $\HH(S)$, for some  Zariski-open $S$ in the space of kinematics $K^{\gen}_{Q,M}$. 
It follows from the construction of $A^G$, which is obtained by removing from $P^G$ hyperplanes with strictly positive coefficients, that 
$$\widetilde{\sigma}_G \subset A^G (\C)\ .$$
See the proof of theorem \ref{thmsigmaGavoidsY}. 
Thus $\mot^{a}_G$ has a canonical Betti element defined by $\sigma_G$ in an identical manner to definition \ref{defnsigmaGclass}:
$$[\sigma_G ] \in \Gamma \big(U^{\mathrm{gen}}_{Q,M}, (\mot^a_G)_B^{\vee})\big)$$
Furthermore, the inclusion 
$A^G \backslash Y_G \subset P^G \backslash Y_G$ defines a morphism of objects 
$$i: \mot_G \To \mot^{a}_G$$
in $\HH(S)$, which respects the Betti-clases $[\sigma_G]$ on both sides. This morphism defines an equivalence on the level of motivic periods \cite{NotesMot} \S2. 
\begin{cor} There is an equality of motivic periods $$I^{\mm}_G(\omega) = [\mot^a_G,  i_{dR} [\omega], [\sigma_G]]^{\mm}\ ,$$
where the left-hand side was defined in  definition \ref{defnmotFeynper}.
In particular, every  Feynman period of $G$ is a period of the affine motive $\mot^a_G$. We could thus use the affine motives $\mot^a_G$ to study
the weights, representations and so on of Feynman periods.
\end{cor} 
Note that the affine motive is excessively large: it has many periods which are unrelated to Feynman graphs. One advantage of the previous corollary, however,
is that it enables us to express every Feynman period as an integral of a globally-defined algebraic differential form, since, by a theorem due to Grothendieck, 
the de Rham cohomology of an affine variety is the cohomology of the complex of global regular differential forms. The price to pay is that the integrand may involve linear 
denominators of the form $\sum_{e \in \gamma} \alpha_e$, where $\gamma$ is a \motic subgraph of $G$. This remark may or may not be of practical use in computing
Feynman periods. 
\begin{defn}  Let us call the \emph{affine motivic periods} of a graph $G$ to be the space of motivic periods of $\mot^a_G$ with respect to  $\sigma_G$.
By the above remarks, it contains the space of motivic periods of $G$. 
\end{defn}

Now, it follows from the product structure on the spaces $A^G$ that the analogue of theorem $\ref{thm: recursive}$ holds on replacing $P^G$ by $A^G$, and hence the
face relations (theorem $\ref{thmFaceRel}$) hold for $\mot^a_G$.  Now apply proposition 10.7 in \cite{NotesMot}, which exploits Artin vanishing for the cohomology of affine schemes,  to deduce the following theorem.

\begin{thm} \label{thmMainAffine}
The affine motivic periods of $G$ of weight $\leq k$ are $k_{Q,M}$-linear combinations of  the affine motivic periods  of \motic descendents of $G$
of degree $\leq k$.
  \end{thm}
Since there are only finitely many descendents of bounded degree (remark \ref{rem:  mot(triv)isQ(0)}), we immediately deduce theorem $\ref{thm: finiteness}$.

\section{The constant cosmic Galois group} \label{sect: ConstantCosmic}
The motivic periods of graphs  of type $(Q,M) = (0,0)$, which have no dependence on external kinematics, 
play a special role in the Galois theory of all Feynman amplitudes and are of particular number-theoretic interest. 

\subsection{A Galois theory of graph periods}
We recall the main definitions  in this   case. We shall make no restrictions on the `physicality'  of the graphs under consideration, i.e., our graphs can have
arbitrary vertex-degrees for the time being.

\begin{defn} For any graph $G$ of type $(0,0)$, recall that 
$$\mot(G) = H^{N_G-1} (P^G \backslash Y_G , D \backslash (D\cap Y_G))\ ,$$
is an effective object in the category $\HH$ defined in \cite{NotesMot} \S2. It consists of triples $(V_B, V_{dR}, c)$ where $V_B, V_{dR}$ are finite-dimensional $\Q$ vector spaces, and $c$ is an isomorphism $c: V_{dR} \otimes \C \overset{\sim}{\rightarrow} V_B \otimes \C$. These spaces are equipped with filtrations which define a mixed Hodge structure. In this context, the word effective means that the Hodge numbers of $\mot(G)$
satisfy $h_{p,q} = 0$ if $p$ or $q$ is negative. 
  Denote the ring of \emph{motivic graph periods} to be the $\Q$-vector space spanned by the motivic periods:
$$\HF_{0,0}^{\mm} = \langle [\mot(G), [\sigma_G], [\omega]]^{\mm} \rangle_{\Q} \quad \subset \quad \Pe^{\mm, +}_{\HH} $$
 where $G$ is \motic (i.e., 1PI) and $[\sigma_G] \in \mot(G)_{B}^{\vee}$ is the  canonical Betti framing. 
 It actually lands in the subspace $\Pe^{\mm,+} \subset \Pe^{\mm, +}_{\HH}$ defined in \cite{NotesMot}, Definition 3.4.
  Define the \emph{space of  motivic  periods} of a fixed graph $G$ to be the vector space
 \begin{equation}
 \HF^{\mm}(G) = \langle [\mot(G), [ \sigma_G], [\omega]]^{\mm} \quad \hbox{ for } \quad [\omega] \in \mot(G)_{dR}\rangle_{\Q}
 \end{equation}
  spanned by its motivic graph periods. We say that a graph $G$ has \emph{weight  at most $n$} and write $w(G) \leq n$ if 
  $W_n \HF^{\mm}(G)  = \HF^{\mm}(G)$, in accordance with \S\ref{sect: Firstapplications}.\footnote{This is a more subtle notion than the naive weight of a graph defined by the weight of  the object $\mot_G$ in $\HH$.  If $W_n \mot(G)=\mot(G)$ then it  is certainly true that the periods of $G$ have weight  $\leq n$, but the  examples given below show that the converse is false.}
 \end{defn}
 
 In the special case when $G$  is overall log-divergent ($N_G= 2h_G$) and  primitive ($N_{\gamma} > 2 h_{\gamma}$ for all $\gamma \subsetneq E_G$), 
the mixed Hodge structure underlying $\mot(G)$ coincides with the graph motive
 of \cite{BEK}. In this case the  amplitude is given by the integral 
\begin{equation}\label{eqn: IGper}
 I_G = \per ([\mot(G), [\omega_G], [\sigma_G]]^{\mm}) =  \int_{\sigma_G} {\Omega_G \over \Psi^2_G}\ 
 \end{equation}
which converges by lemma $\ref{lem: valuationsconvergence}.$
Thanks to \cite{BK,Census} we know hundreds of examples  of periods $(\ref{eqn: IGper})$, and  many  identities between them.  These identities, when taken alone, do not give much control on the  possible integrals  $(\ref{eqn: IGper})$, since every new
algebra generator is arbitrary. However, if these identities hold on the level of the motivic periods $I_G^{\mm}$, as we expect, then when combined with the action of the cosmic Galois group  and  stability, we obtain a very rigid structure, since the Galois conjugates are constrained by the periods of smaller graphs.

\subsection{Invariants and classification} Motivic periods over $\Q$ are studied in some detail in \cite{NotesMot}, \S3. Two constructions worth mentioning are the 
unipotency degree (or coradical filtration) $C_i \Pe^{\mm,+}_{\HH}$ and the decomposition into primitives. Say that a motivic Feynman period is of \emph{unipotency degree} at most $n$ if it lies in  $C_n \HF^{\mm}_{0,0}$. 
In the case of motivic multiple zeta values, the unipotency degree is bounded above by the depth.
The \emph{decomposition into primitives} is a homomorphism \cite{NotesMot}, \S5
$$\Phi: \gr^C  \HF_{0,0}^{\mm} \To  \gr^C_0 \HF_{0,0}^{\mm} \otimes_{\Q} T^c (\gr^C_1 \Or(U^{dR}_{\HH}))$$
where $T^c$ denotes the tensor coalgebra, or shuffle algebra, and $\gr^C_1 \Or(U^{dR}_{\HH})$ is a vector space which can be made explicit.  The unipotency grading
on the left-hand side coincides with the length grading of tensors on the right. 
This map generalises the (highest-length part of) the  decomposition of motivic multiple zeta values into an alphabet of letters $f_{2n+1}$ to all motivic periods.

 We have constructed a map from graphs to representations
 \begin{eqnarray} \label{graphstoreps} \{\hbox{Graphs of type } (0,0)\}  &\To & \mathrm{Rep}_{\Q} (C_{0,0})   \\
 G & \mapsto &   \HF^{\mm}(G)  \nonumber
 \end{eqnarray} 
 where $C_{0,0}$ is the constant cosmic Galois group. This is more subtle than the naive map which sends $G$ to the object $\mot(G)_{dR}$, since it
 takes into account the Betti framing $\sigma_G$. In the notation of \cite{NotesMot} \S2.4 we have 
 $$ \HF^{\mm}(G) \cong  ({}_{\sigma}\mot(G))_{dR} $$
 where ${}_{\sigma}\mot(G)$ denotes the smallest quotient of $\mot(G)$ in the category $\HH$ such that $\sigma_G \in  ({}_{\sigma}\mot(G))_B^{\vee}$. 
  We can apply any of the invariants of motivic periods defined in  \cite{NotesMot} to
    $\HF^{\mm}(G)$.   The challenge, then,  is to relate  invariants of motivic graph periods to topological invariants of  their graphs, and find relations between 
 graphs through which the map $(\ref{graphstoreps})$ factorizes (see \S \ref{sect: programme}).

For example, we called  $G$  mixed  Artin-Tate if  all elements of $\HF^{\mm}(G)$  have Hodge numbers $h_{p,q}=0$ if $p\neq q$.  Graphs of vertex-width $\leq 3$ are of this type \cite{BrFeyn}.

Furthermore, let us call $G$ \emph{separated}\footnote{Or, better, the weaker condition 
${}_{\sigma}(\mot_G)_{dR} = W_0 \, ({}_{\sigma}(\mot_G))_{dR} \oplus F^1 ({}_{\sigma}(\mot_G))_{dR}$}  if it satisfies
$$ (\mot_G)_{dR} = W_0  (\mot_G)_{dR} \oplus F^1 (\mot_G)_{dR}\ .$$
In this case,  \cite{NotesMot} \S4.3 provides a canonical projection
 $$\pi^{\dR,\mm+}: \HF^{\mm}(G) \To \HF^{\dR}(G)$$ from the motivic periods of $G$ to the de Rham periods of $G$, and  gives a handle on the right-hand terms in the motivic coaction.  There is  evidence to suggest, using methods from \cite{Framings}, that a large class of graphs indeed satisfy this property.  
In particular, all graphs of mixed Artin-Tate type are separated.

For any separated Feynman graph $G$ of type $(0,0)$,  apply the projection followed by the single-valued map $\s^{\mm}$ \cite{NotesMot}, \S4.1 to obtain a linear map
$$ \HF^{\mm}(G) \overset{\pi^{\dR,\mm+}}{ \To}  \HF^{\dR}(G) \overset{\s^{\mm}}{ \To}   \Pe^{\mm}_{\HH}\ . $$ 
This defines canonical  single-valued versions of
its motivic periods $ \s^{\mm} \pi^{\dR,\mm+} I^{\mm}_G(\omega)$.  For multiple zeta values, the corresponding single-valued versions  occur in 
 string perturbation theory,  and  have a physical significance since they relate open and closed  superstring amplitudes \cite{ClosedString, Zerbini, Greenetc}.
  Indeed, the formula for the closed string vertex operator for the emission of a closed string state as a product of open string vertex operators precisely mimics
  the definition of the single-valued motivic periods (\cite{NotesMot}, last line of $\S4.1$.)
  
\subsection{Small-graphs principle}  \label{sect: smallgraphsillustration} By way of illustration,
we compute the motivic
periods of graphs with at most three edges and deduce some non-trivial consequences.
We first dispense with two degenerate families of graphs.

\begin{lem} \label{lemtrivgraphmots} If $G$ has  a single vertex, or a single loop, then 
$\mot(G) \cong \Q(0).$
\end{lem} 
\begin{proof} Let $n=E_G-1$. 
Suppose that $h_G=1$. Then   $\Psi_G = \sum_{e\in E_G} \alpha_e$ and  $X_G$ is a hyperplane $H$. Since $G$ has no non-trivial \motic subgraphs, $P^G= \Pro^n$ and $P^G \backslash Y_G = \Pro^n \backslash H \cong \A^n$. Similarly,  every stratum $D_I\cap (D_I\backslash Y_G)$ is  an affine space and has the cohomology of a point. 
 Therefore the  relative cohomology spectral sequence $E_1^{p,q} =\bigoplus_{|I|=p} H^q(D_I\cap (D_I\backslash Y_G))  $, which converges to $\mot(G)$,  satisfies $E^{p,q}_1= 0$ if $q>0$ and hence $\mot(G) \cong \Q(0)$.
 
 Now suppose that $G$ has one vertex. Every subgraph of $G$ is \motic\!\!, and the graph polynomial is $\Psi_G = \prod_{e \in E_G} \alpha_e$. Since the graph hypersurface $X_G$ is the union  $L$ of coordinate axes, its strict transform $Y_G$ in $P^G$ is empty. Therefore
$$\mot(G) = H^n (P^G, D) \cong H^n(P^G \backslash D)^{\vee}(-n)\ ,$$
by Poincar\'e-Verdier duality. Since $P^G \backslash D \cong \Pro^n \backslash L \cong \GG_m^n$,  the right-hand side is $H^n(\GG_m^n)^{\vee}(-n) = \Q(-1)^{\otimes(-n)}(-n) = \Q(0)$. 
\end{proof}
 
All  two-edge \motic (i.e. 1PI) graphs are covered by the previous lemma. The four 1PI graphs with   three edges are depicted below. Their  graph polynomials  are underneath.
\begin{figure}[h!]
    \epsfxsize=10.0cm \epsfbox{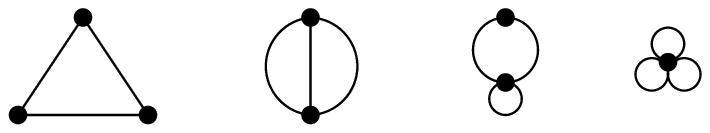}
 $$\quad\alpha_1+\alpha_2+\alpha_3  \qquad \alpha_1\alpha_2\!+\!\alpha_1\alpha_3\!+\!\alpha_2\alpha_3 \qquad  \alpha_1(\alpha_2+\alpha_3) \qquad \alpha_1\alpha_2\alpha_3$$
   \end{figure}
 The two outer graphs are trivial by the previous lemma.  The middle two have non-global periods which are in fact trivial (they satisfy 
 $\mot(G) \cong \Q(0) \oplus \Q(-1)$).  This can be seen in two ways: either by a similar analysis to \S \ref{sect: nonglobalexample}, or by 
 observing that $\mot(G)$  defines a mixed Tate motive which is   unramified at all  primes.  Every Kummer extension of mixed Tate motives which is defined over $\Z$ splits as  a direct sum of Tate motives, so $\mot(G)$ has no non-trivial motivic periods.

\begin{lem}We have 
$W_2 \HF_{0,0}^{\mm}  = W_0 \HF_{0,0}^{\mm}\cong \Q$
\end{lem} 
\begin{proof}
By  theorem   \ref{thmstabgrad} and the  above  calculations, $\gr^W_{i} \HF_{0,0}^{\mm}$ vanishes for $i=1$, and is a direct sum of Tate motives
$\Q(0)$ and $\Q(-1)$ for $i=0, 2$ respectively since it is generated by the cohomology of graphs of degree $\leq 2$. 
 Therefore   $W_2 \HF_{0,0}^{\mm}$ is a  linear combination of Kummer motivic periods (motivic logarithms), which are the motivic periods of extensions of $\Q(-1)$ by $\Q(0)$. By theorem 
 \ref{thmkskeleton}  the motivic periods of $\mot(G)$ of weight $\leq 2$ are equivalent to the motivic periods of its $2$-dimensional skeleton $D^{(2)}$, each of whose
 faces is of one of the above types, or, more trivially, corresponds to a product of two $2$-edge graphs (when $|E_G|=2$, $\Pro^G\backslash Y_G$ is  either  $\Pro^1$  if $G$ has one vertex or  $\A^1$ if it has two, see lemma \ref{lemtrivgraphmots}). 
  From the explicit description of $\Pro^G$ as a blow-up,  one can check that $D^{(2)}$  is  a mixed Tate motive over $\Q$ unramified at all primes by applying the criterion of \cite{Go-Ma}, proposition 4.3. It therefore has no  non-trivial periods. 
  It follows that the only possible periods in weight $\leq 2$ are rational numbers (periods of $\Q(0)$), and rational multiples of the Lefschetz motivic period $\Lef^{\mm}$, whose period is $2\pi i$. The latter is anti-invariant under the action of the real Frobenius involution. On the other hand, motivic periods of graphs are invariant under the real Frobenius, since the Betti class $\sigma_G$ is fixed under its action. This rules out the second case.
    \end{proof}

Since a motivic logarithm $\log^{\mm} (x)$  is determined by its period $\log(x)$,  an alternative approach to proving this lemma would be by direct computation of the periods of the two and three-edge graphs, along the lines of  the method  of  Appendix 2. 

This innocuous-looking statement provides a constraint  to all orders in perturbation theory, \emph{i.e.}, an equation satisfied by all motivic Feynman amplitudes.  

\begin{thm} Let $G$ be a primitive log-divergent graph,   and $I^{\mm}_G$ its motivic amplitude.
Then every Galois conjugate  which is of weight $\leq 2$ is rational.
\end{thm}

In fact, no element of $\HF_{0,0}^{\mm}$ can  have a Galois conjugate of the form $\log^{\mm}(p)$ for $p$ prime.
This theorem is consistent with the coaction conjecture $\ref{conjcoact}$ below.

\begin{example} In \cite{PanzerSchnetz}, Schnetz and Panzer give  examples of amplitudes of graphs  $P_{9,36}=P_{9,75}$ and 
$P_{9,107}=P_{9,111}$   with 9 loops which are Euler sums and  are constrained in a non-trivial way by
this theorem. One  verifies by replacing them with their motivic versions (assuming the period conjecture
for Euler sums) that they never have a Galois conjugate $\log^{\mm}(2)$. 
\end{example}

In a similar way, one easily checks that the motivic period $\Li_2^{\mm}(\zeta_6)$, where $\zeta_6$ is a primitive $6^\mathrm{th}$
root of unity does not occur in $W_4\HF^{\mm}_{0,0}$. This gives a non-trivial constraint on the amplitudes  $P_{7,11}$, $P_{8,33}$ and $P_{9,136}=P_{9,149}$ at seven, eight and nine loops respectively.  See \cite{PanzerSchnetz} for further details.

\subsection{A programme for a Galois theory of graphs}  \label{sect: programme}
There are several  known  families of algebraic relations  between amplitudes of graphs of type $(0,0)$. 
 We expect that many of these relations also hold between motivic periods. This suggests the following list of relations through which the map $(\ref{graphstoreps})$ might 
 factorize:

\begin{enumerate}
\item (Tadpoles).  $G\sim G/e$ where $e$ is a tadpole (self-edge).
\item (Series-Parallel operations). $G\sim G'$ where $G'$ is obtained from $G$ by subdividing an edge  or duplicating
an edge.
\item (Planar duals).  $G\sim G^{\vee}$ where $G$ is  planar, and $G^{\vee}$ is its dual graph.
\item (One and two-vertex joins).  If $G$  is a one or two-vertex join of   graphs $G_1$ and $G_2$ then we expect a relationship between 
$\HF^{\mm}(G)$ and $\HF^{\mm}(G_1) \HF^{\mm}(G_2)$.

\item (Completion and twist identities).  $G\sim G'$ where $G$ and $G'$ are twist-related  or obtained from a four-regular graph by deleting a  vertex \cite{Census}.
\end{enumerate}

Note that $(3)$ could be extended to non-planar graphs if one considers graph motives of matroids. 
If one replaces 
$\mot(G)$ with a cubical version  $\mot_c(G)$  following \cite{BrFeyn},  then $(3)$ holds automatically.
The series and parallel operations are (planar) dual to each other.  The three-edge examples of \S\ref{sect: smallgraphsillustration}   are equivalent to the empty graph  by $(2)$. The relation $(5)$ is  considerably more speculative than the others.

In addition, we can hope for precise information about the weights of graphs. 
We expect that the highest weight-graded quotient of $\HF^{\mm}(G)$ should be related to the  $c_2$-invariant of a graph,\footnote{to set this up rigorously, 
we could enhance the  category of realizations $\HH$ to include an $\ell$-adic component.  The $c_2$-invariant should be obtained from the action of Frobenius on the highest non-trivial  weight-graded piece of $(\mot_G)_{\ell}$.}  which is known to  satisfy several further combinatorial identities.  This suggests, at the very least for graphs satisfying $N_G= 2h_G$,  that
\begin{enumerate}[resume]
\item $w(G) \leq 2N_G -6$,  and $w(G) \leq 2N_G-8$ if $G$ has weight-drop.
\item    $w(G) \leq 2N_G-8$ if $G$ contains a sub-divergence. 
\item  $w(G) \leq n$ if and only if $w(G') \leq n$ whenever $G, G'$ are equivalent under double-triangle reduction.
\end{enumerate}
It is highly  likely that there are  further relations  between motivic periods of graphs  which have   three-valent vertices or  triangles, which remain to be discovered.\footnote{we know, for example, that graphs of vertex width $\leq 3$ evaluate to multiple zeta values but the numbers of such graphs greatly exceeds the dimension of the space  of multiple zeta values of the appropriate weight, so there must exist many relations between these amplitudes.}

Finally, one would like to have some control on the degree of unipotency of motivic periods. This   has not been investigated.

\subsection{Coaction conjecture and  speculation} Recall that a graph $G$ of type $(0,0)$ is said to be in $\phi^n$ if every vertex  of $G$ has degree at most $n$. 
\begin{defn} Let 
$\HF^{\mm}_{\phi^4} \subset  \HF^{\mm}_{0,0}$ denote the $\Q$-vector space spanned by the motivic 
amplitudes  $(\ref{eqn: IGper})$ of primitive log-divergent graphs in $\phi^4$ theory. \end{defn}
The following extraordinary conjecture  was formulated in \cite{PanzerSchnetz} and called the coaction conjecture. It was a  principal motivation for the present paper. 
\begin{conj} \cite{PanzerSchnetz} \label{conjcoact} $\HF^{\mm}_{\phi^4}$ is stable under the action of  $C_{0,0}$. 
\end{conj} 

This conjecture goes far beyond what we can presently prove, but has been verified numerically in hundreds of examples \cite{PanzerSchnetz}.  However, proving any or all of the  properties  of \S\ref{sect: programme} would lead to    spectacular consequences for graph amplitudes in the direction of this conjecture.
For example, let us assume only $(2)$, and call  two graphs \emph{sp-equivalent} if they are obtained 
from each other by series-parallel operations. 
The smallest graph not equivalent to the trivial graph is the wheel with three spokes $W_3$ with six edges (whose amplitude is $6 \zeta(3)$; we expect that $\zetam(3)$ is 
its unique non-trivial motivic period)
followed by the wheel with four spokes $W_4$ with eight edges.  By stability, 
we would  deduce that 
$$W_7 \HF^{\mm}_{0,0} = \Q \oplus \zetam(3) \Q\ .$$
In particular, no $\zetam(2)$ occurs, which would then imply that no motivic period \emph{at any loop order} can have a Galois conjugate
$\zetam(2)$. Since $W_3$ is primitive log-divergent in  $\phi^4$ this would  also prove conjecture $\ref{conjcoact}$ up to weight seven:
$$W_n \HF^{\mm}_{\phi^4} = W_n \HF^{\mm}_{(0,0)} \quad \hbox{ for } n\leq 7\ , $$
since the right-hand side is stable under the action of $C_{0,0}$ by definition. 
Things become interesting at nine edges. There are three non-trivial sp-equivalence classes  of graphs:
the graph obtained by deleting  an edge from the complete graph $K_5$;  the complete bipartite graph $K_{3,3}$; and a planar graph which is the skeleton of 
triangular prism.  By continuing in this manner, and computing motivic periods of ever-larger graphs, one could deduce infinite families of constraints on motivic periods
of Feynman graphs to all orders. This would go a long  way to  explain the remarkable structure observed in  \cite{PanzerSchnetz}. 

\begin{example} 
 To illustrate how the topology of graphs can impinge upon their periods, consider any class of graphs $\mathfrak{C}$ which is stable under 
 edge contraction and \motic subgraphs. It defines a subspace
 $\HF_{\mathfrak{C}}^{\mm}  \subset  \HF^{\mm}_{(0,0)}$
 which is stable under the action of the constant cosmic Galois group $C_{0,0}$. For example, the motivic periods of planar graphs are Galois-stable.

 Specialising further still, consider
the following thought-experiment for the   wheel with $n$-spokes graphs $W_n$.  Since there are several missing elements we shall be very brief.
The key topological property of the wheel graphs
is that, \emph{modulo sp-equivalence},   the vector space they span  is stable under  \motic  descendants.  Another way to say this is  contracting an edge in a wheel leads to a graph which is either sp-equivalent to another wheel, or the trivial graph. Every motic
subgraph of a wheel has the same property.   

  Bloch-Esnault-Kreimer have proved that  (e.g., as an object of $\HH$), 
\begin{equation} \label{wheelcohomclass} H^{2n-1}(\Pro^{2n-1} \backslash X_{W_n})\cong \Q( 3-2n )\ .
\end{equation} 
Let us assume $(2)$ of \S\ref{sect: programme} and furthermore that 
the wheel motives are mixed Tate over $\Z$ and have no non-trivial non-global motivic periods (i.e., the conclusion, but not necessarily the hypotheses of proposition 10.7 in \cite{NotesMot} hold).  This would imply that  all elements in $({}_{\sigma}\mot(G))_{dR}$ 
are images of classes $(\ref{wheelcohomclass})$ via face maps and hence have weights $\equiv 2 \pmod 4$.
By the following proposition, 
$$\HF^{\mm}(W_n) = \Q \oplus \Q \zetam(3) \oplus \Q\zetam(5) \oplus \ldots \oplus \Q \zetam(2n-3)\ . $$
In particular, this would imply that all periods of the wheels graphs are linear combinations of odd zeta values only. It is known
that     the wheel amplitude $(\ref{eqn: IGper})$ itself is an explicit rational multiple of a single odd zeta value of highest weight. 
In this case, therefore, we expect to see  a direct relationship  emerging between the topological properties of  wheel graphs  in the \motic Hopf algebra
and the Galois-theoretic properties of their periods. 
\end{example} 

\begin{prop} Let $M$ be an effective mixed Tate motive over $\Z$ such that
$$\gr^W_{4n} M = 0 \qquad \hbox{ for all  } n \geq 1 \ .$$
Then the real (i.e., Frobenius-invariant) motivic periods of $M$ are $\Q$-linear combinations of $1$ and $\zetam(2n+1)$, for $n\geq 1$. 
\end{prop} 
\begin{proof}   Choose generators  $\sigma_{2n+1}$  of the de Rham graded Lie algebra of $\MT(\Z)$ in odd degrees $-2n-1$ for $n\geq 1$, where the degree is the `MZV-weight', or  one half of the Hodge-theoretic weight.   The non-rational motivic periods of $M$ have only odd degrees by assumption.   Consider a motivic period  $\xi_{2m+1}$ of $M$ of degree $2m+1>0$. Since $\sigma_{2n+1} \xi_{2m+1} $ has even degree $2(n-m)$, it is zero unless $n=m$ and $\sigma_{2n+1} \xi_{2m+1} \in \Q$. Therefore $\xi_{2n+1}$ has de Rham  Galois conjugates itself and $1$. It is therefore   primitive (unipotency degree $\leq 1$).  
By theorem 3.3 of \cite{BrMTZ} it is in the space  $\zetam(2n+1)\Q\oplus (\Lef^{\mm})^{2n+1}\Q$. Since it has real periods and $\per(\Lef^{\mm})= 2i 
\pi$ is imaginary,  it is a rational multiple of $\zetam(2n+1)$.   
\end{proof}

\begin{rem}
The previous argument fails   for a weight-drop graph such as the bipartite graph  $K_{3,4}$, since its motive is non-trivial in  weight $16 \equiv 0 \pmod 4$ and the previous proposition does not apply.  It has \motic graph sub-quotients $W_3$ and $W_4$, 
so the same argument would allow the motivic amplitude of $K_{3,4}$ to have Galois conjugates $\zetam(3)$ and $\zetam(5)$. 
This is entirely   consistent with the fact that 
$$I_{K_{3,4}}= - {216\over 5} \zeta(5,3)  -81 \zeta(5)\zeta(3) +{522\over 5} \zeta(8)$$
which, assuming the period conjecture for multiple zeta values, indeed has non-trivial Galois conjugates $\zeta(3)$ and $\zeta(5)$. 
\end{rem}
In conclusion, the Galois theory of graphs described here,  with a few extra speculative  ingredients such as those outlined in  \S\ref{sect: programme}, seems to predict quite accurately 
the observed patterns of periods in amplitude computations at low loops. 

\section{Examples with general kinematics and conjectures} \label{sect: Finalsection}

\subsection{General kinematics} It is  possible to undertake a classification of the motivic periods of graphs with few edges
and arbitrary external kinematics as  above. By the small graphs principle  and Galois action, it  leads, in the same way as  \S\ref{sect: smallgraphsillustration}, to all order constraints on Feynman periods of any type $(Q,M)$.  We shall not discuss this here, but only make a few brief comments.

First of all, the cohomology of generic one-loop graphs was studied in \cite{BlochKreimer} and can be re-expressed in the language of motivic periods. Applying
 formula $(\ref{eqn: CoactionImG})$ it provides a computation of the motivic coaction. The case of graphs with subdivergences can be treated using the techniques
described here, and one example is treated in full detail  in an appendix \S\ref{sect: nonglobalexample}. 
The recent preprints \cite{CutsCoproducts} and \cite{CutsCoproducts2} give conjectural formulae for the coaction on some examples of graphs (with the caveat that 
one side of the coaction needs to be expressed in terms of  de Rham periods).  I expect that these formulae can be proved from first principles
using the cohomological techniques described here. A very interesting observation of  \emph{loc. cit.} is that the coaction formulae  apparently continue to  hold on the level of $\varepsilon$-expansions
in dimensional regularisation.

Note also that a  complete analysis of Feynman graphs with up to three edges would include the sunrise graph, which involves  the cohomology of a family of elliptic curves  and has a very extensive literature. The results of \cite{Bloch-Vanhove-Kerr}  likewise  can be used to deduce information about the corresponding motivic periods .

\begin{rem} Other interesting classes of graphs to study in this framework are those of type $(Q,M) = (2,0)$ or $(0,1)$ which depend on a single scale. When the scale
factorizes out of the graph polynomials, the corresponding  amplitudes
 effectively depend upon a single number. The massive banana graphs (see \cite{BroadhurstModular}), for example, would seem to  generate a small family of motivic periods which are stable under the cosmic Galois group, and hence should  have interesting arithmetic properties. 
 We suspect that this could explain why certain combinations of periods related to banana graphs are periods of pure motives, and hence, by Deligne's conjecture, are critical values of the underlying $L$-functions. 
  \end{rem}

\subsection{Graphs with many external legs}
The parametric representation is inefficient for graphs with many edges, and the number of edges does not accurately predict the  expected weight of the amplitude.
The existence of the momentum space representation suggests the weight depends on the loop number. 

\begin{conj} \label{conj: weights} Let $G$ be a Feynman graph with $h$ loops, in $d\in 2 \N$ space-time dimensions,  and let $\omega_G$ be the integrand of the Feynman amplitude. Then 
\begin{equation} \label{eqn: conjweight}
\omega_G \in \omega^{\gen}_{dR} W_{dh} \mot_G \ .
\end{equation}
\end{conj}

\begin{rem} This conjecture only gives a bound on the weights of amplitudes, and not general Feynman periods, which could potentially have higher weights.
\end{rem} 
The heuristic rationale behind the conjecture is as follows: 
\begin{itemize}
\item There should exist  objets  $\mot^{\mom}_G$ in a category of realisations such that the 
amplitude in momentum space is the period of  $[\mot^{\mom}_G, \omega_G^{\mom}, \sigma^{\mom}_G]^{\mm} $.
\item The Schwinger trick (universal quadric) should give an equivalence 
$$[\mot^{\mom}_G, \omega_G^{\mom}, \sigma^{\mom}_G]^{\mm} =  [\mot_G, \omega_G, \sigma_G]^{\mm}$$
This was partly carried out in \cite{BEK}, equation (10.4), in the  case of no  kinematics or  subdivergences,
and for the absolute (not  relative) cohomology. 
\item   The momentum space integrand should satisfy
$\omega_G^{\mom} \in W_{dh} (\mot^{\mom}_G)_{dR}$.
\end{itemize}

This conjecture, combined with the stability conjecture $\ref{conjIntro}$, would   yield  powerful identities for Feynman amplitudes.  In particular, it suggest the following. 

\begin{conj}  \label{corofconj} Let $G$ be as in conjecture \ref{conj: weights}. The amplitude of $G$ is a (regularised) period of the  \motic descendants
of $G$ of degree $\leq dh$. 
\end{conj}
Let $d=4$. Since, for every $h$ there are only finitely many graph topologies  of degree $\leq dh$, this would give a finite set of `master integrals'
for graphs with arbitrarily many external legs, at every loop order.

\begin{ex} Conjecture $\ref{conj: weights}$ is  certainly true for one-loop graphs.  Let $G$ be such a graph. Then its Feynman amplitude is of weight $\leq 4$.
It is expressible as a Feynman period of quotients of $G$ with at most five edges (in this case one can do slightly better 
and replace `five' with `three'). 
This is   a theorem due to Nickel \cite{Nickel}, reproved in \cite{BlochKreimer}, and can be made effective. The analogue of  this theorem for graphs with two loops is not  presently known, it seems, and the programme outlined above suggests a generalisation to all higher loop orders.
\end{ex}

\subsection{Further directions}
Some directions for further research include:
\begin{enumerate}
\item One would like to incorporate ultraviolet divergent graphs and the theory of renormalisation along the lines of \cite{Angles}.  Since
the geometry of the Feynman polytope is very close to the BPHZ forest formula,  the theory of renormalisation fits very naturally in the present framework.
One approach, which is closest to that used in physics, would be to allow integrals with logarithmic terms in the numerators. This can be done by defining  a notion of motivic periods with coefficients  (one needs to interpret an integral of a family of motivic periods\footnote{For example, one can make sense of a formula of the form
$\zetam(2) = -\int_0^1 \log^{\mm}(1-x) {dx \over x}$}
 as a motivic period).\footnote{This also seems to   be a possible way
to study dimensional regularisation: the coefficients of a  Taylor expansion in $\varepsilon$ are  integrals with logarithmic numerators.}
Another approach, which is perhaps less satisfactory,  is to differentiate
with respect to a scale in order to turn all integrands into algebraic differential forms, as in \cite{Angles}. Indeed, one can define the graph motives of UV divergent graphs simply by  using the decomposition into angles and
scales of \cite{Angles} and taking  the renormalised graph motive defined there. 

\item It would be interesting for applications to understand situations with infra-red singularities when the genericity assumptions $(\ref{eqn: genericmomenta})$ are not satisfied.
The graph factorisation theorems partially break down in this case, but by enlarging the class of polynomials considered, one might still retrieve a Galois theory of graphs. 
The QED contributions to the anomalous magnetic moment $g-2$  are a fascinating case study.
\item One would like to rethink the problems of resummation of the perturbative expansion in the context of motivic periods.  Taking a sum of amplitudes viewed merely as complex numbers ignores the fact that they are periods and all the structure that that entails. We expect the perturbative expansion can be lifted canonically to a
series whose coefficients are motivic periods. The invariants of motivic periods defined in \cite{NotesMot}
should enable one to sum this perturbative series in a more organised manner, e.g., according to various types,  which may lead to better convergence properties.
\item In our theory, the domain of integration is trivial and all the content of the physical theory is in the integrand. For this reason the de Rham Galois group plays
a privileged role. 
 The remarks in this chapter also suggest that the graph motive in parametric space is not optimal (at least for graphs with many edges), and one must also consider
momentum space or other integral representations, which will give a different bound on the space of Galois conjugates of amplitudes. It seems to be an important fact that amplitudes have several quite different integral representations, each giving different constraints on their Galois theory.
\item We worked exclusively in Euclidean space. In order to analytically continue to Minkowski space, one would like to know where the singularities
of the Feynman integrals are. A worrying possibility is  that the set of singularities of graphs of a fixed type $(Q,M)$ could become dense in the space $K_{Q,M}$
as the loop number increases. This is why our de Rham fiber functor is at the generic point. It would be  interesting to know if there is an open region in kinematic space
where all Feynman amplitudes are non-singular.

\item  There is good evidence to suggest that superstring amplitudes have a Galois theory of their own  (at least at tree-level \cite{SS}). This seems entirely reasonable
given that the moduli spaces $\mathfrak{M}_{g,n}$ have the same product-structure on their stratification as the one exploited here for amplitudes.
\end{enumerate}

\section{Appendix I: some cohomological tools for  periods} \label{sect: App1}

Not every cohomology class in the de Rham realisation of the graph motive can be represented by a global differential form such as $(\ref{eqn:  ExFeynmInt})$.
A study of non-global periods is  necessary
for understanding the  conjugates of amplitudes under the cosmic Galois group.
Therefore in this section we provide some tools for studying such non-global cohomology classes and their periods. The first is a complex which we use to show that the periods  of graphs are  limits of divergent integrals of globally-defined forms.   The second is a spectral sequence which allows us to import known results about the cohomology of graph hypersurfaces in projective space to the study of graph motives. It is related to the \motic Hopf algebra.

\subsection{A relative algebraic \v{C}ech-de Rham complex} \label{sectionRelAlgCech}
Let $D \subset X$ be a simple normal crossing divisor in a smooth scheme $X$ over $\Q$. 
Let $U_i \subset X$, for $i\in I$,  be a covering of $X$ by a finite collection of smooth affine varieties defined over $\Q$.  Let $D_j$, for $j\in J$, denote
the irreducible components of $D$. Write as usual $U_P= \cap_{i\in P} U_i$ for $\emptyset \subsetneq P\subset I$  and    $D_Q= \cap_{i\in Q} D_i$ for $Q\subset J$
with the convention $D_{\emptyset}=X$. 

 Consider the triple complex
\begin{equation} \label{eqn: reltripcomplex}
\Omega^{n,p,q}(\{U_i\},D)=  \bigoplus_{|P|= p, |Q|=q} \Omega^n(U_P \cap D_Q)
\end{equation} 
where $\Omega^n(U_P\cap D_Q)$ denotes the global sections of the sheaf of Kahler differential forms over $\Q$.  The differentials
$\Omega^{n,p,q}(\{U_i\}, D)\rightarrow \Omega^{n+1,p,q}(\{U_i\}, D)$ are given by the usual differential $d$ in the de Rham complex. The  differentials
$\Omega^{n,p,q}(\{U_i\}, D)\rightarrow \Omega^{n,p+1,q}(\{U_i\}, D)$ are given by the differentials in the usual \v{C}ech complex, 
and  the differentials $\Omega^{n,p,q}(\{U_i\}, D)\rightarrow \Omega^{n,p,q+1}(\{U_i\}, D)$ are given by restriction of differential forms to closed subsets $D_{Q \cup \{q\}} \subset D_Q$ with the  standard sign
convention.
The relative algebraic de Rham cohomology 
$$H_{dR}^n (X, D) = H^n(\mathrm{Tot} ( \Omega^{n,p,q}(\{U_i\},D)) $$
is the cohomology of the total complex associated to the triple complex  $(\ref{eqn: reltripcomplex}).$
A cohomology class  of degree $n$ in the latter can be represented by a collection 
\begin{equation}  \label{eqn:  omegaABforms}
 \omega^P_Q \quad \in \quad \Omega^{n+1-p-q} (U_P \cap D_Q) 
\end{equation}
where $P\subset I$ and $Q\subset J$ with    $|P|=p$,  $|Q|=q$ that are mapped to zero by the total differential. 
Associated to the triple complex $(\ref{eqn: reltripcomplex})$ are a number of spectral sequences, for example
$$E_{1}^{p,q} =  \bigoplus_{ |P|=p+1}  H_{dR}^q (U_P, U_P \cap D)  \quad \Longrightarrow \quad  H_{dR}^{p+q}  (X,D)$$
where the differential is  induced by inclusions as for  the \v{C}ech complex.

\subsection{A relative Stokes' theorem}  Let $C_n(U_P\cap D_Q)$ denote the complex of singular $n$-chains with coefficients in $\Q$  on the topological 
space $U_P\cap D_Q(\C)$. By analogy with $(\ref{eqn: reltripcomplex})$, define a  triple complex of singular chains
\begin{equation} \label{eqn: chainreltripcomplex}
C_{n,p,q}(\{U_i\},D)=  \bigoplus_{|P|= p, |Q|=q} C_n(U_P \cap D_Q)
\end{equation} 
where the differentials are given by the boundary map on chains, and the inclusion maps, with the appropriate signs.
The homology of the total complex is the relative Betti homology
$$H_B^n(X,D)^{\vee} = H_n(X(\C), D(\C)) = H_n (\mathrm{Tot} (  C_{n,p,q}(\{U_i\},D)))\ .$$
A relative homology class can be represented by a collection of chains
\begin{equation} \label{eqn: sigmacollection}
 \sigma^P_Q\quad  \in \quad C_{n+1-p-q}(U_P \cap D_Q)
\end{equation}
 whose total differential is zero.  Denote such a collection by $\sigma=\{\sigma^P_Q\}$.

Given such a chain of degree $n$, and a cohomology class $\omega \in H_{dR}^n(X,Z)$ represented by a collection 
$(\ref{eqn:  omegaABforms})$, define the period  (or integration pairing) by
\begin{equation}  \label{eqn: intpairingrelative}
 \int_{\sigma} \omega := \quad \sum_{P,Q} \int_{\sigma^P_Q} \omega^P_Q \ .
 \end{equation} 
 Note that there can be signs in this formula depending on sign conventions for the differentials in the complexes defined earlier. These are not important for the general discussion which follows.
The  following theorem is a corollary of Grothendieck's theorem \cite{Groth}. I was unable to find a suitable reference in the literature.\footnote{Although,
whilst writing up these notes, Huber and M\"uller-Stach kindly sent me  a preliminary draft of their book project on periods, which contains similar considerations.} 
\begin{thm} The pairing $(\ref{eqn: intpairingrelative})$ is well-defined and computes the isomorphism
$$\mathrm{comp}_{B,dR} : H_{dR}^n(X,D) \otimes_{\Q} \C \overset{\sim}{\To} H_B^n(X,D)^{\vee}\otimes_{\Q} \C\ .$$
\end{thm} 
\begin{proof} (Sketch). The pairing is well-defined  by Stokes' formula, along with the definition of the 
differentials in the complexes $(\ref{eqn: reltripcomplex})$ and $(\ref{eqn: chainreltripcomplex})$. 
By some standard homological algebra, the result follows  from Grothendieck's algebraic de Rham theorem for affine varieties, which 
implies that integration defines a natural isomorphism $H^i_{dR}(U_P \cap D_Q) \otimes_{\Q} \C \overset{\sim}{\To} H^i_{B}(U_P \cap D_Q)\otimes_{\Q} \C$ for all $P,Q$. \end{proof}

\subsection{Sectors and  blow-ups in projective space} \label{sect: sectorsblowups}
We can apply the above to the following situation. With the notation of \S\ref{sect: LinearBlowups}, let $S$ be a finite set and  $B\subset 2^S$ be a set of subsets
of $S$ closed under unions.  Let $P^B$ denote the corresponding blow-up  of $\Pro^S$, and $D^B\subset P^B$ the normal crossing divisor defined in $(\ref{eqn: Ddivisordefn})$.  Let $Y\subset P^B$ be a closed subvariety with the property that
$Y \cap \widetilde{\sigma}_B = \emptyset$. We set
$$X = P^B \backslash Y \qquad \hbox{ and }  \quad D = D^B \backslash (D^B \cap Y)\ . $$
The spaces $P^B$ come with a natural affine covering $\{U_{\FF,c}\} = \{\A^{\FF,c} \backslash (\A^{\FF,c} \cap Y) \}$ where the $\A^{\FF,c}$ are isomorphic to  affine spaces $\A^n$. 
 We have in mind, of course,  the case where $B$ is the set of \motic subgraphs of a 
Feynman graph, and $Y$ the strict transform of graph hypersurfaces.

 The polytope $\widetilde{\sigma}_B$ defined in \S\ref{sect: Bpolytope} can be decomposed into  regions  in the following way.
Choose any point which lies in the interior of $\sigma$:
\begin{equation} \label{eqn: zinteriorpoint}
z = (z_1: \ldots: z_n) \in \Pro^S(\R) \quad \hbox{ where } z_i >0  \hbox{ for all } i\ .
\end{equation}
It  defines a point  on every open  $U_{\FF,c}\subset \A^{\FF,c}$ in our covering. Let $\beta_1,\ldots, \beta_n$ be  the coordinates on $\A^{\FF,c}$ defined by  $(\ref{newpistarequations})$. Then
the inverse image of $z$ is given by equations $\beta_i = z_{a_i}/z_{b_i}$ for some indices $a_i, b_i$. 
The equations $\beta_i=0$ and $\beta_i=z_{a_i}/z_{b_i}$ define a hypercube $H_z$ in $U_{\FF,c}$. Let $\sigma^A_B(z)$ be a  face of this hypercube:
$$\sigma^A_B(z): = \{(\beta_1,\ldots, \beta_n) \in H_z: \beta_i = 0 \hbox{ for }  i \in A, \beta_i = z_{a_i}/z_{b_i} \hbox{ for } i \in B
\} \ ,$$
where $A,B\subset \{1,\ldots, n\}$ are disjoint.  
One can verify that $\sigma$ is tessellated by a  set of  $\sigma^P_Q(z)$, over the different charts $\FF,c$, for certain sets $P,Q$, and that 
$\{ \sigma^P_Q(z)\}$ defines a relative Betti homology class representing $\sigma$.

\begin{figure}[h!]
    \epsfxsize=12.0cm \epsfbox{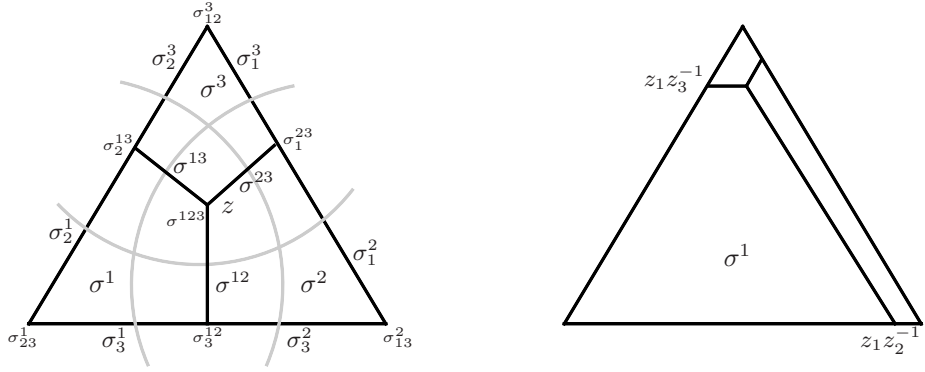}
  {\put(-320,30){{ $\sigma^1$}}}   {\put(-335,50){\small { $\sigma_2^1$}}} {\put(-315,9){\small { $\sigma_3^1$}}} {\put(-350,11){\tiny { $\sigma_{23}^1$}}}
   {\put(-240,30){{ $\sigma^2$}}}   {\put(-220,42){\small { $\sigma_1^2$}}}{\put(-245,9){\small { $\sigma_3^2$}}}{\put(-208,11){\tiny{ $\sigma_{13}^2$}}}
    {\put(-290,55){\tiny{ $\sigma^{123}$}}}    
      {\put(-270,60){{ $z$}}}      
         {\put(-280,11){\tiny { $\sigma_3^{12}$}}}
     {\put(-288,75){\small{ $\sigma^{13}$}}}   {\put(-314,84){\tiny{ $\sigma_2^{13}$}}}
      {\put(-263,68){\small{ $\sigma^{23}$}}}  {\put(-246,86){\tiny{ $\sigma_1^{23}$}}}
    {\put(-272,30){\small{ $\sigma^{12}$}}}
 {\put(-278,104){{ $\sigma^3$}}} {\put(-296,116){\small{ $\sigma_2^3$}}} {\put(-264,116){\small{ $\sigma_1^3$}}} {\put(-280,134){\tiny{ $\sigma_{12}^3$}}}
  {\put(-110,108){\small{ $z_{1}z_3^{-1}$}}}
   {\put(-28,9){\small{ $z_{1}z_2^{-1}$}}}
 {\put(-80,40){{ $\sigma^1$}}}  
 \caption{A decomposition of the  coordinate simplex   in $\Pro^3$ (here $B=\emptyset$),  defined by hyperplanes  $z_i \alpha_i= z_{j} \alpha_j$ for $0<z_{i}<\infty$ and 
 $i=1,2,3$.
 The affine open sets $U_i: \alpha_i \neq 0$ for $i=1,2,3$ are depicted  schematically by  grey arcs. On the right,   $z_2,z_3 \rightarrow 0$. }
 \end{figure}

\begin{cor} \label{correlperiod}
Any relative period of $P^B\backslash Y$ over the domain $\sigma$ is given by 
\begin{equation} \label{eqn: intsigmapq(z)}
\int_{\sigma} \omega = \sum_{P,Q} \int_{\sigma_Q^P(z)} \omega_Q^P\ ,
\end{equation} 
where the $\omega^P_Q$ are differential forms $(\ref{eqn:  omegaABforms})$ on an affine $U^P \cap D_Q$.
\end{cor} 
This corollary gives a means, albeit an inefficient one, to compute  non-global Feynman periods.  It is adapted to the method
of parametric integration \cite{PanzerPhd}.
 For instance, letting all $z_i=1$, 
each term in the  sum in local coordinates is an integral of an algebraic differential form over a cube $[0,1]^m$ for some $m\leq n$. 
\begin{rem}
This procedure is not to be confused with the notion of \emph{sector decomposition} in the physics literature. In that setting, one integrates the (pull-back of) the same globally-defined
form $\omega$  over each sector $\sigma^P$ and sums the contributions. In the above, we are integrating \emph{different forms} $\omega^P_Q$ over each sector. 
\end{rem}

\subsection{Limits and regularisation} \label{limitsandreg}
The integral $(\ref{eqn: intsigmapq(z)})$ does not depend on the point $z$. Since
each open affine $U_i(\C)$ contains the preimage  of $z$  for any point $z$ $(\ref{eqn: zinteriorpoint})$ in the interior of the coordinate
simplex $\sigma$, we can take limits in $(\ref{eqn: intsigmapq(z)})$:
$$
\int_{\sigma} \omega =  \lim_{z\rightarrow \infty} \sum_{P,Q} \int_{\sigma_Q^P(z)} \omega_Q^P\ $$ 
as $z$ tends to any point on the boundary of $\sigma$.   Many of the terms in the sum on the right-hand side
will tend to zero and can be dropped. By repeatedly taking limits, one obtains an expression for the integral on the left-hand side
as limits of possibly divergent  integrals over facets of $\sigma_B$.\footnote{It would be interesting, by studying the asymptotic
behaviour of these integrals as $z$ tends to the boundary, to define a consistent  notion of regularisation of divergent
integrals over faces of the polytope $\sigma_B$ which commutes with these limits.  In this case, we could write
the period of any non-global form as a linear combination of regularised integrals of global forms over faces of $\sigma_B$.
There is   no shortage of 
regularisation techniques for   Feynman integrals in the physics literature.
For instance,   if $P^B= P^G$ is obtained from  a Feynman graph, and $Y_G$ the graph hypersurface, the boundary strata of $P^G$ are related to the \motic Hopf algebra, and  closely resembles the combinatorics of the  BPHZ forest formula (see \cite{Angles}).
 This suggests a possible way to renormalise divergent integrals using the subtraction of counter-terms. }

\subsection{The exceptional locus  spectral sequence} In this section, let $X,D$ denote fibres of $X,D$  as  defined in \S\ref{sectionRelAlgCech}. For reasons which will become apparent
in a moment, let us write $E=D$.  
The divisor $E$ defines a stratification on $X$ by closed subvarieties. Write
$$E_J^o = E_J \backslash (E_J \cap \bigcup_{j\notin J} E_j)\ .$$
For instance, $E_{\emptyset}^o = X \backslash E$.
There is a `Gysin' or residue  spectral sequence
\begin{equation} \label{eqn:  ExcepSS} 
E^{p,q}_1 =  \bigoplus_{|J|=p} H^{q-p} (E^o_J)(-p)  \quad \Longrightarrow\quad  H^{p+q}(X)
\end{equation} 
where the differentials $d_1$ are given by  residues along the irreducible components of $E$ and $p,q\geq 0$.
From now on, let $X=P^G\backslash Y_G$ where $G$ is a Feynman graph, and let $E=D$ be defined by $(\ref{eqn: Ddivisordefn})$.
If  $G$ has no masses or momenta,    we obtain the spectral sequence considered by Bloch in \cite{Bloch}. 
The open strata  $E^o_J$ are complements of graph hypersurfaces in $\GG_m^n$. This spectral sequence is hard to 
control since  the cohomology of the latter is large and there are many cancellations.

Returning now to the case of a general Feynman graph $G$, and 
$X= P^G\backslash Y_G$, we see that 
it is more economical to  take $E \subset D$ to be the  exceptional divisor:
$$E = \cup_{\gamma \subset G} D_{\gamma}\  ,$$
where the union is only over the set of \motic subgraphs of $G$. We shall call the corresponding spectral sequence
$(\ref{eqn:  ExcepSS})$ the \emph{exceptional locus spectral sequence}.

\begin{thm} In this situation, the non-empty strata $E_J$ are indexed by strictly increasing sequences of \motic subgraphs of $G$:
\begin{equation} \label{eqn: Jasnested}
J: \quad \gamma_1 \subsetneq \gamma_2 \subsetneq \cdots \subsetneq \gamma_r \ .
\end{equation}
If we write $\gamma'_i = \gamma_i/\gamma_{i-1}$ for the successive quotients,  where $\gamma_0$  denotes the empty graph, then there is a canonical isomorphism
\begin{equation} \label{eqn: EoJproduct} 
E^o_J \, \cong  \,  \big( \Pro^{N_{\gamma'_1} -1} \backslash X_{\gamma'_1} \big) \times \cdots \times    \big( \Pro^{N_{\gamma'_r} -1} \backslash X_{\gamma'_r} \big) \ .
\end{equation} 
When $J=\emptyset$ is the empty set,  $E^o_{\emptyset} \cong \Pro^{N_G-1} \backslash X_G$. Therefore the $E_1$ terms of the 
spectral sequence $(\ref{eqn:  ExcepSS})$ only involve the cohomology of graph hypersurface complements in projective space of quotients of  \motic subgraphs of $G$.
In particular,
$$E_1^{p,q} = 0$$
if $q\geq N_G$ or if $p\geq h_{G}+1$. If $G$  has no masses or momenta, 
 $E_1^{p,q} =0$ for $p\geq  h_G$. 
\end{thm} 

\begin{proof} By theorem  \ref{prop: PBstructure} and the fact that the union of two \motic subgraphs is \motic\!\!,  two irreducible components $E_{\gamma_1}, E_{\gamma_2}$ of $E$ meet if and only if $\gamma_1, \gamma_2$ are nested.
Iterating, we see that every $E_J$ corresponds to a nested sequence   of \motic graphs $(\ref{eqn: Jasnested})$, and furthermore, by applying  theorem  $\ref{thm: recursive}$, that 
$$E_J \cong ( P^{\gamma'_1} \backslash Y_{\gamma'_1})  \times \ldots \times ( P^{\gamma_r'} \backslash Y_{\gamma'_r}) \ .$$
Now every divisor $E_j$ with $j \notin J$ which meets $E_J$ corresponds to a \motic subgraph $\gamma\subset G$
such that  $ \gamma_{i-1} \subsetneq  \gamma \subsetneq \gamma_{i}$ for some $i$. The latter  are in one-to-one correspondence with the \motic
subgraphs of  $\gamma'_i = \gamma_i/\gamma_{i-1}$ by theorem $\ref{thm: moticproperties}$. Therefore $E^0_J$ is obtained from $E_J$ by removing all the exceptional divisors in each factor, which gives $(\ref{eqn: EoJproduct}).$

That $E_1^{p,q}=0$ for $q\geq E_G$  is a consequence of the fact that  $\Pro^{N_{\gamma'_i} -1} \backslash X_{\gamma'_i}$
is affine of dimension $N_{\gamma'_i}-1$,  since $X_{\gamma'_i}$ is a non-empty hypersurface by lemmas \ref{lem: PsiGvanishing} and \ref{lem:  Xivanishing},
and hence $H_{B/dR}^r(\Pro^{N_{\gamma'_i} -1} \backslash X_{\gamma'_i})=0$ for $r\geq N_{\gamma'_i}$. 
Finally $E_1^{p,q}=0$ whenever $p+1$ is strictly larger than the maximal length of any  chain $(\ref{eqn: Jasnested})$. This is $h_G+1$ if $G$ has
kinematics, and $h_G$ otherwise by  lemma $\ref{lem: coradicaldegree}$. 
\end{proof} 
The terms $(\ref{eqn: EoJproduct})$ are in one-to-one correspondence with the terms in the $r$-fold iteration of the 
reduced \motic coproduct.

The previous theorem implies that the graph motives defined here are extensions of the cohomology 
of the complements of graph hypersurfaces in projective space.

\begin{rem}
There are some variants. Firstly, if $G$ has a \motic subgraph $\gamma$ with exactly one edge $e$, then the graph hypersurface has an 
irreducible component $V(\alpha_e)$.  If we remove all such components from $X_G$ then we can consider the smaller spectral sequence using  $E=\cup_{\gamma} D_{\gamma}$, where $\gamma$ are \motic subgraphs  of $G$ with $\geq 2$ edges. Finally, there is an obvious variant on replacing $X_{G}$ but $X'_{G}$. 
\end{rem}

\section{Appendix II: worked example } \label{sect: nonglobalexample} 
For the benefit of physicists who may not be accustomed to the techniques of the previous section, we give a complete worked example in a simple situation. 
The amplitudes computed here can be obtained directly, but we shall use completely general methods  without taking any shortcuts, except in the very final section. 
Consider the graph
\begin{center} 
\fcolorbox{white}{white}{
  \begin{picture}(292,105) (-30,-5)
    \SetWidth{1.0}
    \SetColor{Black}
       \Line[arrow,arrowpos=0.5,arrowlength=5,arrowwidth=2,arrowinset=0.2](58,51)(80,51)
    \Vertex(80,51){3}
    \Line(80,51)(128,83)
    \Line[double,sep=2](80,51)(128,19)
 \Line[double,sep=2](128,19)(128,83)
    \Vertex(128,19){3}
    \Vertex(128,83){3}
    \Line[arrow,arrowpos=0.5,arrowlength=5,arrowwidth=2,arrowinset=0.2](148,98)(128,83)
    \Text(100,71)[lb]{{\Black{$3$}}}
    \Text(100,25)[lb]{{\Black{$1$}}}
     \Text(52,47)[lb]{{\Black{$q$}}}
       \Text(150,95)[lb]{{\Black{$-q$}}}
    \Text(132,48)[lb]{{\Black{$2$}}}
       \Text(62,80)[lb]{{\Black{$G$}}}
  \end{picture}
}
\end{center}  
Its graph polynomial is 
$$\Xi_G = q^2 \alpha_3 (\alpha_1+ \alpha_2) + (m_1^2 \alpha_1+ m_2^2 \alpha_2) (\alpha_1+\alpha_2+\alpha_3)$$
whose zero locus defines a family of quadrics $X_{\Xi_G} \subset \Pro^3$. For generic values of $q,m_1,m_2$, this quadric meets the coordinate axes at a single point $\alpha_1=\alpha_2=0$ which corresponds  to the \motic subgraph of $G$ spanned by the edges $1,2$. Let  $P^G \rightarrow \Pro^2$ be the blow up of $\Pro^2$ at the point $D_1 \cap D_2$, i.e., $\alpha_1=\alpha_2=0$, and let $Y_G\subset P^G$ be the strict transform of $X_{\Xi_G}$ (only).  The situation is depicted below.

\begin{figure}[h!]
    \epsfxsize=10.0cm \epsfbox{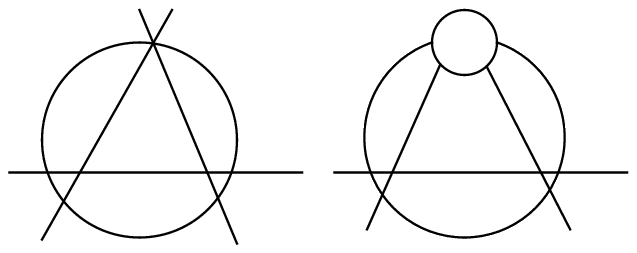}
 {\put(-305,80){{ $X_{\Xi_G}\subset \Pro^2$}}}  
  {\put(-35,80){{ $Y_G\subset P^G$}}}   {\put(-87,86){{ $D_{12}$}}}  
   {\put(-230,25){\small{ $D_{3}$}}}   {\put(-80,25){\small{ $D_{3}$}}}  
    {\put(-252,60){\small{ $D_{2}$}}}   {\put(-113,60){\small{ $D_{2}$}}}  
      {\put(-205,60){\small{ $D_{1}$}}}   {\put(-54,60){\small{ $D_{1}$}}} 
 \end{figure}
The exceptional divisor is called $D_{12}$. 
For simplicity, we  shall  only consider the graph hypersurface $X_{\Xi_G}$, and not $X_{\Psi_G}$, and compute the fibres of 
$$\mot'_G = H^2(P^G \backslash Y_G, D \backslash D \cap Y_G)$$
over $U^{\gen}_{2,2}$, where $D$ is defined in $(\ref{eqn: Ddivisordefn})$.
It satisfies $\gr^W_2 \mot'_G \cong \Q(-1)\oplus \Q(-1)$, and hence has two non-trivial periods.  
   These  were computed in an indirect manner in \cite{BrHyp}, \S5.3.2,  using the fact
 they are necessarily  logarithms of  projective invariants of seven points in $\Pro^2$. 
We give full details of the period computation using the general methods described above.  The calculations are  somewhat tedious but  are more subtle than they may   at first appear.

\subsubsection{Exceptional locus spectral sequence} Let $G$ be the graph above. It has exactly  one non-trivial \motic subgraph $\gamma$, which is  the subgraph spanned by  the
edges $1,2$. Since $D_{12}$ is isomorphic to a copy of $\Pro^1$, 
the  relevant  terms $E^{p,q}_1$  in the exceptional locus spectral sequence, for $(p,q)\in [0,1]\times [1,2]$, are therefore
\begin{eqnarray} \label{excepssinexample}
 H^2(\Pro^2 \backslash X_{\Xi_G}) & \rightarrow & H^1(\Pro^1 \backslash X_{\Xi_\gamma})(-1)   \\
 H^1(\Pro^2 \backslash X_{\Xi_G}) & \rightarrow & H^0(\Pro^1 \backslash X_{\Xi_\gamma})(-1)   \nonumber  
 \end{eqnarray}
 since $G/\gamma$ has only one edge and hence $\Pro^0 \backslash X_{\Xi_{G/\gamma}}$ is a point. Since $X_{\Xi_G}$ is 
 an odd-dimensional  quadric, $H^i(\Pro^2 \backslash X_{\Xi_G}) =0$ for $i=1,2$, and since $\Xi_{\gamma} = (q^2 +m_1^2)\alpha_1+ (q^2+m_2^2)\alpha_2$, 
 we have $\Pro^1 \backslash X_{\Xi_\gamma} = \A^1$. Therefore the above $E_1^{p,q}$ terms in $(\ref{excepssinexample})$ are
 \begin{eqnarray}  
 0  & & 0   \nonumber  \\
 0 & & \Q(-1)    \nonumber 
 \end{eqnarray}
 and hence $H^2(P^G \backslash Y_{G}) = \Q(-1)$.  
  Now consider the relative cohomology spectral sequence $E^{p,q}_1 = \bigoplus_{|I|=p} H^q (D_I \backslash D_I \cap Y_G)$ converging
 to $\mot'_G$. The terms $E^{1,1}_1$ all vanish except for
 $$ H^1(D_3 \backslash Y_{\Xi_{G/3}}) =H^1(\GG_m) \cong  \Q(-1)$$
 since the other faces $D_\bullet  \backslash (D_{\bullet} \cap Y_G) $, for $\bullet \in \{1,2,12\}$, are copies of $\A^1$.
 The face $D_3$ is given by $\alpha_3=0$ and $\Xi_{G/3} = (m_1^2 \alpha_1 + m_2^2 \alpha_2)(\alpha_1+ \alpha_2)$. 
  Since $E^{2,0}_1$ is in weight zero, and by the previous computations $E^{0,2}_1 = H^2(P^G\backslash Y_G) = \Q(-1)$, and $E^{1,2}_1=0$, we deduce that   $\gr^W_2 \mot'_G = \Q(-1)^{\oplus 2}$. 
 We shall compute the period of a class $[\widetilde{\omega}] \in (\mot'_G)_{dR}$ which maps to a generator
$[\omega ] \in H_{dR}^2 (P^G \backslash Y_G) \cong \Q(-1)$. Note that  its image is zero in $H_{dR}^2(\Pro^2 \backslash X_{\Xi_G})=0$.  
 The other period comes from the face $D_3$ via a face map, so is a period of the quotient graph $G/3$.

\subsubsection{Affine covering}  \label{sectAffinecovering} 
The prescription of \S$\ref{sect: LinearBlowups}$ defines the following affine spaces corresponding to maximal flags of subgraphs of $G$, 
where $B = \{ \emptyset, \{1,2\}, \{1,2,3\}\}$, 
\begin{eqnarray} 
\A_{12,1} & = & \A^{\FF,c}   \quad \hbox{ where } \quad (\FF,c) = ( \emptyset \subset \{1,2\} \subset \{1,2,3\} ,  j_1= 2, j_2=3 ) \nonumber \\
\A_{12,2}&  =& \A^{\FF,c} \quad \hbox{ where } \quad (\FF,c) = ( \emptyset \subset \{1,2\} \subset \{1,2,3\} ,  j_1=1, j_2=3 ) \nonumber \\
\A_{1} &=& \A^{\FF,c} \quad \hbox{ where } \quad (\FF,c) = ( \emptyset  \subset \{1,2,3\} ,  j_1=2) \nonumber \\
\A_{2}& =& \A^{\FF,c} \quad \hbox{ where } \quad (\FF,c) = ( \emptyset  \subset \{1,2,3\} ,  j_1=1) \ . \nonumber 
\end{eqnarray} 
Let $(\alpha_1:\alpha_2:\alpha_3)$ be projective coordinates on $\Pro^2$. The affine rings of the above spaces are  $\Or(\A_{12,1}) = \Z[ \beta^{12,1}_1 , \beta^{12,1}_2]$, where, by abuse of notation,
$$\beta^{12,1}_1 = {\alpha_1 \over \alpha_2} \quad , \quad \beta_2^{12,1} = \alpha_2\  , $$
 and  similarly $\Or(\A_{12,2}) = \Z[ \beta^{12,2}_1 , \beta^{12,2}_2]$, with
$$\beta^{12,2}_1 = \alpha_1 \quad , \quad \beta_2^{12,2} = {\alpha_2 \over \alpha_1}\ , $$
and $\alpha_3=1$ in both cases. 
Denote the coordinate rings  of $\A_1$ and $\A_2$  by $\Z[\alpha_1,\alpha_3]$ ($\alpha_2=1$) and   $\Z[\alpha_2,\alpha_3]$ ($\alpha_1=1$) respectively. The charts $\A_{\bullet}$ provide a canonical  affine covering of $P^G$.   The exceptional divisor $D_{12}$ is given by the equations 
$\beta^{12,1}_2=0$ and $\beta^{12,2}_1=0$ in the charts $\A_{12,1}$ and $\A_{12,2}$ respectively.

 Let $q^2, m_1^2, m_2^2$ satisfy the genericity conditions $(\ref{eqn: genericmassmomenta})$, namely
 $$ q^2+m_1^2 \neq 0 \quad , \quad q^2+m_2^2 \neq 0 \quad , \quad m_1^2 \neq 0\quad , \quad m_2^2 \neq 0 \ , $$
 and for a fixed choice of such $q,m_1,m_2$, let us write 
 \begin{eqnarray} \label{eqn:  QandQbar}
 Q  &=& q^2 \alpha_3( \alpha_1+ \alpha_2) + (m_1^2 \alpha_1 + m_2^2 \alpha_2)(\alpha_1+\alpha_2+\alpha_3)\ ,   \\
 \overline{Q} &=&  q^2 ( \alpha_1+ \alpha_2) + (m_1^2 \alpha_1 + m_2^2 \alpha_2)(\alpha_1+\alpha_2+1 )  \ . \nonumber
 \end{eqnarray} 
  Let us denote by $U_{\bullet}\subset \A_{\bullet}$  the open subsets obtained by removing the strict transform of $V(Q)$. Thus
  $\Or(U_{12,i}) =\Z[\beta^{12,i}_1, \beta^{12,i}_2, Q_i^{-1}]$ for $i=1,2$,  where 
  $$Q_1= q^2( \beta^{12,1}_1 +1 ) + (m_1^2 \beta^{12,1}_1 +m_2^2)(\beta^{12,1}_1 \beta^{12,1}_2 + \beta^{12,1}_2 + 1)  $$
  $$Q_2= q^2( \beta^{12,2}_2 +1 ) + (m_1^2+ m_2^2 \beta^{12,2}_2 )(\beta^{12,2}_1 \beta^{12,2}_2 + \beta^{12,2}_1 + 1)  $$
   Likewise
 $\Or(U_1) = \Z[\alpha_1,\alpha_3, Q|_{\alpha_2=1}^{-1}]$ and   $\Or(U_2) = \Z[\alpha_2,\alpha_3, Q|_{\alpha_1=1}^{-1}]$.
  
 \subsubsection{Absolute \v{C}ech-de Rham class} The sets $U_1,U_2, U_{12,1}, U_{12,2}$ form  our 
canonical  open affine covering of $P^G \backslash Y_G$.  Define four closed differential forms:
$$ \omega_{ab}= d\log \overline{Q}\big|_{U_{12,a} \cap U_b} \quad  \in \quad  \Omega^1 (U_{12,a} \cap U_b) \quad \hbox{ where } a, b\in \{1,2\}\ ,$$
where $\overline{Q}$ is defined in $(\ref{eqn:  QandQbar})$.
Consider the element  $\{\omega\}$ of degree $2$ in the  total complex of the absolute \v{C}ech-de Rham double complex\footnote{this is the triple complex $(\ref{eqn: reltripcomplex})$ in the special case when the divisor $D$ is empty.}  $\Omega^{n}(U_P)$ whose only non-zero 
components are the $\omega_{ab}$.  The element $\{\omega\}$ is   closed for the total differential and so defines a class
$$[\omega] \in H_{dR}^2 (P^G \backslash Y_G)\ .$$

\subsubsection{Relative \v{C}ech-de Rham class}
We next wish to extend $\{\omega\}$ to a closed element in the relative \v{C}ech-de Rham triple complex 
$(\ref{eqn: reltripcomplex})$ where $D$ has four irreducible components $D_{12}, D_1, D_2, D_3$.  
This will necessarily  use the fact that $D_i \backslash (D_i \cap Y_G)$ is isomorphic to $\A^1$, for $i=1,2$, and hence has vanishing $H^1$. 

Define four closed 1-forms
$$ \mu_i   \in \Omega^1(D_i \cap U_{12,i}) \quad , \quad \nu_i \in \Omega^1(D_i \cap U_i) $$
where $i\in \{1,2\}$, by 
$$\mu_1 =  {2 m_2^2\, d \beta^{12,1}_2 \over q^2 + m_2^2 (\beta^{12,1}_2+1)} \quad \hbox{ and } \quad \nu_1 =  { (q^2+ m_2^2) d\alpha_3 \over q^2 \alpha_3 + m_2^2 (1+\alpha_3)} $$
and $\mu_2, \nu_2$ are defined by interchanging all subscripts `1' and  `2' (and `${12,1}$' with `$12,2$').
Note that $D_i$ does not meet  $U_{12,j}$ or $U_j$ if $i\neq j$. Consider the element  $\{\widetilde{\omega}\}$ of degree 2
in the total  complex associated to  $(\ref{eqn: reltripcomplex})$ which is zero in every component  of $\Omega^n(U_P \cap D_Q)$ except for 
the eight elements 
$\omega_{ab} \in \Omega^1(U_{12,a} \cap U_b)$ for $a,b\in \{1,2\}$ as above and $\mu_i, \nu_i$, for $i=1,2$.  The element 
$\{\widetilde{\omega}\}$ is closed for the total differential essentially because of the equations 
$$\omega_{ii} |_{D_i}  = \mu_i - \nu_i \quad \hbox{ for } i =1,2\ ,$$
on the open $U_{12,i} \cap U_i\cap D_i \subset D_i$ (recall that the left-hand side is $d\log Q|_{\alpha_i=0}$). 
This defines a class  in relative de Rham cohomology
$$[\widetilde{\omega}] \in H_{dR}^2 ( P^G \backslash Y_G, D \backslash (D \cap Y_G)$$
whose image in absolute de Rham cohomology is $[\omega] \in H_{dR}^2 ( P^G \backslash Y_G)$. 

\subsubsection{The period} We can compute the period 
$$I= \int_{\sigma_G} \{\widetilde{\omega}\}$$
following the prescription of \S\ref{sect: sectorsblowups}. Let $t>0$ and let $z_t\in \Pro^2(\R)$ 
denote the point with projective coordinates $(t:t:1)$.   Only six regions $\sigma^P_Q(z)$ of the domain of integration 
provide a non-zero contribution to the period integral, shown below.

\begin{figure}[h!]
    \epsfxsize=12.0cm \epsfbox{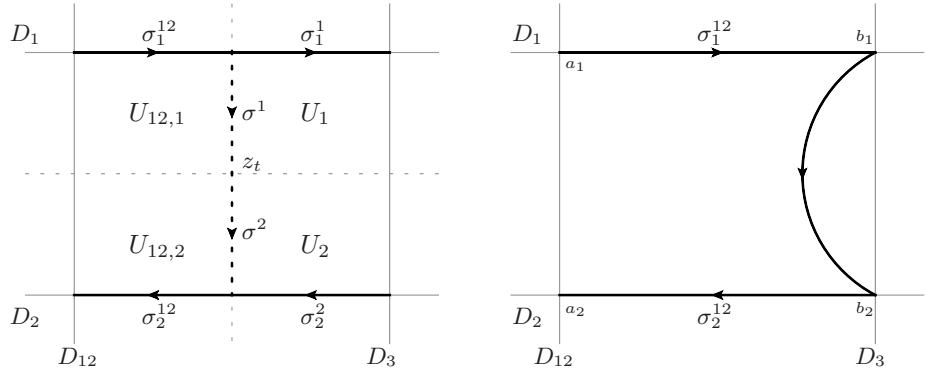}
 {\put(-350,115){{ \small $D_{1}$}}}   {\put(-300,115){{ \small $\sigma^{12}_{1}$}}}   {\put(-240,115){{ \small $\sigma^1_{1}$}}}  
  {\put(-305,85){{ $U_{12,1}$}}}   {\put(-240,85){{ $U_1$}}}  
 {\put(-335,-7){{  \small{ $D_{12}$}}}}  {\put(-215,-7){\small{ $D_{3}$}}}     
 {\put(-305,35){{ $U_{12,2}$}}}   {\put(-240,35){{ $U_2$}}} 
 {\put(-350,8){{ \small $D_{2}$}}}    {\put(-300,8){{ \small $\sigma^{12}_{2}$}}}   {\put(-240,8){{ \small $\sigma^2_{2}$}}}                              
     {\put(-262,85){\small { $\sigma^{1}$}}}  
      {\put(-262,67){\small { $z_t$}}}  
       {\put(-262,40){\small { $\sigma^{2}$}}}  
  {\put(-160,115){{ \small $D_{1}$}}}   {\put(-140,105){{ \tiny $a_{1}$}}}  {\put(-140,12){{ \tiny $a_{2}$}}}  
   {\put(-90,115){{ \small $\sigma^{12}_{1}$}}}  {\put(-30,115){{ \tiny $b_{1}$}}}{\put(-30,12){{ \tiny $b_{2}$}}}
  {\put(-155,-7){{  \small{ $D_{12}$}}}}   {\put(-90,8){{ \small $\sigma^{12}_{2}$}}} 
   {\put(-160,8){{ \small $D_{2}$}}}   
   {\put(-30,-7){\small{ $D_{3}$}}}     
 \caption{Integration of a  period.  Left: the non-trivial contributions to the period integral. Right: 
 taking the limit as $t \rightarrow \infty$ and replacing the paths with tangential base points.}
 \end{figure} 
 Here,  $\sigma^1 = \{(x:t:1), 0\leq x \leq t\}$ and $\sigma^2 = \{(t: x:1),  0 \leq x \leq t\}$.
  
From the general formula $(\ref{eqn: intpairingrelative})$ we have
$$I=  \sum_{i=1,2} \int_{\sigma^{12}_i(t)} \mu_i + \int_{\sigma^i_i(t)} \nu_i + \int_{\sigma^i(t)} \omega_{ii}\ .$$
It can  be computed directly, and does not depend on $t$. 
Instead, we shall compute it by letting $t\rightarrow \infty$ and using tangential base points. 
\subsubsection{Limits and tangential basepoints}
Consider the following basepoints on $D_1$: an ordinary base point  $a_1= D_1 \cap D_{12}$ (given by $\beta^{12,1}_2=0$), and 
a tangential base point  $b_1$ on $D_1$ at the point $D_1 \cap D_{3}$ defined by $-{\partial  /  \partial \beta^{12,1}_2}$. 
Denote the analogous basepoints on $D_2$ by $a_2, b_2$, as shown in the figure.
Since our forms have at most logarithmic poles, we deduce on taking $t\rightarrow \infty$, that 
$$I  = \int_{a_1}^{b_1} \mu_1 + \int_{b_2}^{a_2} \mu_2 + \int_{b_1}^{b_2} d \log \overline{Q}  $$
as shown in the previous figure. The third term is given by 
$$\lim_{t \rightarrow \infty} \log \Big( {\overline{Q}(t,0) \over \overline{Q}(0,t)}\Big) = \lim_{t \rightarrow \infty} \log \Big( {q^2 +m_1^2 (t+1) \over q^2 + m_2^2 (t+1)}\Big)  = \log  \Big( {m_1^2 \over m_2^2}\Big)\ . $$
The first term is given by

$$ \lim_{t\rightarrow -1_{\infty}} \int_{0}^{t} {2 m_2^2\, d x \over q^2 + m_2^2 (x+1)}  = 2  \lim_{t\rightarrow -1_{\infty}}  \log \Big({q^2 +m_2^2 (t+1)  \over 
q^2+ m_2^2} \Big) = 2 \log \Big({m_2^2  \over 
q^2+ m_2^2} \Big) $$
where $ \lim_{t\rightarrow -1_{\infty}}$ denotes the regularised limit as $t\rightarrow \infty$ with respect to the tangential base point $-{\partial/\partial t}$ (set
 $\log t$  to zero). The second term is obtained by  replacing $m_2$ by $m_1$ in this formula and changing the sign. 
In total, 
$$ I = 2 \log \Big( {q^2+ m_1^2 \over q^2 + m_2^2} \Big) - \log \Big( {m_1^2 \over m_2^2}\Big) \ .  $$
We see that $I$ is a linear combination  of limits of divergent generalised amplitudes.

On the other hand, the amplitude of the one-loop graph $G/3$   is proportional
to $\log m_1^2  m_2^{-2}$. By the face relations, this provides another  period of $G$.   From the cohomology calculations, we know that the motivic periods of $\mot'_G$ are motivic logarithms. Since these are uniquely determined by their period, we deduce the:

\begin{cor} The motivic periods  of $\mot'_G$ are spanned by $1$ and the two motivic logarithms (on the space of generic kinematics $U^{\gen}_{2,2}$)
$$\log^{\mm} \Big( {m_2^2 \over m_1^2} \Big) \qquad \hbox{ and } \qquad \log^{\mm} \Big( {q^2 + m_2^2 \over q^2 +m_1^2}\Big)\ .$$
\end{cor}
Thus the periods of $G$ are regularised limits of  linear combinations of  amplitudes of the three graphs $G/1$, $G/2$, $G/3$.
\noindent See \cite{BrHyp} \S5.3.2 for an interpretation of these periods in terms of  hyperbolic geometry.
  This example illustrates how the motivic periods of small graphs can in principle be computed algorithmically.


\begin{thebibliography}{999}


\bibitem{CutsCoproducts} {\bf S. Abreu, R. Britto, H.  Gr\"onqvist}:  {\it  Cuts and coproducts of massive triangle diagrams}, arXiv:1504.00206 

\bibitem{CutsCoproducts2} {\bf S. Abreu, R. Britto, C. Duhr, E. Gardi}: {\it From multiple unitarity cuts to the coproduct of Feynman integrals}, arXiv:1401.3546

\bibitem{BWA} {\bf L. Adams, C. Bogner,  S. Weinzierl} {\it The two-loop sunrise graph in two space-time dimensions with arbitrary masses in terms of elliptic dilogarithms},  arXiv:1405.5640


\bibitem{AM} {\bf P. Aluffi, M. Marcolli}: {\it Feynman motives and deletion-contraction relations}, \url{arXiv:0907.3225v1}, (2009).

\bibitem{An1} {\bf Y. Andr\'e}: {\it Galois theory, motives, and transcendental number theory}, arXiv:0805.2569.

\bibitem{An2} {\bf Y. Andr\'e}: {\it Ambiguity theory, old and new},  arXiv:0805.2568

\bibitem{BB} {\bf P. Belkale, P. Brosnan}: {\it Matroids, motives, and a conjecture of Kontsevich},
Duke Math. J. 116, No.1, 147-188 (2003).

\bibitem{MB} {\bf M. Berghoff}: {\it Wonderful renormalisation},  Ph.D. thesis, Humboldt University (2015). 

\bibitem{Bloch} {\bf S. Bloch}: {\it Motives associated to graphs}, Japan J. Math. 2 (2007), 165Ð196


\bibitem{BEK} {\bf S. Bloch, H. Esnault, D. Kreimer}:  {\it On motives associated to graph polynomials},
Comm. Math. Phys. 267 (2006), no. 1, 181-225.



\bibitem{BlochKreimer} {\bf S. Bloch, D. Kreimer}: {\it  Feynman amplitudes and Landau singularities for 1-loop graphs}, 
Communications in Number Theory and Physics 07/2010; 4(4),  DOI:10.4310/CNTP.2010.v4.n4.a4

\bibitem{BroadhurstModular} {\bf D. Broadhurst}: {\it Multiple zeta values and modular forms in quantum field theory}, in Computer algebra
in Quantum Field theory (2013), pp 33-73. 

\bibitem{BK}{\bf D. Broadhurst, D.  Kreimer}: {\it  Knots and numbers in $\phi^4$ theory to 7 loops
and beyond}, Int. J. Mod. Phys. C 6, 519 (1995).

\bibitem{Bloch-Vanhove-Kerr} {\bf  S. Bloch, M. Kerr, P. Vanhove}: {\it A Feynman integral via higher normal functions}, arXiv:1406.2664

\bibitem{BrMTZ} {\bf F. Brown}: {\it  Mixed Tate motives over $\Z$}, Annals of Math., volume 175, no. 1 (2012). 



\bibitem{IHESyoutube}  \url{https://www.youtube.com/watch?v=PlJIECqRZRA} 

\bibitem{NotesMot} {\bf F. Brown}: {\it Notes on Motivic Periods}, \url{arXiv:1512.06410}. 


\bibitem{BrFeyn}{\bf F.  Brown}:  {\it On the periods of some Feynman graphs}, \url{arXiv:0910.0114v1}, (2009).

\bibitem{BrHyp} {\bf F. Brown}: {\it Dedekind zeta motives for totally real number fields},  Inventiones,  Volume 194, Issue 2 (2013), Page 257-311,  arXiv:0804.1654


\bibitem{Framings}{\bf F. Brown, D. Doryn}: {\it Framings for graph hypersurfaces}, arXiv:1301.3056

\bibitem{Angles}{ \bf F. Brown, D. Kreimer}: {\it Angles, scales, and parametric renormalization}, {Letters in Mathematical Physics}, Volume 103, Issue 9 (2013), Page 933-1007,   arXiv:1112.1180

\bibitem{Modularphi4} {\bf F. Brown, O. Schnetz}: {\it Modular forms in quantum field theory},  Commun. Number Theory Phys.
7(2013) 293-325,  arXiv:1304.5342

\bibitem{BSY} {\bf F. Brown, O.  Schnetz, K. Yeats}: \emph{ Properties of $c_2$ invariants of Feynman graphs}, arXiv:1203.0188 (2012).




\bibitem{Folle} {\bf P. Cartier} :   {\it La folle journ\'ee, de Grothendieck \`a Connes et Kontsevich}, Evolution des notions d'espace et de sym\'etrie', Publications IHES (1998), \url{archive.numdam.org/article/PMIHES_1998__S88__23_0.pdf}


	
\bibitem{Bootstrap}:	{\bf  S. Caron-Huot, L.  Dixon, A. McLeod, M. von Hippel}: {\it Bootstrapping a Five-Loop Amplitude Using Steinmann Relations}, Physical Review Letters, Volume 117, Issue 24, id.241601



\bibitem{CM1} {\bf A. Connes, M. Marcolli}: {\it Renormalisation and motivic Galois theory}, \url{http://arxiv.org/abs/math/0409306} (2004)

\bibitem{CM2} {\bf A. Connes, M. Marcolli}: {\it Quantum fields and motives}: \url{http://arxiv.org/abs/hep-th/0504085}

\bibitem{Davidychev} {\bf A. I. Davydychev, R. Delbourgo}, {\it A geometrical  angle on  Feynman integrals}, J. Math. Phys. 39 (1998) 4299, \url{arXiv:hep-th/9709216}

\bibitem{Do} {\bf D. Doryn}, \emph{On the cohomology of graph hypersurfaces associated to certain Feynman graphs}, Comm. Num. Th. Phys. 4 (2010), 365-415.

\bibitem{DorynPoints} {\bf D. Doryn}: {\it On one example and one counterexample in counting rational points on graph hypersurfaces}, 
 \url{arXiv:1006.3533v1} (2010).

\bibitem{Dupont}  {\bf C. Dupont}: {\it Relative cohomology of bi-arrangements}, arXiv:1410.6348
\bibitem{Greenetc} {\bf E.  D'Hoker, M.  Green, O. Gurdogan, P. Vanhove } : {\it Modular Graph Functions},   arXiv:1512.06779


\bibitem{Cluster} {\bf   J. Golden, A. Goncharov, M. Spradlin, C. Vergu, A. Volovich} :{\it Motivic Amplitudes and Cluster Coordinates}, \url{http://arxiv.org/abs/1305.1617}

\bibitem{Go-Ma} {\bf A. Goncharov, Yu.  Manin}: {Multiple $\zeta$-motives and moduli spaces $M_{0,n}$}, Compos. Math., 140(1):1-14, (2004)


\bibitem{Groth} {\bf A. Grothendieck} : {\it 
On the de Rham cohomology of algebraic varieties},
Publications mathŽmatiques de lÕI.H.E.S., tome 29 (1966), p. 95-103

\bibitem{FeynCat} {\bf R.  Kaufmann, B.  Ward} : {\it Feynman categories}, arXiv:1312.126



\bibitem{Kir}{\bf G. Kirchhoff}: {\it Ueber die Aufl\"osung der Gleichungen, auf welche man bei der Untersuchung der linearen Vertheilung galvanischer Str\"ome geh\"uhrt wird}, Annalen der Physik und Chemie 72 no.\ 12 (1847), 497-508.


\bibitem{KoZa}{\bf M. Kontsevich,   D. Zagier}: \emph{Periods}, Mathematics unlimited-2001 and beyond, 771-808, Springer, Berlin, (2001).



\bibitem{Corolla}{ \bf   D. Kreimer, M. Sars, W. D. van Suijlekom}: {\it Quantization of gauge fields, graph polynomials and graph homology}, Annals of Physics 08/2012; 336, DOI:10.1016/j.aop.2013.04.019



 \bibitem{MWZ} {\bf S. M\"uller-Stach, S. Weinzierl, R. Zayadeh}: {\it 
 A second-order differential equation for the two-loop sunrise graph with arbitrary masses},   arXiv:1112.4360 
    
\bibitem{Nickel} {\bf B.G. Nickel}: {\it Evaluation of Simple Feynman Graphs},  J.Math.Phys.19:542-548,1978


\bibitem{PanzerPhd} {\bf   E. Panzer} : {\it Feynman integrals and hyperlogarithms}, Ph.D thesis, arXiv:1506.07243

\bibitem{Panzer2} {\bf  E. Panzer}: {\it Algorithms for the symbolic integration of hyperlogarithms with applications to Feynman integrals}, 
 	Computer Physics Communications, 188 (2015), pp. 148-166 
	

\bibitem{PanzerSchnetz} {\bf E. Panzer, O. Schnetz}: {\it The Galois action on $\phi^4$ periods}, preprint.



\bibitem{SS} {\bf   O. Schlotterer, S. Stieberger}: {\it  Motivic Multiple Zeta Values and Superstring Amplitudes}, arXiv:1205.1516




\bibitem{Census} {\bf O. Schnetz}: {\it Quantum periods: A census of $\phi^4$ transcendentals}, Jour.\ Numb.\ Theory
and Phys.\ 4 no.\ 1 (2010), 1-48.

\bibitem{Schnetz1loop} {\bf O. Schnetz} : {\it The geometry of one-loop amplitudes}, arXiv:1010.5334

\bibitem{Sm} {\bf V. Smirnov}: {\it  Evaluating Feynman
integrals}, Springer Tracts in Modern Physics 211. Berlin: Springer.
ix, 247 p. (2004). [ISBN 3-540-23933-2]


\bibitem{Speer} {\bf E. Speer}: {\it Ultraviolet and infrared singularity structure of generic Feynman amplitudes},
Ann. Inst. H. PoincarŽ Sect. A (N.S.) 23 (1975), no. 1, 1-21. 

\bibitem{SpeerWestwater} {\bf E. Speer, M. Westwater}: {\it Generic Feynman amplitudes}, Ann. Inst. H. PoincarŽ Sect. A (N.S.) 14 (1971), 1-55.


\bibitem{ClosedString} {\bf S. Stieberger, T.  Taylor}: {\it
Closed string amplitudes as single-valued open string amplitudes}, Nuclear Phys. B 881 (2014), 269-287. 


\bibitem{Wein} {\bf S. Weinberg}: {\it High-Energy Behavior in Quantum Field Theory},  Phys.\ Rev.\ 118, no.\ 3 838--849 (1960).


\bibitem{Zerbini} {\bf F. Zerbini}: {\it Single-valued multiple zeta values in genus 1 superstring amplitudes},  arXiv:1512.05689









\end{thebibliography}
\end{document}